\documentclass[article, onecolumn, a4paper, 11pt]{IEEEtran}

\usepackage{psfrag}
\usepackage{graphicx}
\usepackage{fancyhdr}

\usepackage{amsmath,amsthm}
\usepackage{amssymb}
\usepackage{color}
\usepackage{graphics}
\newcommand{\mw}[1]{#1}

\newcommand{\trans}[1]{#1^{\textnormal{\textsf{\tiny T}}}} 
\newcommand{\Asym}{\textnormal{Asym}}
\newcommand{\Sym}{\textnormal{Sym}}
\newcommand{\K}{\{1,\ldots, K\}}
\renewcommand {\epsilon} {\varepsilon}

\newcommand{\abs}[1]{\left\lvert#1\right\rvert}

\newcommand{\vect}[1]{\boldsymbol{#1}}

\newcommand{\Ca}{\mathcal{C}}
\renewcommand{\eta}{\mathcal{S}}
\newcommand{\bfY}{\vect{Y}}
\newcommand{\bfN}{\vect{N}}
\newcommand{\bfV}{\vect{V}}
\newcommand{\bfX}{\vect{X}}

\newcommand{\mat}[1]{\mathsf{#1}}
\newcommand{\deta}[1]{\det\left(#1\right)}

\allowdisplaybreaks[4]

\def\N{{\mathbb N}}        
\def\R{{\mathbb R}}        

\def\1{{\mathbf 1}}        

\newtheorem{theorem}{Theorem}
\newtheorem{lemma}[theorem]{Lemma}
\newtheorem{corollary}[theorem]{Corollary}

\newtheorem{proposition}[theorem]{Proposition}
\newtheorem{remark}{Remark}
\newtheorem{definition}{Definition}
\newtheorem{example}{Example}
\newtheorem{conclusion}{Conclusion}

\theoremstyle{remark}

\newcommand{\bvect}[1]{#1^n}

\allowdisplaybreaks[3]

\begin{document}

\title{Cognitive Wyner Networks with Clustered  Decoding}

\author{\authorblockN{Amos Lapidoth, Nathan Levy, Shlomo Shamai (Shitz), and Mich\`ele~Wigger}\thanks{This work was in part supported by the European Commission in the framework of
the FP7 Network of Excellence in Wireless COMmunications NEWCOM++
and NEWCOM\#. The paper was in part presented at the \emph{International Symposium on Information Theory 2009}, in Seoul, Korea, July 2009. 

A.~Lapidoth is with the Department of Information Technology and Electrical Engineering at ETH Zurich, Switzerland \{email: lapidoth@isi.ee.ethz.ch\}. 
N.~Levy was with the Department of Electrical Engineering at the Technion---Israel Institute of Technology \{email: nhlevy@gmail.com\}.
S.~Shamai is with the Department of Electrical Engineering at the Technion---Israel Institute of Technology \{email: sshlomo@ee.technion.ac.il\}.
M.~Wigger was with the Department of Information Technology and Electrical Engineering at ETH Zurich, Switzerland. She is now with the Communications and Electronics Department, at
Telecom ParisTech, France \{email: michele.wigger@telecom-paristech.fr\}.}}
\date{\today}


\maketitle

\begin{abstract}
  We study an interference network where equally-numbered transmitters
  and receivers lie on two parallel lines, each transmitter opposite
  its intended receiver. We consider two short-range interference
  models: the ``asymmetric network,'' where the signal sent by each
  transmitter is interfered only by the signal sent by its left
  neighbor (if present), and a ``symmetric network,'' where it is
  interefered by both its left and its right neighbors.  Each
  transmitter is cognizant of its own message, the messages of the
  $t_\ell$ transmitters to its left, and the messages of the $t_r$
  transmitters to its right. Each receiver decodes its message based on
  the signals received at its own antenna, at the $r_\ell$ receive
  antennas to its left, and the $r_r$ receive antennas to its
  right. 

  For such networks we provide upper and lower bounds on the
  multiplexing gain, i.e., on the high-SNR asymptotic logarithmic
  growth of the sum-rate capacity. In some cases our bounds meet,
  e.g., for the asymmetric network.
    
  Our results exhibit an equivalence between the transmitter
  side-information parameters $t_\ell, t_r$ and the receiver
  side-information parameters $r_\ell, r_r$ \mw{in the sense that increasing/decreasing $t_\ell$ or $t_r$ by a positive integer  $\delta$ has the same effect on the prelog as increasing/decreasing $r_\ell$ or $r_r$ by $\delta$}.  Moreover---even in
  asymmetric networks---there is an equivalence between the left
  side-information parameters $t_\ell, r_\ell$ and the right
  side-information parameters $t_r, r_r$.

\end{abstract}


\section{Introduction}
\label{sec:Introduction}


We consider a large cellular mobile communication system where $K$ cells are positioned on a linear array. 
We assume short-range interference where the signal sent by the mobiles in a cell interfere only with the signals sent in the left adjacent cell and/or the right adjacent cell, depending on the position of the mobile within the cell. Similarly, the signal sent by a base station interferes only with the signals sent by the base stations in the adjacent cell(s). Our goal is to determine the throughput of such a cellular system at high signal-to-noise ratio (SNR). 

The high-SNR throughput of our system (where we assume constant
non-fading channel gains) does not depend on the number of mobiles in
a cell (provided this number is not zero), because in each cell there is only one base station. Therefore, we restrict attention to setups with only one mobile per cell. 

We particularly focus on two regular setups. The first setup exhibits \emph{asymmetric interference}: the communication in a cell is only interfered by the signals sent in the cell to its left  but not by the signals sent in the cell to its right (e.g., because all the mobiles lie close to the right border of their cells). The second setup exhibits \emph{symmetric interference}: the communication in a cell is interfered by the signals sent in the cells to its left and to its right (e.g., because the mobiles  lie in the center of their cells). 
The symmetric setup was introduced in
\cite{Wyner-94,Hanly-Whiting-93}. 

On a more abstract level, our communication scenario is described as follows: $K$ transmitters wish to communicate independent messages to  their $K$ corresponding receivers, and it is assumed that  these communications interfere. Moreover, the $K$ transmitters are assumed to be located on a horizontal line, and the $K$ receivers are assumed to lie on a parallel line, each receiver opposite its corresponding transmitter. We consider two specific networks. In the \emph{asymmetric network}, each receiver observes a linear combination of the signals transmitted by its corresponding transmitter,  the signal of the
transmitter to its left, and additive white Gaussian noise (AWGN). See Figure~\ref{fig: setting-handoff}.  In the \emph{symmetric
network}, each receiver observes a linear combination of the signal
transmitted by its corresponding transmitter, the two signals of
the transmitter to its left and the transmitter to its right, and AWGN. See
Figure~\ref{fig: setting-full}. 
The symmetric network is also known as Wyner's linear model or the full Wyner model; the asymmetric network is known as the asymmetric Wyner model or the soft hand-off model. 
\begin{figure}[htb]
\psfrag{MK1}{$M_1$}
\psfrag{MK2}{$M_2$}
\psfrag{MK3}{$M_3$}
\psfrag{MK4}{$M_4$}
\psfrag{MK5}{$M_5$}
\psfrag{MK6}{$M_6$}
\psfrag{MK7}{$M_7$}
\psfrag{XK1}{$X_1$}
\psfrag{XK2}{$X_2$}
\psfrag{XK3}{$X_3$}
\psfrag{XK4}{$X_4$}
\psfrag{XK5}{$X_5$}
\psfrag{XK6}{$X_6$}
\psfrag{XK7}{$X_7$}
\psfrag{YK1}{$Y_1$}
\psfrag{YK2}{$Y_2$}
\psfrag{YK3}{$Y_3$}
\psfrag{YK4}{$Y_4$}
\psfrag{YK5}{$Y_5$}
\psfrag{YK6}{$Y_6$}
\psfrag{YK7}{$Y_7$}
\psfrag{parameters}{$K=7$, $t_\ell=2$, $t_r=1$, $r_\ell=2$, and $r_r=1$}
\psfrag{Receiver5}{Receiver~$5$}
\psfrag{Receivingantennas}[l][l]{\footnotesize{Receive antennas (with AWGN)}}
\psfrag{Transmitters}{\footnotesize{Transmitters}}
\begin{center}
\includegraphics[scale=0.55]{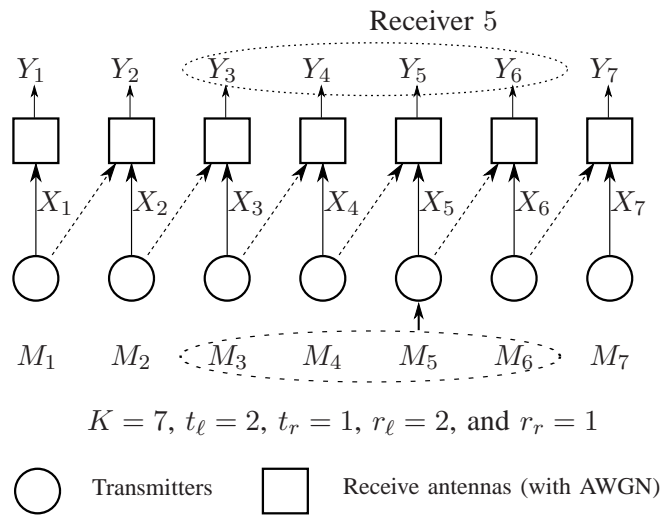}
\caption{Asymmetric network}
\label{fig: setting-handoff}
\end{center}
\end{figure}

\begin{figure}[htb]
\psfrag{MK1}{$M_1$}
\psfrag{MK2}{$M_2$}
\psfrag{MK3}{$M_3$}
\psfrag{MK4}{$M_4$}
\psfrag{MK5}{$M_5$}
\psfrag{MK6}{$M_6$}
\psfrag{MK7}{$M_7$}
\psfrag{XK1}{$X_1$}
\psfrag{XK2}{$X_2$}
\psfrag{XK3}{$X_3$}
\psfrag{XK4}{$X_4$}
\psfrag{XK5}{$X_5$}
\psfrag{XK6}{$X_6$}
\psfrag{XK7}{$X_7$}
\psfrag{YK1}{$Y_1$}
\psfrag{YK2}{$Y_2$}
\psfrag{YK3}{$Y_3$}
\psfrag{YK4}{$Y_4$}
\psfrag{YK5}{$Y_5$}
\psfrag{YK6}{$Y_6$}
\psfrag{YK7}{$Y_7$}
\psfrag{Receiver5}{Receiver~$5$}
\psfrag{Receivingantennas}[l][l]{\footnotesize{Receive antennas (with AWGN)}}
\psfrag{Transmitters}{\footnotesize{Transmitters}}
\psfrag{parameters}{$K=7$, $t_\ell=2$, $t_r=1$, $r_\ell=1$, and $r_r=2$}
\begin{center}
\includegraphics[scale=0.55]{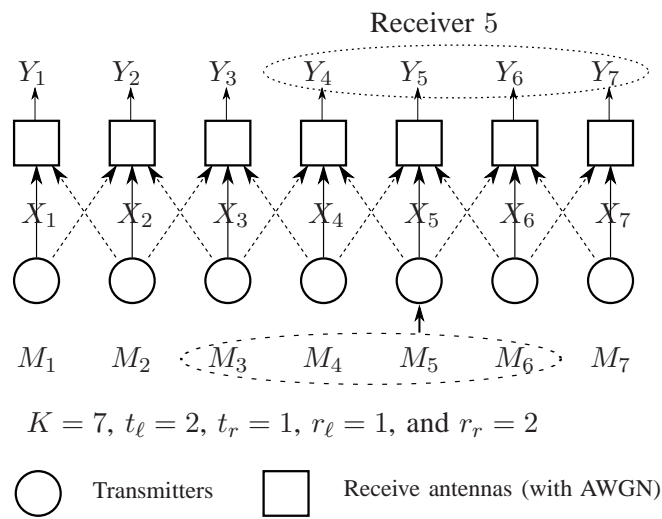}

\caption{Symmetric network}
\label{fig: setting-full}

\end{center}
\end{figure}

In \cite{Wyner-94, Hanly-Whiting-93} 
the receivers were allowed to
\emph{fully} cooperate in their decoding, and thus the communication scenario
was modeled as a multiple-access channel (MAC). In contrast, here we
assume that each receiver has to decode its message individually, and
therefore our communication scenario is modeled as an interference
network. However, we still allow for \emph{partial}
cooperation between neighboring receivers as encountered in the
  uplink of cellular mobile systems where the neighboring basestations---because they can communicate over a backhaul---can
  cooperate  in the form of
\emph{clustered local decoding}. That means, each receiver
beside its own antenna has access also to the antennas of some of the
receivers to its left and to its right. 

Similarly, we also want to
allow for (partial) cooperation between the transmitters in the form of \emph{message cognition}. That means, that each transmitter besides its own message is cognizant also of the messages of some transmitters to its
left and to its right. Such a scenario could be envisioned in the uplink of cellular mobile systems where the mobiles can communicate over bluetooth connections before communicating to their corresponding base stations.\footnote{Our setup can also model a downlink scenario where the receiving mobiles relay their observed signals to the mobiles in neighboring cells (e.g., using bluetooth connections) and the transmitting base stations use the backhaul to share their messages.}

Notice that the described model represents a combination of the cognitive model in \cite{Lapidoth-Shamai-Wigger-ISIT07} and the clustered decoding model in \cite{local-decoding}.
Also, clustered local processing is in a way a compromise between the
joint (multi-cell) decoding in \cite{Wyner-94, Hanly-Whiting-93} and the single
(single-cell) decoding in
\cite{Shamai-Wyner-97-I-II,Lapidoth-Shamai-Wigger-ISIT07}. 
Clustered decoding has also been considered in \cite{Kim}  for fully-connected interference networks. \mw{The cognitive transmitter model considered here has been refined in \cite{rate-limited}, where the transmitters can exchange parts of their messages prior to the actual communication over rate-limited pipes, similar to~\cite{willems83,maricyateskramer07,simeoneetal08, WK08, BrossLapidothWigger11}. }

Our focus in this paper is on the high-SNR asymptotes of the
sum-capacities of these networks. Formally, we present our
results in terms of the multiplexing gain or the asymptotic multiplexing gain per user; the asymptotic multiplexing gain per user is defined as the multiplexing gain of a network divided by the number of transmitter/receiver pairs $K$ in the asymptotic regime of large $K$.
We present lower and upper bounds on the multiplexing gain and the asymptotic multiplexing gain per user for the two networks. 

For the asymmetric network our upper and lower bounds coincide, and thus yield the exact multiplexing gain and asymptotic multiplexing gain per user. The results exhibit an equivalence between cooperation at the transmitters and cooperation at the receivers. Moreover, the asymptotic multiplexing gain per user also exhibits an equivalence between the transmitters' information about their right-neighbors' messages and their information about their left-neighbors' messages. Similarly, they also exhibit, an equivalence between the receivers' information about the signals observed at their right-neighbors' antennas and their information about the signals observed at their left-neighbors' antennas. This result surprises in view of the asymmetry of the network.

For the symmetric network our upper and lower bounds coincide only in some special cases. In these special cases  the multiplexing gain---and thus also the asymptotic multiplexing gain per user---again exhibits an equivalence between cooperation at the transmitters and cooperation at the receivers. \mw{For the symmetric network, we mostly assume that the nonzero cross-gains are all equal. Our techniques extend to general cross-gains, but the statement of the results becomes cumbersome and is therefore omitted. Instead, we also consider a random model where the cross-gains are drawn from a continuous distributions. Our main results continue to hold (with probability~$1$) for this randomized setup.}

\mw{For large number of users $K\gg 1$,  our multiplexing-gain results are of the form $\eta_\infty \cdot K+ o(K)$, where $\eta_\infty \in[0.5, 1]$ is strictly monotonic  in the side-information parameters $t_\ell, t_r, r_\ell, r_r$. That means, if we increase one or several of the side-information parameters, then also the factor $\eta_\infty$increases.\footnote{The parameter $\eta_\infty$ is called the  \emph{asymptotic multiplexing-gain per user} and will be introduced formally in the next section.} The results in~\cite{GamalAnnapureddy4,GamalAnnapureddy2, GamalAnnapureddy1, GamalAnnapureddy3} suggest that this strict monotonicity relies on the weak connectivity of the network, i.e., the fact that there are relatively few interference links. Indeed, \cite{GamalAnnapureddy4,GamalAnnapureddy2, GamalAnnapureddy1, GamalAnnapureddy3} show that for \emph{fully-connected networks}, i.e., when all the transmitted signals interfere at all received signals, and when there is no clustering at the receivers ($r_\ell=r_r=0$), then for the side-information pattern considered here, $\eta_\infty=1/2$, irrespective of $t_\ell$ and $t_r$. This result holds even in the stronger setup where for each message we can choose the set of $t_\ell+t_r+1$ adjacent transmitters that are cognizant of this message \cite{GamalAnnapureddy4}. In this stronger setup, a given transmitter~$k$ might not know  Message~$k$ intended to its corresponding receiver~$k$.
Notice that sometimes (for example for the networks considered here),  the multiplexing gain can indeed be increased by assigning a given Message $M_k$ to subsets of $t_\ell+t_r+1$ transmitters that does not contain the original transmitter~$k$ \cite{GamalAnnapureddy4}. We will describe this in more detail after describing our results. 

General interference networks with transmitter cooperation have also been studied in \cite{Lapidoth-Shamai-Wigger-ITW07, Gesbert, Lozano, Wangetal11}. In particular, in \cite{Lapidoth-Shamai-Wigger-ITW07}, the authors completely characterized the set of networks and transmitter side-informations that have full multiplexing gain $K$ or multiplexing gain $K-1$. In \cite{Wangetal11}, a network is presented where adding an interference link to the network---while keeping the same transmitter side-informations---can increase the multiplexing gain. 
}

\mw{The asymmetric network has also been studied by Liu and Erkip~\cite{LiuErkip}, with a focus on finite-SNR results but without transmitter cognition or clustered decoding. For general $K\geq 3$, \cite{LiuErkip} characterizes the maximum sum-rate that is achievable using a simple Han-Kobayashi scheme without time-sharing and where the inputs follow a Gaussian distribution.  For $K=3$, they show that this scheme achieves the sum-capacity in \emph{noisy-interference and mixed-interference regimes} and it achieves the entire capacity region in a \emph{strong interference regime}. Zhou and Yu \cite{ZhouYu} considered a cyclic version of this model where additionally the $K$-th transmitted signal interferes with the first receive signal, i.e., the interference pattern is cyclic. In \cite{ZhouYu}, an expression for the  Han-Kobayashi region with arbitrary (also non-Gaussian) inputs is presented. It is shown that this region achieves within 2 bits of the $K$-user cyclic asymmetric network in the \emph{weak-interference regime}.  In the strong interference regime, it achieves capacity. (In their achievability proofs it suffices to consider  Gaussian inputs.) For $K=3$ the authors also present an improved Han-Kobayashi scheme involving time-sharing that achieves rates within 1.5 bits of capacity. Finally,  \cite{ZhouYu} also characterizesthe  generalized degrees of freedom (GDoF) of the symmetric capacity assuming that all cross-gains in the network are equal. Interestingly, this result shows that the GDoF of the $K$-user cyclic asymmetric Wyner network with equal cross-gains has the same GDoF as the standard two-user interference channel \cite{EtkinTseWang08}.}

Other related results on Wyner-type networks can be found in \cite{BergelYellinShamai,  Shamai-Somekh-Simeone-Sanderovich-Zaidel-Poor-JWCC-2007,
SimeoneSomekhPoorShamai2009, SimeoneSomekhKramerPoorShamai2008, hanly, Bergel, Bjornson, Simeone,  Jafar2, Jafar3}.

The lower bounds in our paper  are based on coding strategies that silence some of the transmitters and thereby split the network into non-interfering subnetworks that can be treated separately. Depending on the considered setup, a different scheme is then used for the transmission in the subnetworks. 
In some setups, a part of the messages is transmitted using an interference cancellation scheme and the other part is transmitted using Costa's dirty-paper coding. (Costa's dirty paper coding can also be replaced by a simple linear beamforming scheme as e.g.,  in \cite{Lapidoth-Shamai-Wigger-ITW07}, see also \cite{Maddah08, CadambeJafar08, Jafar11}.) In other setups, 
the messages are transmitted using one of the following elementary bricks of multi-user information theory depending on the available side-information: an optimal multi-input/multi-output (MIMO) scheme, an optimal MIMO multi-access scheme, or an optimal MIMO broadcast scheme. \mw{Introducing also Han-Kobayashi type ideas to our coding strategies might improve the performance of our schemes for finite SNR. }

Our upper bounds rely on an extension of Sato's multi-access channel (MAC) bound
\cite{Sato-IT-1981} to apply for more general interference networks with more than two transmitters and receivers and where the transmitters and the receivers have side-information (see also
\cite{Kramer-upper,Lapidoth-Shamai-Wigger-ISIT07, GamalAnnapureddy4} and in particular \cite[Lemma~1 and Theorem~3]{Lapidoth-Shamai-Wigger-ITW07}). 
More specifically, we first partition the $K$ receivers into groups $\mathcal{A}$ and $\mathcal{B}_1,\ldots, \mathcal{B}_q$, and
 we allow the receivers in Group~$\mathcal{A}$ to cooperate. Then, we let a
genie reveal specific linear combinations of the noise sequences to
these receivers in Group $\mathcal{A}$. 
Finally, we request that the receivers in Group~$\mathcal{A}$ jointly decode all messages $M_1,\ldots,
M_K$ whereas all other receivers do not have to decode anything. 
We choose the genie-information so that: for each $i=1,\ldots, q$, if the receivers in Group~ $\mathcal{A}$ can successfully decode their own  messages and the messages intended for the receivers in groups~$\mathcal{B}_1,\ldots, \mathcal{B}_{i-1}$, then 
they can also reconstruct the outputs observed at the receivers in Group $\mathcal{B}_i$. In this case, they can also  decode the messages intended for the receivers in group $\mathcal{B}_i$ at least as well as the Group~$\mathcal{B}_i$ receivers. This iterative argument is used to show that the capacity
region of the resulting MAC is included in the capacity region of
the original network. The upper bound is then concluded by upper bounding the multiplexing gain of the  MAC.

We conclude this section with notation and an outline of the paper. 
\mw{Throughout the paper, $\mathbb{R}$, $\mathbb{N}$, and $\mathbb{N}_0$ denote the sets of real numbers, natural numbers, and nonnegative integers. Their $m$-fold Cartesian products are denoted $\mathbb{R}^n$, $\mathbb{N}^m$, and $\mathbb{N}_0^m$.} Also,  $\log(\cdot)$ denotes the natural logarithm, and  $a \mod b$ denotes the rest in the Euclidean division of $a$ by $b$. Random variables are denoted by upper case letters, their realizations by lower case letters. Vectors are denoted by bold letters: random vectors by upper case bold letters and deterministic vectors by lower case bold letters. Given a sequence of random variables $X_1,\ldots, X_n$ we denote by $X^n$ the tuple $(X_1,\ldots, X_n)$ and by $\bfX$ the $n$-dimensional column-vector $\trans{(X_1,\ldots, X_n)}$. \mw{For sets we use calligraphic symbols, e.g., $\mathcal{A}$. The difference of two sets $\mathcal{A}$ and $\mathcal{B}$ is denoted $\mathcal{A}\backslash \mathcal{B}$. We further use the Landau symbols, and thus $o(x)$ denotes a function that grows sublinearly in $x$.}

The paper is organized as follows. In Section~\ref{sec:asym} we describe the channel model and the results for the asymmetric network; in Section~\ref{sec:sym} the channel model and the results for the symmetric network. In Section~\ref{sec:converse} we present a \emph{Dynamic-MAC Lemma} that we use to prove our converse results for the multiplexing-gain. In the  rest of the paper we prove our presented results: in Section~\ref{sec:thmg} the results for the asymmetric network; in Section~\ref{sec:symach} the achievability results for the symmetric network with symmetric side-information; in Section~\ref{app:lowerbound} the achievability results for the symmetric network with general side-information; and finally in Section~\ref{sec:ub} the converse results for the symmetric network with general side-information parameters.


\section{Asymmetric Network}\label{sec:asym}

\subsection{Description of the Problem}
\label{sec:DescriptionOfTheProblem}


We consider $K$ transmitter/receiver pairs that are labeled from $\{1,\ldots, K\}$. 
The goal of the communication is that, for each
$k\in\K$, Transmitter~$k$ conveys Message $M_k$ to Receiver~$k$.  The
messages $\{M_k\}_{j=1}^K$ are assumed to be independent with $M_{k}$
being uniformly distributed over the set
$\mathcal{M}_{k}\triangleq\{1,\ldots, \lfloor e^{n R_k}\rfloor\}$,
where $n$ denotes the block-length of transmission and $R_k$ the
rate of transmission of  Message $M_k$.

In our setup, all the transmitters and the receivers are equipped with a single antenna and the channels are  discrete-time and real-valued. Denoting the time-$t$ channel input at Transmitter~$k \in \{1,\ldots, K\}$ 
by $x_{k,t}$, the time-$t$ channel output at Receiver~$k$'s antenna can be expressed as:
\begin{equation}
\label{channel}
Y_{k,t}= x_{k,t} + \alpha_k x_{k-1,t} +N_{k,t}, \quad k\in\{1,\ldots, K\};
\end{equation}
where for each $k\in\{1,\ldots, K\}$ the noise sequence $\{N_{k,t}\}$ is an independent sequence of independent and identically distributed (i.i.d.) standard Gaussians; where the cross-gain $\alpha_k$ is some given non-zero real number; and where to simplify notation we defined $x_{0,t}$  to be deterministically 0 for all times $t$. Thus, the communication of the $k$-th transmitter/receiver pair is interfered only by the communication of the transmitter/receiver pair to its left; see Figure~\ref{fig: setting-handoff}.

It is assumed that each
transmitter beside its own message is 
also cognizant of the $t_\ell\geq 0$ previous messages and the $t_r\geq 0$ following messages. That means, for
each $k\in\K$, Transmitter~$k$ knows messages $M_{k-t_\ell}, \ldots, M_k, \ldots, M_{k+t_r}$,
where  $M_{-t_\ell+1}, \ldots, M_0$ and $M_{K+1}, \ldots, M_{K+t_r}$ are defined to be
deterministically zero.  Thus, Transmitter~$k$ can produce its sequence
of channel inputs $X_k^n$
as
\begin{equation}\label{eq:enc}
{X}^n_{k} = f_k^{(n)}(M_{k-t_\ell}, \ldots, M_k, \ldots, M_{k+t_r}),
\end{equation}
for some encoding function 
\begin{equation}\label{eq:encofun}
f_{k}^{(n)}\colon
\mathcal{M}_{k-t_\ell}\times \cdots \times \mathcal{M}_{k} \times \cdots \times \mathcal{M}_{k+t_r} \rightarrow \R^n.\end{equation}

The channel input sequences are subject to symmetric average
block-power constraints, i.e., with probability 1 they have to satisfy
\begin{equation}\label{eq:power}
\frac{1}{n} \sum_{t=1}^n X_{k,t}^2 
\leq P, \quad k \in \K.
\end{equation}

\newcommand{\param}{t_\ell, t_r, r_\ell, r_r}

Each receiver observes the signals received at its own
antenna,  at the $r_\ell\geq 0$ antennas to
its left, and at the $r_r\geq 0$ antennas to its right.  
Receiver~$k$, for $k\in\K$, can thus produce its guess of Message $M_k$ based on
the output sequences ${Y}^n_{k-r_\ell}, \ldots,{Y}^n_{k+r_r}$, i.e., as 
\begin{equation}\label{eq:dec}
\hat{M}_{k} \triangleq \varphi^{(n)}_k({Y}_{k-r_\ell}, \ldots, {Y}^n_{k+r_r}) , 
\end{equation}
for some decoding function 
\begin{equation}\label{eq:decofun}
\varphi_k^{(n)}\colon \mathbb{R}^{n(r_\ell+r_\ell+1)} \rightarrow \mathcal{M}_k,
\end{equation}
where 
${Y}^n_{-r_\ell+1},\ldots,{Y}^n_{0}$ and ${Y}^n_{K+1},\ldots, {Y}^n_{K+r_r}$ are assumed to
be deterministically 0.

The parameters $t_\ell, t_r, r_\ell, r_r\geq 0$ are given positive integers. We call $t_\ell$ and $t_r$ the \emph{transmitter side-information parameters} and $r_\ell$ and $r_r$ the \emph{receiver side-information parameters}. Similarly, we call $t_\ell$ and $r_\ell$ the \emph{left side-information parameters} and $t_r$ and  $r_r$ the \emph{right side-information parameters}.

For the described setup we say that a rate-tuple $(R_1,\ldots, R_K)$ is achievable if, 
as the block-length $n$
tends to infinity, the average probability of error decays to 0, i.e.,
\begin{equation*}
\lim_{n\rightarrow 0}\Pr \left[(M_1,\ldots, M_{K}) \neq (\hat{M}_1,\ldots, \hat{M}_K)\right] =0.
\end{equation*}
The closure of the set of all rate-tuples $(R_1,\ldots, R_K)$ that are achievable is called the capacity region, which we denote by $\Ca^{\Asym}$. To make the dependence on the number of transmitter/receiver pairs $K$, the side-information parameters $\param$, and the power $P$ explicit, we mostly write $\Ca^{\Asym}(K,\param;P)$.
The sum-capacity is defined as the supremum of the sum-rate $\sum_{k=1}^{K} R_k$ over all achievable tuples $(R_1,\ldots,
R_K)$ and is denoted by $\Ca^{\Asym}_{\Sigma}(K, t_\ell,t_r,
r_\ell,r_r;P)$.
Our main focus in this work is on the high-SNR asymptote of the sum-capacity which is characterized
by the \emph{multiplexing gain}:\footnote{The multiplexing gain is
  also referred to as the ``high-SNR slope", ``pre-log", or ``degrees
  of freedom"}
\begin{eqnarray*}
 \eta^{\textnormal{Asym}} (K,t_\ell,t_r,r_\ell,r_r)\triangleq \varlimsup_{P\rightarrow\infty}  \frac{\Ca^{\textnormal{Asym}}_{\Sigma}(K,t_\ell,t_r,r_\ell,r_r;P)}{\frac{1}{2}\log(P) }, \end{eqnarray*}
and for large networks $(K\gg 1)$ by the \emph{asymptotic multiplexing gain per user}:
\begin{eqnarray*}
\eta^{\textnormal{Asym}}_{\infty}(t_\ell,t_r,r_\ell,r_r) \triangleq \varlimsup_{K\rightarrow\infty}  \frac{ \eta^{\textnormal{Asym}} (K,t_\ell,t_r,r_\ell,r_r)}{K}. \end{eqnarray*}
\subsection{Results}
\label{sec:TheSoftHandoffNetwork}

\begin{theorem}\label{th:mg}
The multiplexing gain of the asymmetric model is 
\begin{equation}\label{eq:mg}
\eta^{\textnormal{Asym}}(K,t_\ell,t_r,r_\ell,r_r)=K-\left\lceil\frac{ K-t_\ell-r_\ell-1}{t_\ell+t_r+r_\ell+r_r+2} \right\rceil .\end{equation}
\end{theorem}
\begin{proof}
See Section~\ref{sec:ach} for the direct part and
Section~\ref{sec:convTh1} for the converse.
\end{proof}

Specializing Theorem~\ref{th:mg} to the case $r_\ell=r_r=0$ where each receiver has access only to its own
receive antenna, recovers the result in \cite{Lapidoth-Shamai-Wigger-ISIT07}.

\begin{remark}
  Notice that Expression \eqref{eq:mg} depends only on the sum of the
  {left side-information parameters} $t_\ell+r_\ell$ and on
  the sum of the {right side-information parameters} $t_r+r_r$.  This shows an equivalence between cognition of messages at
  the transmitters and clustered local decoding at the receivers.

  Notice however that the left side-information parameters $r_\ell$
  and $t_\ell$ do not play the same role as the right side-information
  parameters $t_r$ and $r_r$.  In fact, left side-information can be more valuable (in terms of increasing the multiplexing gain)
  than right side-information. 
    
The difference in
the roles of left and right side-information is only a
boundary effect and vanishes when $K\rightarrow \infty$, see Corollary~\ref{th:mginfty} and Remark~\ref{rem:leftrightsym} ahead.\end{remark}

As a corollary to Theorem~\ref{th:mg} we can derive the
\emph{asymptotic multiplexing gain per user}.
\begin{corollary}\label{th:mginfty}
The asymptotic  multiplexing gain per user
of the asymmetric network is 
\begin{equation}\label{eq:mginfty}
\eta^{\textnormal{Asym}}_{\infty}(t_\ell,t_r,r_\ell,r_r)= \frac{t_\ell+t_r+r_\ell+r_r+1}{t_\ell+t_r+r_\ell+r_r+2}.
\end{equation}
\end{corollary}

\begin{remark}\label{rem:leftrightsym}
  The asymptotic multiplexing gain per-user in \eqref{eq:mginfty}
  depends on the parameters $t_\ell$, $t_r$, $r_\ell$, and $r_r$ only
  through their sum. Thus, in the considered
  setup the asymptotic multiplexing gain per-user only depends on the
  total amount of side-information at the transmitters and receivers
  and not on how the side-information is distributed. In particular, cognition of
  messages at the transmitters  and clustered
  local decoding at the receivers are equally valuable, and---despite
  the asymmetry of the interference network---also left and right
  side-information are equally
  valuable.
\end{remark}

\mw{  El Gamal, Annapureddy, and Veeravalli~\cite{GamalAnnapureddy4}  showed that when $r_\ell=r_r=0$ and when for each message one can freely choose the set of  $t_\ell+t_r+1$ transmitters to which this message is assigned,  then the asymptotic multiplexing gain per-user is equal to $\frac{2(t_\ell+t_r+1)}{2(t_\ell+t_r+1)+1}$ and thus larger than $\eta_\infty^{\textnormal{Asym}}$ in~\eqref{eq:mginfty}. They also showed that in this modified setup, each message $M_k$  should again be assigned to $t_\ell+t_r+1$ adjacent transmitters, but these transmitters do not necessarily include Transmitter~$k$. }

\section{Symmetric Network}
\label{sec:sym}
\subsection{Description of the Problem}\label{sec:symmodel} 
The symmetric network is defined in the same way as the asymmetric network in Section~\ref{sec:asym}, except that the channel law (\ref{channel}) is replaced by
\begin{eqnarray}\label{eq:symchannel}
Y_{k,t}= \alpha_{k,\ell} X_{k-1,t} + X_{k,t} + \alpha_{k,r} X_{k+1,t} +N_{k,t}, \hspace{1cm} \nonumber \\ \hfill \quad k\in\{1,\ldots, K\}.
\end{eqnarray}
Like for the asymmetric network, for each $k\in\{1,\ldots,K\}$ the
symbol $X_{k,t}$ denotes Transmitter~$k$'s channel input at time
$t$;  the symbols  $X_{0,t}$ and $X_{K+1,t}$ are
deterministically zero; the cross-gains $\{\alpha_{k,\ell}, \alpha_{k,r}\}$ are given non-zero real numbers; and $\{N_{k,t}\}$ are 
i.i.d. standard
Gaussians. Let $\mat{H}_{\textnormal{Net}}$ denote the $K$-by-$K$ channel matrix of the entire network: its row-$j$, column-$i$ element equals 1 if $j=i$, it equals $\alpha_{j,\ell}$ if $j-i=1$,  it equals $\alpha_{j, r}$ if $j-i=-1$, and it equals 0 otherwise.

The message cognition at the transmitters 
is again described by the nonnegative
integers $t_\ell$ and $t_r$ and the encoding rules in \eqref{eq:enc},
and the clustered decoding  by the nonnegative integers
$r_\ell$ and $r_r$ and the decoding rules in~\eqref{eq:dec}. 

The channel input sequences have to satisfy the power
constraints~\eqref{eq:power}.

Achievable rates, channel
capacity, sum-capacity, multiplexing gain, and the asymptotic multiplexing gain per user are defined
analogously to Section~\ref{sec:asym}. For this symmetric model and for a
given positive integer $K>0$, nonnegative integers $t_\ell, t_r,
r_\ell, r_r \geq 0$, and power $P>0$ the capacity region is
denoted by $\Ca^{\textnormal{Sym}}(K,t_\ell, t_r, r_\ell, r_r;P)$, the
sum-capacity by
$\Ca^{\textnormal{Sym}}_{\Sigma}(K,t_\ell,t_r,r_\ell,r_r;P)$, the multiplexing gain by
$\eta^{\textnormal{Sym}}(K,t_\ell,t_r,r_\ell,r_r)$, and the asymptotic
multiplexing gain per user by
$\eta^{\textnormal{Sym}}_{\infty}(t_\ell,t_r,r_\ell,r_r)$.

\mw{We shall mostly restrict attention to equal cross-gains, i.e., $\alpha_{k,\ell}=\alpha_{k,r}=\alpha$ for all $k\in\{1,\ldots, K\}$ and some $\alpha \neq 0$. However, our proof techniques extend also to non-equal cross gains. In fact, by inspection of the proofs, one sees that they only depend on the cross-gains through the ranks of various principal submatrices of the  network's channel matrix and the fact that the cross-gains are nonzero. 

A formulation of our results for general cross-gains would involve conditions on the rank of various principal submatrices of the network's channel matrix and be very cumbersome. We therefore omit it. Instead, we will extend our results extend to a setup where all cross-gains are drawn according to a continuous distribution.\footnote{Such cross-gains are typically called \emph{generic} \cite{CadambeJafar08,Jafar11}. Here, we refrain from calling them so as to avoid confusion with generic subnets which we introduce in our achievability proofs.} In this case, all principal submatrices of the  channel matrix  are full rank  and all cross-gains are nonzero with probability~$1$. 
}

\subsection{Results}\label{sec:det}

\mw{We mostly restrict attention to equal cross-gains, i.e., $\alpha_{k,\ell}=\alpha_{k,r}=\alpha$ for all $k\in\{1,\ldots, K\}$ and some $\alpha \neq 0$. }
This assumption motivates the following definition of a channel matrix.
\begin{definition}
For every positive integer $p\geq1$ and real number $\alpha$ we define
$\mat{H}_p(\alpha)$ to be the  $p\times p$ matrix with diagonal elements all equal to 1, elements above and below the diagonal equal to $\alpha$, and all other elements equal to 0.
\end{definition}
Notice that under the assumption of all equal cross-gains $\alpha$, the network's channel matrix is $\mat{H}_{\textnormal{Net}}=\mat{H}_K(\alpha)$. 

We first present our results for \emph{symmetric side-information}
where 
\begin{equation}\label{eq:equal}
t_\ell+r_\ell=t_r+r_r, 
\end{equation}
followed by our results 
for general side-information parameters $r_\ell,t_\ell, r_r, t_r\geq 0$. We treat the special case with symmetric side-information separately, because for this case we have stronger results than for general side-information.

\subsubsection{Symmetric Side-Information}
\label{sec:ressymsi}
Throughout this subsection we assume that the parameters
$t_\ell, t_r, r_\ell, r_r$ satisfy~\eqref{eq:equal}.
\begin{theorem}[Symmetric Side-Information]\label{th:symsym}
Depending on the value of $\alpha$ and the parameters $K, t_\ell, t_r, r_\ell, r_r$, the multiplexing gain satisfies the following conditions.
\begin{enumerate}
\item \label{e1} If $K\leq t_\ell+r_\ell+1$:
\begin{equation}\label{eq:S1}
\eta^{\Sym}(K,t_\ell,t_r, r_\ell, r_r) = K-\delta_1,
\end{equation}
where 
$\delta_1$ equals 1, if $\deta{\mat{H}_{K}(\alpha)}=0$ and 0 otherwise.
\item \label{e2} If $K> t_\ell+r_\ell+2$ and  $\deta{\mat{H}_{t_\ell+r_\ell+1}(\alpha)}\neq 0$:
\begin{IEEEeqnarray}{rCl}\label{eq:S2a}
K-\left \lfloor \frac{K}{t_\ell+r_\ell+2}\right \rfloor -1 &\leq& \eta^{\Sym}(K,t_\ell,t_r, r_\ell, r_r) \nonumber \\
 &  \leq& K-\left \lfloor \frac{K}{t_\ell+r_\ell+2}\right \rfloor. \nonumber \\ 
\end{IEEEeqnarray}
 \item \label{e3} If  $K> t_\ell+r_\ell+2$;  $\deta{\mat{H}_{t_\ell+r_\ell+1}(\alpha)}\neq 0$; and  $\deta{\mat{H}_{t_\ell+r_\ell}(\alpha)}\neq 0$, then\begin{equation}\label{eq:S2}
\eta^{\Sym}(K,t_\ell,t_r, r_\ell, r_r)  = K-\left \lfloor \frac{K}{t_\ell+r_\ell+2}\right \rfloor.
\end{equation} 
(This third case  is a special case of the second case. It is interesting because almost all values of $\alpha$ lead to this case and because for this case we can improve the lower bound in \eqref{eq:S2a} to meet the upper bound.)
\item \label{e4} If $K> t_\ell+r_\ell+2$ and $\deta{\mat{H}_{t_\ell+r_\ell+1}(\alpha)}= 0$:
\begin{IEEEeqnarray}{rCl}\label{eq:S3}
 K-\left \lfloor \frac{K}{t_\ell+r_\ell+1}\right \rfloor &\leq &\eta^{\Sym}(K,t_\ell,t_r, r_\ell, r_r) \nonumber \\ &  \leq& K- 2\left \lfloor \frac{K}{2(t_\ell+r_\ell)+3} \right \rfloor - \delta_2,\nonumber \\ 
\end{IEEEeqnarray}
where $\delta_2$ equals 1 if $(K \mod (2(t_\ell+r_\ell)+3)) > (t_\ell+r_\ell+1)$ and 0 otherwise. 
\end{enumerate}
\end{theorem}
\begin{proof}
The achievability results are proved in Section~\ref{sec:symach}. 
The converse in~\eqref{eq:S1} can be proved by first allowing all the transmitters to cooperate and all the receivers to cooperate, and then using the well-known expression for the capacity of the multi-antenna Gaussian point-to-point channel. The converse to~\eqref{eq:S2a} and~\eqref{eq:S2} follows by specializing Upper bound~\eqref{eq:genub1} in Proposition~\ref{th:genub} ahead to $t_\ell+r_\ell=t_r+r_r$. Similarly, the converse to \eqref{eq:S3} follows by specializing  \eqref{eq:genub2} to $t_\ell+r_\ell=t_r+r_r$. 
\end{proof}
\mw{
\begin{remark}
Inspecting the achievability and the converse proofs of~\eqref{eq:S2}, we see that they continue to hold for arbitrary cross-gains provided that they are non-zero and that various principal submatrices of the network's channel matrix $\mat{H}_{\textnormal{Net}}$ have full rank. When the cross-gains are drawn at random according to a continuous distribution both these properties are satisfied with probability $1$, and thus for this random setup with symmetric side-information parameters the multiplexing gain is as given in~\eqref{eq:S2} (with probability~$1$). 
\end{remark}}
\begin{remark}
\mw{We observe that when $\mat{H}_{t_\ell+r_\ell+1}(\alpha)$ and $\mat{H}_{t_\ell+r_\ell}(\alpha)$ are full rank (or when the cross-gains are drawn from a continuous distribution), the multiplexing gain only depends on the sum of the
side-information parameters $(t_\ell+r_\ell)$. Or equivalently they only depend on the sums $(t_r+r_r)$ or
$(t_\ell+t_r+r_\ell+r_r)$. Thus, in these cases,}
message cognition at the transmitters and clustered local decoding at
the receivers are equivalent in terms of increasing the
multiplexing gain.
\end{remark}

The following corollary is obtained from  Theorem~\ref{th:symsym} by letting $K$ tend to $\infty$.
\begin{corollary}\label{th:symmgperuser}
If 
$\deta{\mat{H}_{t_\ell+r_\ell+1}(\alpha)}\neq 0$, then
the asymptotic multiplexing gain per-user is given by
        \begin{equation}\label{eq:MGSYM}
        \eta^{\textnormal{Sym}}_{\infty}(t_\ell,t_r,r_\ell,r_r)=\frac{t_\ell+r_\ell+1}{t_\ell+r_\ell+2}.
        \end{equation}
Otherwise, it satisfies
      \[\frac{t_\ell+r_\ell}{t_\ell+r_\ell+1} \leq \eta^{\textnormal{Sym}}_{\infty}(t_\ell,t_r,r_\ell,r_r)\leq \frac{t_\ell+r_\ell+\frac{1}{2}}{t_\ell+r_\ell+\frac{3}{2}}.\]
      Thus, for a few values $\alpha\neq0$ the asymptotic multiplexing gain per-user drops.
\end{corollary}

\begin{remark}
When $\deta{\mat{H}_{t_\ell+r_\ell+1}(\alpha)}\neq 0$, then to obtain the same asymptotic multiplexing-gain per-user 
in this symmetric network as in the
asymmetric network before,  we need double the ``amount'' of
side-information $t_\ell+t_r+r_\ell+r_r$.
\end{remark}

\mw{El Gamal et al. \cite{GamalAnnapureddy4} showed that also here a larger asymptotic multiplexing gain per-user is achievable when the messages are assigned to the transmitters in a different way (even when $r_\ell=r_r=0$). In particular, if each message can be freely assigned to $t_\ell+t_r+1$ transmitters, then an asymptotic multiplexing gain per-user of $\frac{2(t_\ell+t_r+1)}{2(t_\ell+t_r+1)+2}$ is achievable \cite{GamalAnnapureddy4}, which is larger than $\eta_\infty^{\textnormal{Sym}}$ in \eqref{eq:MGSYM}.}

\begin{example}
\label{ex:7}
Consider a symmetric network with symmetric side-information
$r_\ell+t_\ell=r_r+t_r=2$. Let $K$ be $7$. Then, if $\alpha\notin \{-\sqrt{2}/2,
\sqrt{2}/2\}$, by Theorem~\ref{th:symsym} the multiplexing gain
is 6, and in contrast,  if $\alpha\in\{-\sqrt{2}/2, \sqrt{2}/2\}$ the multiplexing gain is only 5. 

By Corollary~\ref{th:symmgperuser} the asymptotic multiplexing
gain  per-user is $3/4$, if $\alpha\notin \{-\sqrt{2}/2, \sqrt{2}/2\}$, but it
is at most $5/7$ (which is smaller than $3/4$) if
$\alpha\in\{-\sqrt{2}/2,\sqrt{2}/2\}$.
\end{example}

Notice however, that even though the multiplexing gain is discontinuous at certain values of $\alpha$, this does not 
imply that for fixed powers $P$ also the sum-rate capacity of the network is discontinuous in $\alpha$. 

We conclude this section with a result on 
the \emph{high-SNR power-offset} which is defined as
\begin{eqnarray*}
\lefteqn{\mathcal{L}^{\textnormal{Sym}}_{\infty}(K, t_\ell, t_r, r_\ell,
r_r)}\\ & \triangleq & \varlimsup_{P\rightarrow\infty}  \left(\frac{\eta^{\Sym}}{2}\log(P)  -  \Ca^
{\textnormal{Sym}}_{\Sigma}(K,t_\ell,t_r,r_\ell,r_r;P) \right).
\end{eqnarray*}
\begin{proposition}[Symmetric Side-Information] \label{prop:poweroffset}
Assume (\ref{eq:equal}). Let  $\alpha^*$  be
such that $\deta{\mat{H}_{r_\ell+t_\ell+1}(\alpha^*)}=0$. Also, let $K=q (r_\ell+t_\ell+2)-1$ for some positive integer $q$. Then, there exists a function $c_0(\cdot)$, bounded in the neighborhood of $\alpha^*$ such that for all
$\alpha$ sufficiently close to $\alpha^*$:
\begin{equation*}
\mathcal{L}_\infty^{\textnormal{Sym}} (K, t_\ell, t_r, r_\ell,
r_r)\geq  -\nu\log|\alpha-\alpha^*| + c_0(\alpha^*),
\end{equation*}
where $\nu$ is the multiplicity of $\alpha^*$ as a root of the polynomial $\deta{\mat{H}_{r_\ell+t_\ell+1}(X)}$.
\end{proposition}
In other words, when $\alpha$ approaches the critical value $\alpha^*$, the power offset goes to infinity.
\begin{proof}
See Appendix~\ref{sec:poweroffset}.
\end{proof}

\subsubsection{Results for General Parameters $t_\ell, t_r, r_\ell, r_r
  \geq 0$}\label{sec:gensym}

\begin{proposition}\label{prop:genlb}
The multiplexing gain of the symmetric network with general side-information parameters satisfies the following four lower bounds.
\begin{enumerate}
\item \label{1} It is lower bounded by:
                \begin{equation}
                \eta^{\textnormal{Sym}}(K,t_\ell,t_r,r_\ell,r_r)\geq
                K-2\left \lfloor \frac{ K}{t_\ell+t_r+r_\ell+r_r } \right \rfloor - \theta_1, \label{eq:lowergen1}
                \end{equation}
                where \begin{equation*}
                \theta_1 = \begin{cases} 0 & \textnormal{if }\kappa_1 =0\\ 1 &  \textnormal{if }\kappa_1=1 \\ 2 &  \textnormal{if }\kappa_1\geq 2
                \end{cases} 
                \end{equation*} for 
                \begin{equation*}
                \kappa_1\triangleq  ( K \mod (_\ell+t_r+r_\ell+tr_r)) .
                \end{equation*}
    \item  \label{2}Moreover, irrespective of the right side-information $t_r$ and $r_r$:
  \begin{equation}
                \eta^{\textnormal{Sym}}(K,t_\ell,t_r,r_\ell,r_r)\geq
                K-2\left \lfloor \frac{ K}{t_\ell+r_\ell+1} \right \rfloor - \theta_2,\label{eq:lowergen4}
                \end{equation}
                where \begin{equation*}
                \theta_2 = \begin{cases} 0 & \textnormal{if }\kappa_2 =0\\ 1 &  \textnormal{if }\kappa_2=1 \\ 2 &  \textnormal{if }\kappa_2\geq 2
                \end{cases} 
                \end{equation*} for 
                \begin{equation*}
                \kappa_2\triangleq  ( K \mod (t_\ell+r_\ell)) .
                \end{equation*}          
\item \label{3}The  lower bound \eqref{eq:lowergen4} in \ref{2}) remains valid if on the right-hand side of \eqref{eq:lowergen4} we replace the parameters $t_\ell$ and $r_\ell$ by $t_r$ and $r_r$.                                
\item \label{4} Finally, irrespective of the transmitter side-information $t_\ell$ and $t_r$: 
  \begin{IEEEeqnarray}{rCl}\label{eq:Sl2}
\eta^{\textnormal{Sym}}(K,t_\ell,t_r,r_\ell,r_r)\geq K-2\left\lfloor \frac{K}{r_\ell+r_t+3}\right \rfloor - \theta_3,\nonumber \\ \label{eq:lowergen2}
  \end{IEEEeqnarray}
    where \begin{equation*}
                \theta_3 = \begin{cases} 0 & \textnormal{if }\kappa_3 =0\\ 1 &  \textnormal{if }\kappa_3=1 \\ 2 &  \textnormal{if }\kappa_3\geq 2
                \end{cases} 
                \end{equation*} for 
                \begin{equation*}
                \kappa_3\triangleq  ( K \mod (r_\ell+r_r+3)) .
                \end{equation*}
\end{enumerate}              
\end{proposition}
\begin{proof}
See Section~\ref{app:lowerbound}.
\end{proof}
The lower bound in \ref{2}) is useful only when $t_r=r_r=0$, the lower bound in \ref{3}) only when $t_\ell=r_\ell=0$, and the bound in \ref{4}) only when $t_\ell+t_r \leq 2$.

\begin{proposition}\label{th:genub}
The multiplexing gain is upper bounded by the following three upper bounds. 
\begin{enumerate}
\item \label{1a}
It is upper bounded by:
\begin{IEEEeqnarray}{rCl}\label{eq:genub1}
\lefteqn{ \eta^{\textnormal{Sym}}(K,t_\ell,t_r,r_\ell,r_r)}\qquad \nonumber \\&  \leq& K-2\left \lfloor \frac{K}{t_\ell+t_r+r_\ell+r_r+4} \right \rfloor-{\theta_4},\IEEEeqnarraynumspace
\end{IEEEeqnarray}
 where    \begin{equation*}
                \theta_4 = \begin{cases} 0 & \textnormal{if }\kappa_4 <\min\{t_\ell+r_\ell+2, t_r+r_r+2\} \\ 1 &  \textnormal{if }\kappa_4 \geq \min\{t_\ell+r_\ell+2, t_r+r_r+2\}
                \end{cases} 
                \end{equation*} for 
                \begin{equation*}
                \kappa_4\triangleq  ( K \mod (t_\ell+t_r+r_\ell+r_r+4)) .
                \end{equation*}
  \item  \label{2a} Moreover, if $\deta{\mat{H}_{r_\ell+t_\ell+1}(\alpha)}=0$:
  \begin{IEEEeqnarray}{rCl}\label{eq:genub2}
\lefteqn{ \eta^{\textnormal{Sym}}(K,t_\ell,t_r,r_\ell,r_r)}\qquad \nonumber \\&  \leq&
  K-2\left \lfloor \frac{K}{t_\ell+t_r+r_\ell+r_r+3} \right \rfloor-{\theta_5},\IEEEeqnarraynumspace
  \end{IEEEeqnarray}
where  \begin{equation*}
                \theta_5 = \begin{cases} 0 & \textnormal{if }\kappa_5 <t_r+r_r+1 \\ 1 &  \textnormal{if }\kappa_5 \geq t_r+r_r+1
                \end{cases} 
                \end{equation*} for 
                \begin{equation*}
                \kappa_5\triangleq  ( K \mod (t_\ell+t_r+r_\ell+r_r+3)) .
                \end{equation*}
\item  \label{3a} The upper bound in \ref{2a}) holds also if everywhere (except for $\eta^{\textnormal{Sym}}(K,t_\ell,t_r,r_\ell,r_r)$) one exchanges the subscripts $\ell$ and $r$.
\end{enumerate}
\end{proposition}
\begin{proof}
See Section~\ref{sec:ub}.
\end{proof}
From Propositions~\ref{prop:genlb} and~\ref{th:genub} we obtain the following corollary.
\begin{corollary}\label{cor:gen}
Irrespective of the parameter $\alpha$, the asymptotic multiplexing
gain per user satisfies
\begin{eqnarray*}
\lefteqn{\max\left\{ \frac{r_\ell+r_r+1}{r_\ell+r_r+3},\frac{t_\ell+t_r+r_\ell+r_r-2}{t_\ell+t_r+r_\ell+r_r}\right\}\nonumber } \hspace{1.4cm} \\&\leq&
\eta^{\textnormal{Sym}}_\infty(t_\ell,t_r,r_\ell,r_r) \\ & &\hspace{1.4cm} \leq \frac{t_\ell+t_r+r_\ell+r_r+2}{t_\ell+t_r+r_\ell+r_r+4} .\end{eqnarray*}
\end{corollary}
\mw{
\begin{remark}
All lower bounds in Proposition~\ref{prop:genlb} and Upper bound~\eqref{eq:genub1} in Proposition~\ref{th:genub} continue to hold (with probability~$1$), if the cross-gains are randomly drawn according to a continuous distribution. As a consequence also Corollary~\ref{cor:gen} remains valid in this random setup. 

This can be seen by inspecting the proofs and noticing that they hold for arbitrary nonzero cross-gains; in our random setup the cross-gains are nonzero with probability 1.
\end{remark}}

\mw{\section{Converse Proofs}\label{sec:converse}
 Our converse proofs all rely on the following lemma. 
 
 
 For a given set of receivers $\mathcal{S}\subseteq \mathcal{K}$, let $\mathcal{R}_{\mathcal{S}}$ denote the set of indices $k\in \mathcal{K}$ such that Antenna~$k$ is observed by at least one of the receivers in $\mathcal{S}$.
 \begin{lemma}[Dynamic-MAC Lemma]\label{lem:dynMAC}
 Consider a general interference network with message cognition and clustered decoding. Let $\vect{V}_0,\ldots, \vect{V}_g$, for $g\in\mathbb{N}_0$, be a set of genie-signals and  let $\mathcal{A},\mathcal{B}_1, \mathcal{B}_2, \ldots, \mathcal{B}_q$, $q\in\mathbb{N}$, form a partition of the set of receivers $\mathcal{K}$, such that for all $k\in\mathcal{K}$ the differential entropy
 \begin{equation}\label{eq:assum_noise}
h\big( \{\vect{N}_k\}_{k\in\mathcal{R}_\mathcal{A}}| \vect{V}_0,\ldots, \vect{V}_q\big)
 \end{equation} 
 is finite and bounded in $P$.\footnote{For the lemma to hold, it suffices that the differential entropies grow slower than any multiple of $n\log(P)$.}
 If for any given encoding and decoding functions $f_1^{(n)},\ldots, f_K^{(n)}$ and $\varphi_1^{(n)}, \ldots, \varphi_K^{(n)}$ there exist deterministic functions $\xi_1, \ldots, \xi_q$ on the respective domains such that for each $i \in \{1,\ldots, q\}$:
 \begin{equation}\label{eq:assum_fun}
\{\vect{Y}_k\}_{k\in \mathcal{R}_{\mathcal{B}_i}} = \xi_i\big( \{ \vect{Y}_k\}_{k\in \mathcal{R}_{\mathcal{A}_i}}, \{ M_{k}\}_{k\in\mathcal{A}_i}, \vect{V}_1,\ldots, \vect{V}_g\big),
 \end{equation}
 where $\mathcal{A}_i\triangleq \mathcal{A} \cup \mathcal{B}_1 \cup \ldots \cup \mathcal{B}_{i-1}$, then the multiplexing gain of the network is upper bounded as
\begin{equation}\label{eq:converse}
  \eta\leq |\mathcal{R}_{\mathcal{A}}|.
  \end{equation}
 \end{lemma}
 \begin{proof}
 To prove our desired upper bound we introduce a \emph{Cognitive MAC},
whose capacity region $\Ca_{\textnormal{MAC}}$ includes the capacity region of the original network,
\begin{IEEEeqnarray}{rCl}\label{eq:supseteq}
 \Ca &\subseteq&
 \Ca_{\textnormal{MAC}},
\end{IEEEeqnarray} 
and whose multiplexing gain $\eta_{\textnormal{MAC}}$
 is upper bounded as
\begin{equation}\label{eq:uppermg}
\eta_{\textnormal{MAC}}  \leq |\mathcal{R}_\mathcal{A}|.
\end{equation}
Combining~\eqref{eq:supseteq} and \eqref{eq:uppermg} establishes the desired lemma. 

The Cognitive MAC  is obtained from the original network by 
 revealing the genie-information  $\vect{V}_{0}, \ldots, \vect{V}_{g}$ to the receivers in Group $\mathcal{A}$ and by requiring that all the receivers that are in Group $\mathcal{A}$ jointly decode  \emph{all  messages} $M_1,\ldots, M_K$, whereas all other receivers do not have to decode anything. Since the only remaining receivers in Group~$\mathcal{A}$ can all cooperate in their decoding, the Cognitive MAC is indeed a MAC with only one receiver.

We now prove Inclusion~\eqref{eq:supseteq} using a dynamic version of Sato's MAC-bound idea \cite{Sato-IT-1981}. Specifically, we show that  every coding scheme for the original network can be modified to a coding scheme for the Cognitive MAC  such
  that whenever the original scheme is successful (i.e, all messages
  are decoded correctly), then so is the modified scheme. 
Fix a coding scheme for the
  original network. The transmitters of
  the Cognitive MAC   apply the same encodings as in the original scheme. The only receiver of the Cognitive MAC, i.e., the Group~$\mathcal{A}$ receiver, performs the decoding in $q+1$ rounds $0, \ldots, q$. In round $i=0$, it decodes the  messages $\{M_k\}_{k\in\mathcal{A}}$ in the same way as in the given original
  scheme. In rounds $i=1,\ldots, q$,
  \begin{itemize}
  \item it attempts to reconstruct the channel outputs  $\{\vect{Y}_k\}_{k\in\mathcal{R}_{\mathcal{B}_i}}$
  observed by the receivers in Group~$\mathcal{B}_i$ using the previously  decoded messages $\{M_k\}_{k\in\mathcal{A}_i}$, the observed or previously reconstructed channel outputs $\{\vect{Y}_{k}\}_{k\in\mathcal{R}_{\mathcal{A}_i}}$,  and the genie-information $\vect{V}_0,\ldots, \vect{V}_g$; then
  \item  it  decodes the messages $\{M_k\}_{k\in\mathcal{B}_i}$ based on its reconstructions of the outputs $\{\vect{Y}_k\}_{k\in\mathcal{R}_{\mathcal{B}_i}}$ in the same way as the receivers in Group~$\mathcal{B}_i$ did in the original
  scheme. 
 \end{itemize} 
 By Assumption~\eqref{eq:assum_fun},  the round-$i$ reconstruction step is successful if all previous rounds'~$0,\ldots, i-1$ reconstruction and decoding steps were successful. Thus, the additional reconstruction steps in the Cognitive MAC decoding do not introduce additional error events compared to the
original decoding procedure, and Inclusion~\eqref{eq:supseteq} follows. 

We are left with showing that the multiplexing gain of the Cognitive MAC  is upper bounded by $|\mathcal{R}_{\mathcal{A}}|$. 
Since the  Group A receiver is required to decode all $K$
messages $M_1,\ldots, M_K$,  by Fano's inequality, reliable
communication is possible only if 
\begin{IEEEeqnarray}{rCl}
n\sum_{k=1}^K R_k & \leq & I\big(\{ \vect{Y}_{k}\}_{k\in\mathcal{R}_\mathcal{A}}, \vect{V}_{1},\ldots,
\vect{V}_{g}; M_1,\ldots, M_K\big)\nonumber \\
& = & I\big(\{ \vect{Y}_{k}\}_{k\in\mathcal{R}_\mathcal{A}}; M_1,\ldots, M_K | \vect{V}_{1},\ldots,
\vect{V}_{g} \big)\nonumber \\
& \leq & h\big(\{ \vect{Y}_{k}\}_{k\in\mathcal{R}_\mathcal{A}}\big)- h\big( \{\vect{N}_{k}\}_{k\in\mathcal{R}_\mathcal{A}}| \vect{V}_1, \ldots, \vect{V}_{g}\big).\nonumber \\\label{eq:MGup}
\end{IEEEeqnarray}
The multiplexing gain of $h\big(\{ \vect{Y}_{k}\}_{k\in\mathcal{R}_\mathcal{A}}\big)$ is bounded by $|\mathcal{R}_\mathcal{A}|$. Moreover, by assumption,  $h\big( \{\vect{N}_{k}\}_{k\in\mathcal{R}_\mathcal{A}}| \vect{V}_1, \ldots, \vect{V}_{g}\big)$ is finite and bounded in $P$. We therefore obtain from~\eqref{eq:MGup}
\begin{equation}\label{eq:limlog}
\varlimsup_{P\rightarrow \infty} \frac{\sum_{k=1}^K R_k}{\frac{1}{2}\log(P)} \leq |\mathcal{R}_\mathcal{A}|,
\end{equation}
which gives the desired bound~\eqref{eq:uppermg}.
 \end{proof}

}


 
 \section{Proof of Theorem~\ref{th:mg}}\label{sec:thmg}

Define   
\begin{IEEEeqnarray}{rCl}
\label{eq:defgamma}
\gamma& \triangleq& \left\lceil\frac{ K-t_\ell-r_\ell-1}{t_\ell+t_r+r_\ell+r_r+2} \right\rceil \\
\beta  & \triangleq &
t_\ell+t_r+r_\ell+r_r+2,\label{eq:beta}\\
\kappa & \triangleq & (K \mod \beta).
\end{IEEEeqnarray}

\subsection{Achievability Proof of Theorem~\ref{th:mg}}\label{sec:ach}

We derive a lower bound by giving an appropriate coding scheme. The idea is to silence some of the transmitters, which decomposes our asymmetric network  into several subnets (subnetworks), and  to apply a scheme based
on Costa's dirty-paper coding\footnote{Alternatively, also the simpler 
    \emph{partial interference cancellation} scheme in
    \cite{Lapidoth-Shamai-Wigger-ITW07}, which is based on linear beam-forming, could be used instead of the
    dirty-paper coding.} and
on successive interference cancellation in each of the subnets.

\subsubsection{Splitting the Network into Subnets}

We silence transmitters $j \beta$, for
$j \in\{1,\dots,\left\lfloor K/\beta\right\rfloor\}$; moreover, if $\kappa >(t_\ell+r_\ell+1)$ we also silence Transmitter $K$.
This splits the network into $\lceil{K/\beta}\rceil$ non-interfering
subnets. The first $\lfloor{K/\beta}\rfloor$ subnets all have the same topology; they
consist of $(t_\ell+t_r+r_\ell+r_r+1)$ active transmit antennas and $(t_\ell+t_r+r_\ell+r_r+2)$
receive antennas.
We refer to these subnets as \emph{generic} subnets. If
$K$ is not a multiple of $\beta$, there is an additional last subnet
with
\[\begin{cases} \kappa \textnormal{ active transmit antennas,} & 
\textnormal{if } \kappa \leq (t_\ell+r_\ell+1),\\ (\kappa-1) \textnormal{ active
  transmit antennas,} & \textnormal{if }\kappa > (t_\ell+r_\ell+1),
\end{cases}\]
and with $\kappa$ receive antennas. We refer to such a subnet as a \emph{reduced} subnet.

As we shall see, in our scheme  each transmitter ignores its side-information about the messages pertaining to transmitters in other subnets.
Likewise, each receiver ignores its side-information about the outputs of antennas  belonging to receivers in other subnets. Therefore, we can
describe our scheme for each subnet separately.

The scheme employed over a subnet depends on whether the scheme is generic or reduced and on the parameter $r_r\geq 0$. We describe the different schemes in the following subsections.

\subsubsection{Scheme over a Generic Subnet when $r_r>0$} \label{sec:genpos}
For
simplicity, we assume that the parameters $K,t_\ell,t_r,r_\ell,r_r$ are such that
the first subnet is generic and  describe the scheme for this first
subnet. 

In the special case $r_\ell=2$, $t_\ell=2$, $t_r=1$, and $r_r=1$ the scheme is illustrated in
Figure~\ref{fig:schemeasym}.
\begin{figure*}[htb]
\psfrag{MH1}[bc][bl]{$\hat M_1$}
\psfrag{MH2}[bc][bl]{$\hat M_2$}
\psfrag{MH3}[bc][bl]{$\hat M_3$}
\psfrag{MH4}[bc][bl]{$\hat M_4$}
\psfrag{MH5}[bc][bl]{$\hat M_5$}
\psfrag{MH6}[bc][bl]{$\hat M_6$}
\psfrag{MH7}[bc][bl]{$\hat M_7$}
\psfrag{MH8}{}
\psfrag{MK1}[bc][bl]{$M_1$}
\psfrag{MK2}[bc][bl]{$M_2$}
\psfrag{MK3}[bc][bl]{$M_3$}
\psfrag{MK4}[bc][bl]{$M_4$}
\psfrag{MK5}[bc][bl]{$M_5$}
\psfrag{MK6}[bc][bl]{$M_6$}
\psfrag{MK7}[bc][bl]{$M_7$}
\psfrag{MK8}{}
\psfrag{XK1}{$X_1$}
\psfrag{XK2}{$X_2$}
\psfrag{XK3}{$X_3$}
\psfrag{XK4}{$X_4$}
\psfrag{XK5}{$X_5$}
\psfrag{XK6}{$X_6$}
\psfrag{XK7}{$X_7$}
\psfrag{YK1}[bc][bl]{$Y_1$}
\psfrag{YK2}[bc][bl]{$Y_2$}
\psfrag{YK3}[bc][bl]{$Y_3$}
\psfrag{YK4}[bc][bl]{$Y_4$}
\psfrag{YK5}[bc][bl]{$Y_5$}
\psfrag{YK6}{}
\psfrag{YK7}[bc][bl]{$Y_7$}
\psfrag{YK8}[bc][bl]{$Y_8$}
\psfrag{Receiver5}{Receiver~$5$}
\psfrag{DEC}{\footnotesize{Decoding by}}
\psfrag{IC}{\footnotesize{interference cancellation}}
\psfrag{DEC1}[bc][tc]{\footnotesize{Decoding by}}
\psfrag{IC1}[bl][tl]{\footnotesize{interference cancellation}}
\psfrag{parameters}{}
\psfrag{S}[tc][tr]{\footnotesize{Signal based on}}
\psfrag{TM}[tc][tl]{\footnotesize{transmitter's message only}}
\psfrag{only1}[tc][bl]{}
\psfrag{S2}[tc][tr]{\footnotesize{Signal based on}}
\psfrag{TM2}[tc][tl]{\footnotesize{transmitter's message only}}
\psfrag{only2}[tc][bl]{}
\psfrag{Silenced}{\footnotesize{Silenced}}
\psfrag{DPC1}[bl][bc]{\footnotesize{DPC to cancel}}
\psfrag{DPC}[bl][bl]{\footnotesize{DPC to cancel}}
\psfrag{LS}[bl][bl]{\footnotesize{signals from left}}
\psfrag{RS}[bl][bl]{\footnotesize{signals from right}}
\begin{center}
\includegraphics[scale=0.6]{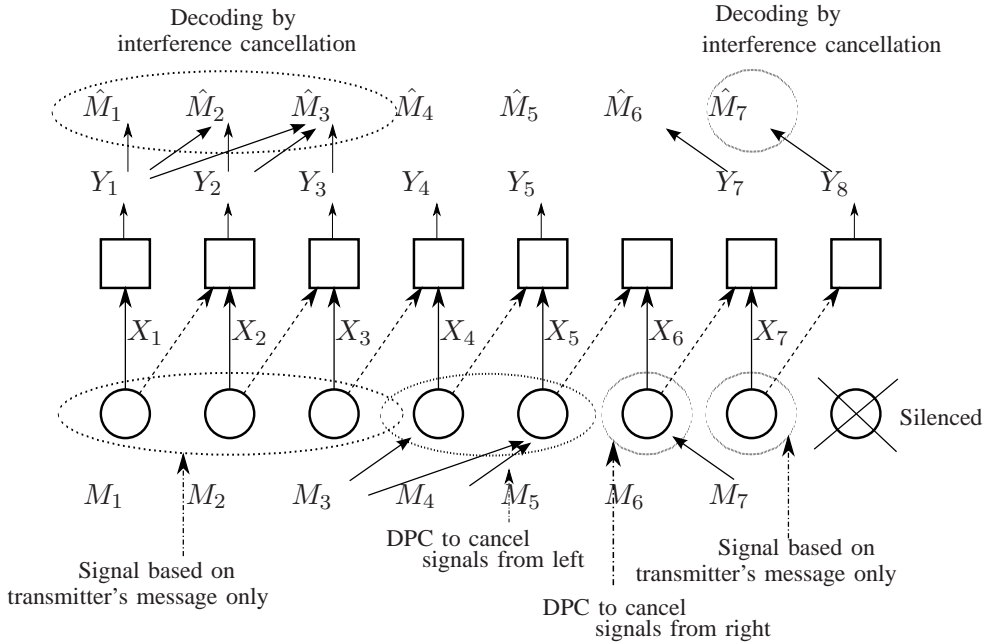}
\vspace{-0.9cm}

\caption{Scheme in a generic subnet for parameters $t_\ell=2$, $t_r=1$, $r_\ell=2$, and $r_r=1$.}
\label{fig:schemeasym}
\end{center}
\end{figure*}
In general, in the first subnet, we wish to transmit
Messages $M_1,\ldots,
M_{r_\ell+t_\ell+t_r+r_r+1}$.
Define the sets (some of which may be empty)
\begin{align*}
\mathcal{G}_1&=\{1,\ldots, r_\ell+1\}\\
\mathcal{G}_2&=\{r_\ell+2,\ldots, r_\ell+t_\ell+1\}\\
\mathcal{G}_3&=\{r_\ell+t_\ell+2,\ldots, r_\ell+t_\ell+t_r+1\}\\
\mathcal{G}_4&=\{r_\ell+t_\ell+t_r+2,\ldots, r_\ell+t_\ell+t_r+r_r+1\}.
\end{align*}

Messages~$\{M_k\}_{k\in\mathcal{G}_1}$ are transmitted as follows.
\begin{itemize}
\item For each $k\in\mathcal{G}_1$ we construct a single-user Gaussian code~$\mathcal{C}_k$ of power $P$, blocklength $n$, and rate $R_k=\frac{1}{2} \log(1+P)$.\footnote{In order to satisfy the block-power constraint imposed on the input sequences, the power of these Gaussian codebooks should be chosen slightly smaller than $P$. Similarly, for the probability of error tending to 0 as $n\to \infty$ the rate $R_k$ should be slightly smaller than $1/2\log(1+P)$. However, these are technicalities which we ignore for readability.} The code $\mathcal{C}_k$ is revealed to Transmitter~$k$ and to Receivers~$k, \ldots, r_\ell+1$.
\item Each Transmitter~$k\in\mathcal{G}_1$ ignores the side-information about other transmitters' messages
  and codes for a Gaussian single-user channel.  That is, it picks the codeword from codebook~$\mathcal{C}_k$ that corresponds to its message $M_k$ and  sends this codeword over the channel.
\item Decoding is performed using \emph{successive interference cancellation}, starting by decoding  Message $M_1$ based on the outputs of the first antenna $Y_1^n$. 

Specifically, each  
  Receiver~$k\in\mathcal{G}_1$ decodes as follows. Let $\hat{X}_0^n$ be an all-zero sequence of length $n$. Receiver~$k$ initializes $j$ to 1 and while $j \leq k$:
  \begin{itemize}
  \item It computes the difference 
  \begin{equation}\label{eq:diffj}
  Y_j^n - \alpha \hat X_{j-1}^n,
  \end{equation} and decodes Message~$M_j$ based on this difference using an optimal ML-decoder.   
 Let $\hat{M}_j$ denote the decoded message.\footnote{Notice that all receivers~$k=j, \ldots, r_\ell+1$ decode Message $M_j$ in the same way, and thus they produce the same estimate $\hat{M}_j$.}
  \item It picks the codeword $x_j^n(\hat{M}_j)$ from codebook $\mathcal{C}_j$ that corresponds to the guess $\hat{M}_j$ and produces this codeword as its reconstruction of the input $\hat{X}_{j}^n$:
  \begin{equation}
  \hat{X}_j^n = x_j^n(\hat{M}_j).
  \end{equation}
  \item It increases the index $j$ by $1$.
  \end{itemize}
  \item    Notice that  each Receiver~$k\in\mathcal{G}_1$ has access to the output signals $Y_1^n,\ldots Y_k^n$ because $k\leq r_\ell+1$, and thus the described decoding can indeed be applied.
    \item For each $k\in\mathcal{G}_1$, if Message $M_{k-1}$ was decoded correctly, i.e., $\hat{M}_{k-1}=M_{k-1}$,  we have 
    \begin{equation}
    Y_k^n- \alpha \hat{X}_{k-1}^n = X_k^n+N_k^n.
    \end{equation}
Thus, in this case, Message $M_k$ is decoded based on the interference-free outputs $X_k^n+N_k^n$, and, by construction of the code $\mathcal{C}_k$, the average probability of error  
   \begin{equation}\label{eq:err1}
   \Pr[\hat{M}_k=M_k] \to 0 \quad \textnormal{ as }\quad n \to \infty.
      \end{equation}

\end{itemize}
If $t_\ell\geq 1$, Messages~$\{M_k\}_{k\in\mathcal{G}_2}$ are transmitted as follows.
\begin{itemize}
\item For each $k\in\mathcal{G}_2$, we construct a dirty-paper code $\mathcal{C}_k$ that is of power~$P$, blocklength $n$, and rate $R_k=\frac{1}{2}\log(1+P)$, and that is designed for noise variance~$1$ and interference variance~$\alpha^2 P$ (which is the variance of $\alpha X_{k-1}$). The code $\mathcal{C}_k$ is revealed to Transmitters~$k,\ldots, r_\ell+t_\ell+1$ and to Receiver~$k$.
\item Each Transmitter~$k\in\mathcal{G}_2$ 
  computes the interference term
  $\alpha\bvect{X}_{k-1}$ and  uses the dirty-paper code $\mathcal{C}_k$ to encode its 
  message $M_k$ and mitigate
  this interference $\alpha\bvect{X}_{k-1}$. It then sends the resulting sequence over the channel.

\item Each Receiver~$k\in\mathcal{G}_2$ ignores all the side-information about other receivers' outputs. It    decodes its desired message $M_k$ solely based on its own outputs
   \begin{equation}\label{eq:outputYk}
   \bvect{Y}_k= \bvect{X}_{k} + \alpha \bvect{X}_{k-1} + \bvect{N}_{k}
   \end{equation}
  applying  dirty-paper decoding with code $\mathcal{C}_k$.
   \item Transmitter~$k\in\mathcal{G}_2$ can compute $\alpha\bvect{X}_{k-1}$ because 
in our scheme 
  $\bvect{X}_{k-1}$ depends only on messages $M_{r_\ell+1},\ldots, M_{k-1}$, and these messages are known  to Transmitter~$k$ because  
  $(k-(r_\ell+1))\leq t_\ell$ for all $k\in\mathcal{G}_2$.
   \item By construction, the sequence $\bvect{X}_{k} $, which encodes Message $M_k$, can perfectly mitigate the interference $\alpha X_{k-1}^n$, and    the average probability of error 
      \begin{equation}\label{eq:err2}
   \Pr[\hat{M}_k=M_k] \to 0 \quad \textnormal{ as } \quad n \to \infty.
      \end{equation}
   
\end{itemize}

If $t_r\geq 1$, Messages~$\{M_{k}\}_{k\in\mathcal{G}_3}$ are transmitted as follows.
\begin{itemize}
\item For each $k\in\mathcal{G}_3$, we construct a dirty-paper code $\mathcal{C}_k$ of power $\alpha^2 P$ (the power of $\alpha X_k$), blocklength $n$, and rate $R_k=\frac{1}{2} \log(1+\alpha^2 P)$, and that is designed for noise variance $1$ and interference variance $P$ (the variance of $X_{k+1}^n$). The code $\mathcal{C}_k$ is revealed to Transmitters~$r_\ell+t_\ell+2, \ldots, k$ and to Receiver~$k$.
\item
  Each Receiver
  $k\in\mathcal{G}_3$ decodes its desired message $M_k$ based on the outputs of the antenna to its right
  \begin{equation}\label{eq:Ykone}
  \bvect{Y}_{k+1}=\bvect{X}_{k+1}+ \alpha \bvect{{X}}_k +  \bvect{N}_{k+1},
  \end{equation} to which it has access because $r_r\geq 1$. The exact decoding procedure is explained shortly.
    
 \item  Each Transmitter~$k\in\mathcal{G}_3$ 
  computes the ``interference" sequence
  $\bvect{X}_{k+1}$ and applies the dirty-paper code $\mathcal{C}_k$ to encode Message $M_k$ and mitigate this ``interference" $\bvect{X}_{k+1}$. Denoting the produced sequence  by  $\tilde{X}_k^n$, Transmitter~$k$ sends 
 \begin{equation} \label{eq:sentsequence}
 \bvect{X}_k=\frac{1}{\alpha}\bvect{\tilde{X}}_k.
 \end{equation}

(The scaling by $1/\alpha$ in \eqref{eq:sentsequence} reverses the amplification by $\alpha$ the
  sequence $\bvect{X}_k$ experiences on its path to  Receiver~$(k+1)$, see~\eqref{eq:Ykone}.)

\item Each Receiver~$k\in\mathcal{G}_3$  
applies the dirty-paper decoding of code $\mathcal{C}_{k}$ to the outputs 
  \begin{IEEEeqnarray}{rCl}
\bvect{Y}_{k+1}& =& \alpha X_{k}^n + X_{k+1}^n +N_{k+1}^n \\
& = & \bvect{\tilde{X}}_k + \bvect{X}_{k+1}+\bvect{N}_{k+1}.\label{eq:Ykone2}
   \end{IEEEeqnarray} 
   
   \item   Notice that Transmitter~$k\in\mathcal{G}_3$  can compute the "interference" $\bvect{X}_{k+1}$ non-causally, because this latter
   only depends on messages $M_{k+1},\ldots,$
  $M_{r_\ell+t_\ell+t_r+2}$ which are known to Transmitter~$k$.
  
  Also, by construction of the code $\mathcal{C}_k$, the sequence $\bvect{\tilde{X}}_k$ is  average block-power constrained to $\alpha^2 P$ and thus, by~\eqref{eq:sentsequence},  the transmitted sequence $\bvect{X}_k$   is average block-power constrained to $P$.
  
   \item By construction,  the sequence $ \bvect{\tilde{X}}_k $, which encodes  Message $M_k$, can perfectly mitigate the ``interference" $ X_{k+1}^n$, and
   the average probability of error  
      \begin{equation}\label{eq:err3}
   \Pr[\hat{M}_k=M_k] \to 0 \quad \textnormal{ as } \quad n \to \infty.
      \end{equation}
\end{itemize}

Messages $\{M_{k}\}_{k\in\mathcal{G}_4}$ are transmitted as follows.
\begin{itemize}
\item For each $k\in\mathcal{G}_4$, we  construct a single-user Gaussian codebook $\mathcal{C}_k$ of power $\alpha^2 P$, blocklength $n$, and rate $R_k=\frac{1}{2} \log(1+\alpha^2 P)$. The codebook $\mathcal{C}_k$ is revealed to Transmitter~$k$ and to Receivers~$k, \ldots, r_\ell+t_\ell+t_r+r_r+1$.
\item Each Transmitter~$k\in\mathcal{G}_4$ ignores the side-information about other transmitters' messages
  and codes for a Gaussian single-user channel.  That is, it picks the codeword from  code~$\mathcal{C}_k$ that corresponds to its message $M_k$ and sends this codeword over the channel.
\item Decoding is performed using {successive interference cancellation}, starting by decoding Message~$M_{r_\ell+t_\ell+t_r+r_r+1}$ based on the outputs of the last antenna $Y_{r_\ell+t_\ell+t_r+r_r+2}^n$. 

Specifically, Receiver~$k\in\mathcal{G}_4$ decodes its desired Message $M_{k}$ as follows.  
 Let $\hat{X}_{r_\ell+t_\ell+t_r+r_r+3}^n$ be an all-zero sequence of length $n$. 

Receiver~$k$ initializes $j$ to $r_\ell+t_\ell+t_r+r_r+1$, and  while $j \geq k$:
  \begin{itemize}
  \item It computes the difference 
  \begin{equation}\label{eq:diffjone}
  Y_{j+1}^n -  \hat X_{j+1}^n,
  \end{equation} and decodes Message~$M_j$ based on this difference using an optimal ML-decoder. 
  
 Let $\hat{M}_j$ denote the resulting guess of Message $M_j$.
  \item It reconstructs the input sequence ${X}_j^n$ by picking the codeword $x_j^n(\hat{M}_j)$ from codebook $\mathcal{C}_j$ that corresponds to Message $\hat{M}_j$:
  \begin{equation}
  \hat{X}_j^n = x_j^n(\hat{M}_j).
  \end{equation}
  \item It decreases $j$ by $1$.
  \end{itemize}

\item Notice that Receiver~$k\in\mathcal{G}_4$ has access to the output signals $Y_k^n,\ldots Y_{r_\ell+t_\ell+t_r+r_r+2}^n$ because $k\geq  r_\ell+t_\ell+t_r+2$.
    \item For each $k \in\mathcal{G}_4$, if the previous message $M_{k-1}$ has been decoded correctly, i.e, $\hat{M}_{k-1}=M_{k-1}$,  we have 
    \begin{equation}
    Y_{k+1}^n-  \hat{X}_{k+1}^n =\alpha X_k^n+N_{k+1}^{n}.
    \end{equation}
    Thus, in this case, Message $M_k$ is decoded based on the interference-free outputs $\alpha X_k^n+N_{k+1}^{n}$, and, by construction of the code $\mathcal{C}_k$, the
   average probability of error 
   \begin{equation}\label{eq:err4}
   \Pr[\hat{M}_k=M_k] \to 0 \quad \textnormal{ as }\quad n \to \infty.
      \end{equation}
\end{itemize}
To summarize, in the described scheme we sent messages $M_1,\ldots, M_{r_\ell+t_\ell+t_r+r_r+1}$ with vanishingly small average probability of error, see~\eqref{eq:err1}, \eqref{eq:err2}, \eqref{eq:err3}, and \eqref{eq:err4}, and at rates
\begin{IEEEeqnarray}{rCl}
R_1 = \ldots = R_{r_\ell+t_\ell+1} &=&\frac{1}{2}\log(1+P) \\
R_{r_\ell+t_\ell+2} = \ldots= R_{r_\ell+t_\ell+t_r+r_r+1} & = & \frac{1}{2}\log(1+\alpha^2P). \IEEEeqnarraynumspace
\end{IEEEeqnarray}

\begin{conclusion}\label{i2equal0}
Our scheme for $r_r\geq 0$  achieves a multiplexing gain of $(t_\ell+r_\ell+r_r+t_r+1)$ over  a generic subnet. It uses all  $(t_\ell+r_\ell+r_r+t_r+1)$ active transmit antennas of the subnet and all $(t_\ell+r_\ell+r_r+t_r+2)$ receive antennas.
 \end{conclusion}

\subsubsection{Scheme over a Generic Subnet when $r_r=0$} \label{sec:genzero}
We again assume that the first subnet is generic and focus on this first subnet.
When $r_r=0$  we  
transmit Messages $M_{1},\ldots, M_{r_\ell+t_\ell+1}$ and
$M_{r_\ell+t_\ell+3},\ldots, M_{r_\ell+t_\ell+t_r+2}$ over the first subnet. 

Messages~$\{M_k\}_{k\in (\mathcal{G}_1\cup \mathcal{G}_2)}$ are transmitted in the same way as in the previous section~\ref{sec:genpos}.  
Messages~$\{M_{k+1}\}_{k\in \mathcal{G}_3}$ are transmitted in a similar way as  Messages~$\{M_k\}_{k\in\mathcal{G}_3}$ in the previous section~\ref{sec:genpos}, except that now each Transmitter~$k\in\mathcal{G}_3$ sends Message $M_{k+1}$ (as opposed to Message $M_k$) and accordingly, each output sequence $Y_{k+1}^n$ is used by Receiver~$k+1$ to decode Message $M_{k+1}$ (as opposed to  Receiver~$k$ decoding Message~$M_k$ based on $Y_{k+1}^n$). More specifically: 
\begin{itemize}
\item For each $k\in\mathcal{G}_3$, we construct a dirty-paper code $\mathcal{C}_{k+1}$ that is of power $\alpha^2 P$ (the power of $\alpha X_k$), blocklength $n$, and rate $R_{k+1}=\frac{1}{2} \log(1+\alpha^2 P)$, and that is designed for noise variance $1$ and interference variance $P$ (the variance of $X_{k+1}^n$). The code $\mathcal{C}_{k+1}$ is revealed to Transmitters~$r_\ell+t_\ell+2,\ldots, k$ and to Receiver~$k+1$.
 \item   
 Transmitter~$k\in \mathcal{G}_3$ applies the dirty-paper code $\mathcal{C}_{k+1}$ to encode Message $M_{k+1}$ and mitigate the ``interference" $\bvect{X}_{k+1}$. 
 Denoting the sequence produced by the dirty-paper code by $\tilde{X}_k^n$, Transmitter~$k$ sends
 \begin{equation}
 \bvect{X}_k=\frac{1}{\alpha}\bvect{\tilde{X}}_k.
 \end{equation}

\item Each Receiver~$k+1$, for $k\in\mathcal{G}_3$, ignores its side-information about outputs observed at other antennas. It decodes its desired Message~$M_{k+1}$ solely
based on the outputs at its own antenna
   \begin{IEEEeqnarray}{rCl}
\bvect{Y}_{k+1}&=& \alpha X_k^n + X_{k+1}^n + N_{k+1}^n \\
& = &\bvect{\tilde{X}}_k + \bvect{X}_{k+1}+\bvect{N}_{k+1}
   \end{IEEEeqnarray}
using the dirty-paper  decoding of  code $\mathcal{C}_{k+1}$.
   
  \item Notice that Transmitter~$k\in\mathcal{G}_3$ can
  compute the ``interference" sequence $X_{k+1}^n$ because this latter 
   only depends on messages $M_{k+2},\ldots,$
  $M_{r_\ell+t_\ell+t_r+2}$ which are known to Transmitter~$k$.
   \item By construction, the sequence $\tilde{X}_k^n$, which encodes Message $M_{k+1}$, can completely mitigate the ``interference" $X_{k+1}^n$, and     the average probability of error 
   \begin{equation}\label{eq:err5}
   \Pr[\hat{M}_{k+1}\neq M_{k+1}] \to 0 \quad \textnormal{ as } \quad n \to \infty.
   \end{equation}
\end{itemize}
To summarize, in the described scheme we transmit Messages $M_1,\ldots, M_{r_\ell+t_\ell+1}$ and $M_{r_\ell+t_\ell+3},\ldots, M_{r_\ell+t_\ell+t_r+2}$ with vanishingly small average probability of error, see~\eqref{eq:err1}, \eqref{eq:err2}, and \eqref{eq:err5}, and  at rates 
\begin{IEEEeqnarray}{rCl}
R_1 = \ldots = R_{r_\ell+t_\ell+1} &=&\frac{1}{2}\log(1+P) \\
R_{r_\ell+t_\ell+3} = \ldots= R_{r_\ell+t_\ell+t_r+2} & = & \frac{1}{2}\log(1+\alpha^2P). \IEEEeqnarraynumspace
\end{IEEEeqnarray}

\begin{conclusion}\label{i2equal0r}
Our scheme for $r_r =0$ and $t_r\geq 1$ achieves a multiplexing gain of $(r_\ell+t_\ell+t_r+1)$ over a generic subnet. If $t_r\geq 1$, it uses all  $(r_\ell+t_\ell+t_r+r_r+1)$ active transmit antennas and all $(r_\ell+t_\ell+t_r+2)$ receive antennas of the subnet. If $t_r=0$ it uses all  $(r_\ell+t_\ell+1)$ active transmit antennas; but it only uses the first $(r_\ell+t_\ell+1)$  receive antennas  and ignores the last antenna of the subnet. 
 \end{conclusion}
 
 \subsubsection{Scheme over a Reduced Subnet}

Let
\begin{subequations}\label{eq:param}
 \begin{align}
 r_\ell'&\triangleq \min\left[\left(\kappa-1\right),r_\ell\right]\\
 t_\ell'&\triangleq \min\left[\left(\kappa-r_\ell-1\right)_+,t_\ell\right]\\
 t_r'&\triangleq \min\left[\left(\kappa-r_\ell-t_\ell-2\right)_+,t_r\right]\\
 r_r'&\triangleq \min\left[\left(\kappa-r_\ell-t_\ell-t_r-2\right)_+,r_r\right]
 \end{align}
 \end{subequations}
 where $(x)_{+}$ is defined as $\max\{x,0\}$. 
 In  a reduced subnet we apply one of the two schemes described for the generic subnet but now with  reduced side-information parameters $r_\ell', t_\ell', t_r', r_r'$. If $r_r'>0$,  we apply the scheme in Subsection~\ref{sec:genpos} otherwise we apply the scheme in Subsection~\ref{sec:genzero}.
 Notice that, by definition, $r_\ell' \leq r_\ell$, $t_\ell' \leq t_\ell$, $t_r \leq t_r'$, and $r_r' \leq r_r$, and thus the transmitters and receivers have enough side-information to apply the described scheme with these parameters.
  
When $\kappa \leq (t_\ell+r_\ell+1)$, then the reduced subnet consists of an equal number $\kappa$ of active transmit and receive antennas because  the last transmit antenna has not been silenced. In this case, also $t_r'=r_r'=0$ and by Conclusion~\ref{i2equal0r}, the  scheme in Subsection~\ref{sec:genzero}  achieves multiplexing gain $\kappa$ over such a subnet.
 
 When $\kappa > (t_\ell+r_\ell+1)$, the subnet consists of $\kappa-1$ active transmit antennas and $\kappa$ receive antennas. By Conclusions~\ref{i2equal0} and \ref{i2equal0r}, one of the schemes in Subsections~\ref{sec:genpos} or \ref{sec:genzero} achieves multiplexing gain  $\kappa-1$ over such a subnet. 
 
 To summarize, we  achieve a multiplexing gain of
 \begin{equation}\label{eq:redmg}
 \begin{cases} \kappa, & \textnormal{ if } \kappa \leq t_\ell+r_\ell+1\\
  \kappa-1, & \textnormal{ if } \kappa > t_\ell+r_\ell+1
  \end{cases}
 \end{equation}
 over a reduced subnet of size $\kappa$.
 \subsubsection{Performance Analysis over the Entire Network}
  
Over the first $\lfloor K/\beta\rfloor$ generic subnets we  achieve a multiplexing gain of $\beta-1$ and, if it exists, then over the last reduced subnet we achieve a multiplexing gain of either $\kappa$ or $\kappa-1$, see~\eqref{eq:redmg}. Over the entire network we thus achieve a multiplexing gain of 
 \begin{equation}
 K-\gamma=  \begin{cases} K-\lfloor K/\beta\rfloor,   & \textnormal{ if } \kappa \leq t_\ell+r_\ell+1\\
  K-\lfloor K/\beta\rfloor-1,& \textnormal{ if } \kappa > t_\ell+r_\ell+1.
  \end{cases}
 \end{equation}
 
This proves the desired lower bound.
 
\begin{remark}\label{eq:cyclo}
In the described scheme a subset of $\gamma$ messages is completely
ignored and not sent over the network. Using time-sharing we can obtain a fair scheme that  
sends all messages  at almost equal rates and achieves a multiplexing gain of at least $K-\gamma-1$.  
More specifically, the idea is to time-share $\beta$
schemes where in 
the $i$-th scheme, $i\in\{1,\ldots, \beta\}$, we silence transmitters $\{i+ j
\beta\}_{j\in \left \{1,\ldots, ,\left \lfloor \frac{K-i}{\beta}
    \right \rfloor\right\}}$, and if $(K \mod \beta) \geq
(i+t_\ell+r_\ell+1)$, then we also silence the last transmitter~$K$. This splits the network into $\gamma$ or $\gamma+1$ subnets:
a possibly reduced first subnet, $\gamma-2$ or  $\gamma-1$ generic
subnets, and a possibly reduced last subnet. In each of the subnets, depending on whether it is generic or reduced, one of the schemes described above   is used.
\end{remark}

\mw{
\subsection{Converse to Theorem~\ref{th:mg}}\label{sec:convTh1}

Apply the Dynamic-MAC Lemma~\ref{lem:dynMAC} to the following choices: 
\begin{itemize}
\item $q=1$;
\item $g=\gamma-1$;
\item  $\mathcal{A} \triangleq \bigcup_{m=0}^{g} \mathcal{A}(m)$,
where 
 for $m=0,\ldots, g-1$,
\begin{equation*}
\mathcal{A}(m)\triangleq \{m\beta + r_\ell+2 , \ldots,(m+1)\beta-r_r\}
\end{equation*}
and
\begin{equation*}
\mathcal{A}(g)\triangleq  \{g\beta + r_\ell+2 , \ldots,K\}.
\end{equation*}
\item $\mathcal{B}_1 \triangleq \mathcal{K} \setminus \mathcal{A}$; 

\item  genie-information
 \begin{IEEEeqnarray}{rCl}\label{gen0}
\vect{V}_{0}& \triangleq & \vect{N}_1 + \sum_{\nu=1}^{r_\ell+t_\ell+1}
\left(-\frac{1}{\alpha}\right)^\nu \vect{N}_{1+ \nu},
\end{IEEEeqnarray}
and, for $m\in
\{1,\ldots,g\}$:
\begin{IEEEeqnarray}{rCl}\label{gen1}
\vect{V}_{m}& \triangleq &\vect{N}_{1+m\beta}+\sum_{\nu=1}^{r_\ell+t_\ell+1}
\left(-\frac{1}{\alpha}\right)^\nu \vect{N}_{1+m\beta+\nu}
\nonumber \\ && +  \sum_{\nu=1}^{t_r+r_r} \left(-\alpha\right)^{\nu}
  \vect{N}_{1+m\beta - \nu}.
\end{IEEEeqnarray}
\end{itemize}
Notice that by our choice of $\mathcal{A}$, the set difference
\begin{equation}\label{eq:inter}
\mathcal{K} \backslash  \mathcal{R}_{\mathcal{A}} = \{1 + m\beta\}_{m=0}^g.
\end{equation}
Since for each $m=0,\ldots, g$ the genie-information $\vect{V}_m$ contains an additive noise term $N_{1+m\beta}$, which is not present in all other genie-informations $\{\vect{V}_{m'}\}_{m'\neq m}$,  \eqref{eq:inter} implies that the differential entropy in~\eqref{eq:assum_noise} is finite. Moreover, the differential entropy does not depend on $P$ because neither does the genie-information. 
 In the following, we show that also the second assumption~\eqref{eq:assum_fun} of Lemma~\ref{lem:dynMAC} is satisfied and that thus we can apply the lemma for the described choice. This then proves the desired converse because, by~\eqref{eq:inter}, $|\mathcal{R}_{\mathcal{A}}|=K-g-1=K-\gamma$.  
 
 By~\eqref{eq:inter}, the set $\{M_{k}\}_{k\in\mathcal{A}}$ 
  includes all messages $\{{M}_{r_\ell+2+\nu+m\beta}\}_{\substack{0\leq \nu\leq t_\ell+t_r\\0\leq m \leq
  \gamma-1}} $, where out of range indices should be ignored. From $\{M_{k}\}_{k\in\mathcal{A}}$ it is thus  possible to reconstruct 
 the input sequences $\{\vect{X}_{t_\ell+r_\ell+2+m\beta}\}_{m=0}^{ g}$:
\begin{IEEEeqnarray*}{rCl}
\lefteqn{\vect{{X}}_{r_\ell+t_\ell+2+m \beta}} \quad \\
&=&f^{(n)}_{r_\ell+t_\ell+2+m\beta}({M}_{r_\ell+2+m
  \beta}, \ldots, {M}_{r_\ell+t_\ell+t_r+2+m\beta}).
\end{IEEEeqnarray*}
Using these reconstructed sequences, the output sequences
observed at the receivers in Group~$\mathcal{A}$, and the genie-information
$\{\vect{V}_{m}\}_{m=0}^{g}$,
it is then possible to reconstruct all channel outputs not observed by the receivers in Group~$\mathcal{A}$, \eqref{eq:inter}:
\begin{IEEEeqnarray*}{rCl}\vect{Y}_1&=& 
 -\sum_{\nu=1}^{r_\ell+t_\ell+1}
  \left(-\frac{1}{\alpha}\right)^{\nu} \bfY_{1+\nu} \nonumber \\ && +\left(-\frac{1}{\alpha}\right)^{r_\ell+t_\ell+1}
 \vect{{X}}_{r_\ell+t_\ell+2}  
  + \vect{V}_0
\end{IEEEeqnarray*}
and, for $m\in\{1,\ldots, g\}$:
\begin{IEEEeqnarray*}{rCl}
\lefteqn{\vect{{Y}}_{1+m\beta}}\\
& = &
- \sum_{\nu=1}^{r_\ell+t_\ell+1} 
 \left(-\frac{1}{\alpha}\right)^{\nu} \bfY_{1+m\beta+\nu} -\sum_{\nu=1}^{t_r+r_r} \left( -\alpha\right)^{\nu}
\vect{Y}_{1+m \beta -\nu} \\ &&+\left(-\frac{1}{\alpha}\right)^{r_\ell+t_\ell+1}
 \vect{{X}}_{r_\ell+t_\ell+2+m \beta}\nonumber\\
 & & -
 \left(-\alpha \right)^{t_r+r_r+1}
\vect{X}_{r_\ell+t_\ell+2+(m-1)\beta}
+ \vect{V}_{m}. 
\end{IEEEeqnarray*}
This establishes that Assumption~\eqref{eq:assum_fun} holds, and concludes the proof.
}
 
 
\section{Achievability Proof of Theorem~\ref{th:symsym}}\label{sec:symach}

For each of the four lower bounds \ref{e1})--\ref{e4}) in Theorem~\ref{th:symsym}, i.e., Inequalities \eqref{eq:S1}--\eqref{eq:S3},  we present a scheme achieving this lower bound.
The four schemes are similar: they all rely on the idea of switching off 
some of the transmitter/receiver pairs, and on using the strategy over the resulting subnets. (Here, by silencing transmitter/receiver pairs we intend that we silence the antennas at the transmitters and ignore the corresponding antennas at the receivers.) This splits the networks into non-interfering subnets.
In each scheme we silence a different set of transmitter/receiver pairs. As we will see we do this in a way that splits the network into subnets that have at most $t_\ell+r_\ell+1$ active transmitter/receiver pairs.

We first describe the strategy used to communicate over the subnets (Section~\ref{sec:subnetscheme}). Then, we present the set of transmitter/receiver pairs that needs to be silenced in each of the four schemes, so that they achieve the lower bounds in \ref{e1})--\ref{e4}) (Sections~\ref{sec:S1}--\ref{sec:S3}).

\subsection{Strategy used in the Subnets}\label{sec:subnetscheme}

Consider a subnet with $\kappa$ transmitter/receiver pairs, where  
$\kappa\leq t_\ell+r_\ell+1$.
We first present a coding strategy that achieves multiplexing gain $\textnormal{rank}(\mat{H}_{\kappa}(\alpha))$ when 
\begin{equation}\label{eq:assumK}
\kappa=t_\ell +r_\ell+1.
\end{equation}
Then we describe how to modify this strategy to achieve a multiplexing gain of $\textnormal{rank}(\mat{H}_{\kappa}(\alpha))$ when $\kappa< t_\ell+r_\ell+1$.

Depending on which of the following three cases applies, we use a different scheme to communicate over the subnet. 
\begin{itemize}
\item[1.)] If the transmitters and the receivers have the same amount of side-information: 
\begin{equation}r_\ell +r_r=t_\ell+t_r \label{ass:1}\end{equation}
we use Multi-Input/Multi-Output (MIMO)  point-to-point scheme.
\item[2.)] If the transmitters have more side-information than the receivers: \begin{equation} r_\ell+ r_r<t_\ell +t_r\label{ass:2}\end{equation}
we use a MIMO broadcast scheme.
\item[3.)] If the receivers have more side-information than the transmitters: \begin{equation}r_\ell+r_r >t_\ell +t_r\label{ass:3}\end{equation}
we use a MIMO multi-access scheme.
\end{itemize}

We first describe the MIMO point-to-point scheme for case 1.). In this case
\eqref{eq:equal} and \eqref{ass:1} imply that 
\begin{equation}
t_\ell=r_r\quad \textnormal{and} \quad t_r =r_\ell.
\end{equation}
Therefore, since  $\kappa=r_\ell+t_\ell+1$, \eqref{eq:assumK}, all $\kappa$ transmitters are cognizant of Message $M_{t_r+1}$ and Receiver~$(t_r+1)$ has access to 
all $\kappa$ antennas in the subnet. Thus, all the transmitters can act as a single transmitter that transmits Message $M_{t_r+1}$ to Receiver~$(t_r+1)$ which can decode the Message based on all the antennas in the subnet.  Using an optimal MIMO point-to-point scheme for this transmission achieves a multiplexing gain of
$\textnormal{rank}(\mat{H}_{\kappa}(\alpha))$ over the subnet.

We next describe the MIMO broadcast scheme for case 2.).  Notice that~\eqref{eq:equal} and \eqref{ass:2} imply that 
\begin{equation}\label{eq:delta1}
r_\ell < t_r.
\end{equation}
By \eqref{eq:equal} and \eqref{eq:assumK}, all the transmitters are cognizant of Messages 
$M_{r_\ell+1},\ldots, M_{t_{r}+1}$  and Receivers $(r_\ell+1),
\ldots, (t_r+1)$ \emph{jointly} have access to all the
$\kappa$ antennas in the subnet.  Thus, all
the transmitters in the subnet can act as a big common transmitter that transmits Messages
$M_{r_\ell+1},\ldots, M_{t_{r}+1}$ to the independent Receivers
$(r_\ell+1), \ldots, (t_r+1)$. 
where Receiver~$(r_\ell+1)$ decodes based  on antennas $1, \ldots, r_\ell+1$ (and ignores the other antennas), Receivers~$(r_\ell+2),
\ldots, t_\ell$ decode based only on their own
antennas, and Receiver~$(t_r+1)$ decodes based  on antennas $t_r+1, \ldots, t_r+r_r+1$.\footnote{Notice that
  the described assignment of antennas to receivers is only one
  possible assignment that leads to the desired multiplexing gain. Other assignments are possible.} Using an  optimal MIMO broadcast scheme for this transmission we can  
  achieve a multiplexing gain of $\textnormal{rank}(\mat{H}_{\kappa}(\alpha))$ over the subnet.
  
For parameters $t_\ell=2$, $t_r=3$, $r_\ell=1$, and $r_r=0$ the scheme is illustrated in Figure~\ref{fig:BC}. %
\begin{figure}[htb]
\psfrag{MH1}{}
\psfrag{MH2}[bc][bl]{$\hat M_2$}
\psfrag{MH3}[bc][bl]{$\hat M_3$}
\psfrag{MH4}[bc][bl]{$\hat M_4$}

\psfrag{MK1}{}
\psfrag{MK2}[bc][bl]{$M_2$}
\psfrag{MK3}[bc][bl]{$M_3$}
\psfrag{MK4}[bc][bl]{$M_4$}

\psfrag{XK1}{$X_1$}
\psfrag{XK2}{$X_2$}
\psfrag{XK3}{$X_3$}
\psfrag{XK4}{$X_4$}

\psfrag{YK1}[cc][cl]{$Y_1$}
\psfrag{YK2}[cc][cl]{$Y_2$}
\psfrag{YK3}[cc][cl]{$Y_3$}
\psfrag{YK4}[cc][cl]{$Y_4$}
\psfrag{parameters}{}
\psfrag{MK}{\footnotesize{Messages known to all}}
\psfrag{T}{\footnotesize{transmitters}}
\begin{center}
\includegraphics[scale=0.60]{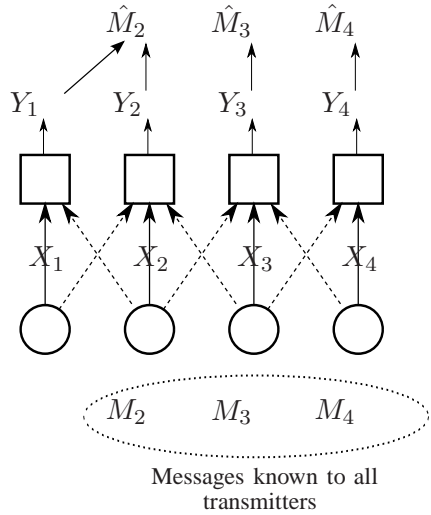}
\vspace{-1cm}

\caption{Broadcast scheme employed in a subnet for parameters $\kappa=4$, $t_\ell=2$, $t_r=3$, $r_\ell=1$, and $r_r=0$}
\label{fig:BC}
\end{center}
\end{figure}

We finally describe the MIMO multi-access scheme for case 3.). Notice that~\eqref{eq:equal} and \eqref{ass:3} imply that 
\begin{equation}\label{eq:delta2}
t_r < r_\ell.
\end{equation}
By \eqref{eq:equal} and \eqref{eq:assumK}, each transmitter knows at least one of the Messages $M_{t_r+1}, \ldots, M_{r_\ell+1}$, and Receivers $(t_r+1),\ldots,(r_\ell+1)$ all have access to all $\kappa$ receive antennas in the subnet. In our scheme the first $t_r+1$ transmitters $1,\ldots, t_r+1$ act as a big common transmitter that transmits Message $M_{t_r+1}$. Similarly, the last $t_\ell+1$ transmitters $r_\ell+1,\cdots,r_\ell+t_\ell+1$ act as a big common transmitter that transmits Message $M_{r_\ell+1}$.  Transmitters $t_r+2, \ldots, r_\ell$ act as single transmitters that transmit their own messages.  
Receivers~$(t_r+1),\ldots,(r_\ell+1)$ act as a single big common receiver that decodes Messages $M_{r_\ell+1},\ldots, M_{r_\ell+t_\ell+1}$ based on all the antennas in the network.
Applying an optimal MIMO MAC scheme for this transmission achieves multiplexing gain
 $\textnormal{rank}(\mat{H}_{\kappa}(\alpha))$ over the subnet.
 
 For  parameters  $t_\ell=2$, $t_r=0$, $r_\ell=1$, and $r_r=3$ the scheme is illustrated in Figure~\ref{fig:MAC}.
\begin{figure}[htb]
\psfrag{MK1}[tc][tl]{$M_1$}
\psfrag{MK2}[tc][tl]{$M_2$}
\psfrag{MK3}{}
\psfrag{MK4}{}

\psfrag{MH1}[cc][cl]{$\hat M_1$}
\psfrag{MH2}[cc][cl]{$\hat M_2$}
\psfrag{MH3}{}
\psfrag{MH4}{}

\psfrag{XK1}{$X_1$}
\psfrag{XK2}{$X_2$}
\psfrag{XK3}{$X_3$}
\psfrag{XK4}{$X_4$}

\psfrag{YK1}[cc][cl]{$Y_1$}
\psfrag{YK2}[cc][cl]{$Y_2$}
\psfrag{YK3}[cc][cl]{$Y_3$}
\psfrag{YK4}[cc][cl]{$Y_4$}

\psfrag{parameters}{$\kappa=4$, $t_\ell=2$, $t_r=0$, $r_\ell=1$, and $r_r=3$}
\psfrag{AT}{\footnotesize{at least one of those}}
\psfrag{ET}{\footnotesize{Each transmitter knows}}
\psfrag{TA}{\footnotesize{They have access to all}}
\psfrag{An}{\footnotesize{the antennas}}
\psfrag{Decoding by interference cancellation}{\footnotesize{Decoding by interference cancellation}}
\psfrag{parameters}{}
\psfrag{Signal based on the transmitter's message only}{ \footnotesize{ Signal based on the transmitter's message only}}
\psfrag{Silenced}{\footnotesize{ Silenced}}
\psfrag{DPC to cancel the signal from the left}{\footnotesize{ DPC to cancel the signal from the left}}
\psfrag{DPC to cancel the signal from the right}{\footnotesize{ DPC to cancel the signal from the right}}
\psfrag{Receivingantennas}[l][l]{\footnotesize{Receive antennas (with AWGN)}}
\psfrag{Transmitters}{\footnotesize{Transmitters}}
\begin{center}
\includegraphics[scale=0.6]{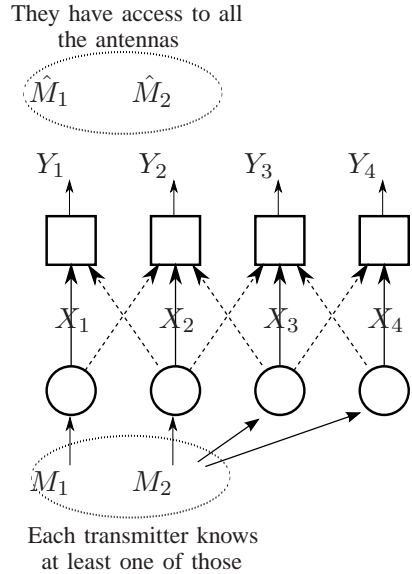}
\vspace{-1.7cm}

\caption{Multi-access scheme employed  in a subnets for parameters $\kappa=4$, $t_\ell=2$, $t_r=0$, $r_\ell=1$, and $r_r=3$.}
\label{fig:MAC}
\end{center}
\end{figure}

We conclude that with the above described schemes we can achieve a multiplexing gain of  $\textnormal{rank}(\mat{H}_{\kappa}(\alpha))$ when $\kappa=t_\ell+r_\ell+1$,  irrespective of the specific values of $t_\ell$ and $r_\ell$. 

We now consider the case where 
\begin{equation}\label{eq:qq}\kappa< t_\ell+r_\ell+1.
\end{equation} 
In this case  
we choose parameters  $t_\ell'\leq t_\ell$, $t_r'\leq t_r$, $r_\ell'\leq r_\ell$, and $r_r'\leq r_r$ such that 
\begin{equation}\kappa=t_\ell'+r_\ell'+1=t_r'+r_r'+1,
\end{equation} and depending on the choice of $t_\ell', t_r', r_\ell', r_r'$ we  apply one of the  three schemes above. This way, we achieve multiplexing gain $\textnormal{rank}(\mat{H}_{\kappa}(\alpha))$ over the subnet also when \eqref{eq:qq} holds.

We obtain the following proposition. 
\begin{proposition}\label{prop5}
For every subnet with $\kappa \leq t_\ell+r_\ell+1$ transmitter/receiver pairs one of the three schemes described above acheives a multiplexing gain of  $\textnormal{rank}(\mat{H}_{\kappa}(\alpha))$.
\end{proposition}

\subsection{Auxiliary Results}\label{sec:aux}
The following auxiliary results will be used in the proofs ahead.

\begin{lemma}\label{lem:alpha}
Let a real number $\alpha$ and a positive integer $p$ be given such that $\deta{\mat{H}_{p}(\alpha)}=0$. Then, the following statements hold.
\begin{enumerate}
\item The integer $p\geq 2$.
\item The determinants $\deta{\mat{H}_{p-1}(\alpha)}$, $\deta{\mat{H}_{p+1}(\alpha)}$, and $\deta{\mat{H}_{p+2}(\alpha)}$ are all non-zero. Moreover, if $p>2$ (and thus $\mat{H}_{p-2}(\alpha)$ is defined) also $\deta{\mat{H}_{p-2}(\alpha)}$ is non-zero. 
\end{enumerate}
\end{lemma}
\begin{proof}
See Appendix~\ref{sec:alpha}. 
\end{proof}
\mw{This lemma generalizes to nonequal nonzero cross-gains in the following way. 
For each positive integer $p\leq K$, let $\mat{H}_{\textnormal{Net},p}$ denote the $p$-th principal minor of $\mat{H}_{\textnormal{Net},p}$. Then, Lemma~\ref{lem:alpha} remains valid if the matrices $\mat{H}_q(\alpha)$ are replaced by  $\mat{H}_{\textnormal{gen},q}$ for $q\in\{p-2, p-1, p, p+1, p\}$. This can be verified by inspecting the proof. (The main change concerns~\eqref{eq:detrecursion}, where $\alpha^2$ needs to be replaced by the product $\alpha_{k,\ell}\cdot \alpha_{k-1,r}$, for some $k\in\mathcal{K}$, which by assumption is again nonzero. All  other steps remain unchanged.) Therefore, the lemma can also be used to extend our results to nonequal nonzero cross-gains.}
\begin{corollary}\label{lem:alpha2}
For every real number $\alpha$ and positive integer $p$, the rank of the matrix $\mat{H}_{p}(\alpha)$ is either $p$ or $p-1$.
\end{corollary}
\begin{proof}
Follows by noting that $H_{p-1}(\alpha)$ is a sub-matrix of $H_p(\alpha)$ and by Lemma~\ref{lem:alpha}.
\end{proof}

\subsection{ Achieving the Lower Bound in~\eqref{eq:S1}}\label{sec:S1}

Recall that \eqref{eq:S1} holds under the assumption that $K\leq t_\ell+r_\ell+1$. 
In this case, we do not silence any transmitter/receiver pairs but we directly
apply one of the threes schemes in the previous Subsection~\ref{sec:subnetscheme}. By Proposition~\ref{prop5} this way we can
achieve a multiplexing gain of
$\textnormal{rank}(\mat{H}_{K}(\alpha))$, which trivially equals $K$
if $\deta{\mat{H}_{K}(\alpha)}\neq 0$ and by Corollary~\ref{lem:alpha2} equals $K-1$ otherwise.

\newcommand{\rank}[1]{\textnormal{rank}\left(#1\right)}

\subsection{ Achieving the Lower Bound in~\eqref{eq:S2a}}\label{sec:S2a}
Recall that \eqref{eq:S2a} holds under the assumption that $K> (t_\ell+r_\ell+2)$ and  $\deta{\mat{H}_{t_\ell+r_\ell+1}(\alpha)}\neq0$. 
We define 
\begin{IEEEeqnarray}{rCl}\label{eq:tka}
\tilde{\kappa}&\triangleq& K \mod (t_\ell+r_\ell+2)\\
\tilde{\gamma}&\triangleq &\left \lfloor \frac{K}{t_\ell+r_\ell+2}\right \rfloor\label{eq:gammat}
\end{IEEEeqnarray}
and notice that by assumption $\tilde{\gamma}\geq 1$.
 
We switch off the transmitter/receiver pairs $\{g(t_\ell+r_\ell+2)\}_{g=1}^{\tilde{\gamma}}$, i.e., in total $\tilde{\gamma}$ transmitter/receiver pairs. This decomposes the network into $\tilde{\gamma}$ subnets with $(t_\ell+r_\ell+1)$ transmitter/receiver pairs and possibly a smaller last network with $\tilde{\kappa}\leq (t_\ell+r_\ell+1)$ transmitter/receiver pairs. Thus, in each subnet we can apply one of the schemes described in Subsection~\ref{sec:subnetscheme}. By Proposition~\ref{prop5}, this achieves multiplexing gain $\rank{\mat{H}_{t_\ell+r_\ell+1} (\alpha)}$
over the first $\tilde{\gamma}$ subnets  and multiplexing gain $\rank{\mat{H}_{\tilde{\kappa}} (\alpha)}$
 over the last smaller network (if it exists). By assumption $\deta{\mat{H}_{t_\ell+r_\ell+1}(\alpha)}\neq0$ and thus $\rank{\mat{H}_{t_\ell+r_\ell+1}(\alpha)}=(t_\ell+r_\ell+1)$; moreover, by Corollary~\ref{lem:alpha2},   $\rank{\mat{H}_{\tilde{\kappa}}(\alpha)}$ is either equal to $\tilde{\kappa}$ or to $\tilde{\kappa}-1$. Thus, we achieve  at least the  
 desired multiplexing gain of $K-\left \lfloor \frac{K}{t_\ell+r_\ell+2}\right \rfloor-1$. In fact, whenever $\tilde{\kappa}=0$ or $\deta{\mat{H}_{\tilde{\kappa}}(\alpha)}\neq0$, then we can even achieve a multiplexing gain of $K-\left \lfloor \frac{K}{t_\ell+r_\ell+2}\right \rfloor$.

\subsection{Achieving the Lower Bound in~\eqref{eq:S2}}\label{sec:S2}
Recall that \eqref{eq:S2} holds under the assumption that $K> (t_\ell+r_\ell+2)$; that  $\deta{\mat{H}_{t_\ell+r_\ell+1}(\alpha)}\neq0$; and that $\deta{\mat{H}_{t_\ell+r_\ell}(\alpha)}\neq0$.

We distinguish two  cases depending on $\tilde{\kappa}$ as defined in \eqref{eq:tka}:
\begin{enumerate}
\item $\rank{\mat{H}_{\tilde{\kappa}}(\alpha)} =\tilde{\kappa}$; 
\item $\rank{\mat{H}_{\tilde{\kappa}}(\alpha)} <\tilde{\kappa}$.
\end{enumerate}

In case 1) we use the same scheme as in the previous Subsection~\ref{sec:S2a}. As described above, this scheme achieves a multiplexing gain of $\rank{\mat{H}_{t_\ell+r_\ell+1} (\alpha)}$
over each of the first $\left \lfloor \frac{K}{t_\ell+r_\ell+2}\right \rfloor$ subnets  and a multiplexing gain of $\rank{\mat{H}_{\tilde{\kappa}} (\alpha)}$
 over the last smaller network. Since we assumed that $\deta{\mat{H}_{t_\ell+r_\ell+1}(\alpha)} \neq 0$ and that $\rank{\mat{H}_{\tilde{\kappa}}(\alpha)} =\tilde{\kappa}$, we conclude we achieve the desired multiplexing gain of $K-\left \lfloor \frac{K}{t_\ell+r_\ell+2}\right \rfloor$ over the entire network.
 
We now treat case 2). 
Notice that in this case $\tilde{\kappa}< t_\ell+r_\ell$ because we assumed that $\deta{\mat{H}_{t_\ell+r_\ell+1}(\alpha)}\neq  0$ and that $\deta{\mat{H}_{t_\ell+r_\ell}(\alpha)}\neq0$.

We switch off transmitter/receiver pairs $\{g(t_\ell+r_\ell+2)\}_{g=1}^{\tilde{\gamma}-1}$ and transmitter/receiver pair $\tilde{\gamma}(t_\ell+r_\ell+2)-1$, where $\tilde{\gamma}$ is defined in \eqref{eq:gammat}. This way, the first  $\tilde{\gamma}-1$ subnets are of size $t_\ell+r_\ell+1$, the next subnet is of size $(t_\ell+r_\ell)$, and the last is of size $\tilde{\kappa}+1$ (where $\tilde{\kappa}$ is defined in \eqref{eq:tka}). 
Thus, all the subnets consist of at most $t_\ell+r_\ell+1$ transmitter/receiver pairs, and we can apply one of the three schemes described in Subsection~\ref{sec:subnetscheme}.

Since $\deta{\mat{H}_{t_\ell+r_\ell+1}(\alpha)}\neq 0$, by Proposition~\ref{prop5},  we achieve a multiplexing gain of $t_\ell+r_\ell+1$ over each of the first $\tilde{\gamma}-1$ subnets. Moreover, since we assumed that $\deta{\mat{H}_{t_\ell+r_\ell}(\alpha)}\neq0$, we further achieve a multiplexing gain of $(t_\ell+r_\ell)$ over the $\tilde\gamma$-th subnet. Finally, since we assumed that $\deta{\mat{H}_{\tilde{\kappa}}(\alpha)}=0$, by Lemma~\ref{lem:alpha}, $\deta{\mat{H}_{\tilde{\kappa}+1}(\alpha)}\neq0$, and thus we achieve a multiplexing gain of $\tilde{\kappa}+1$ over the last subnet. We conclude that our scheme achieves full multiplexing gain (i.e., multiplexing gain equal to the number of transmitter/receiver pairs) in each subnet and hence a multiplexing gain of $K- \left \lfloor    \frac{K }{t_\ell+r_\ell+2} \right \rfloor$ over the entire network.

\subsection{ Achieving the Lower Bound in~\eqref{eq:S3}}\label{sec:S3}
Recall that \eqref{eq:S3} holds under the assumptions that $K> t_\ell+r_\ell+2$ and $\deta{\mat{H}_{t_\ell+r_\ell+1}} =0$. 

We switch off every $(t_\ell+r_\ell+1)$-th transmitter/receiver pair, i.e., in total $\left \lfloor 
\frac{K}{t_\ell+r_\ell+1}\right \rfloor$ transmitter/receiver pairs, and, depending on the values of $t_\ell, t_r, r_\ell, r_r$, we apply one of the three schemes in Subsection~\ref{sec:subnetscheme} over the resulting subnets. Following similar lines as in the previous proof, it can be shown that all the resulting subnets have full-rank channel matrices and thus by Proposition~\ref{prop5} a multiplexing gain of $K- \left \lfloor \frac{K}{t_\ell+r_\ell+1}\right \rfloor$ is achieved over the entire network. The details of the proof are omitted.


\section{Proof of Proposition~\ref{prop:genlb}}\label{app:lowerbound}

We first prove the lower bound in \ref{2}), followed by the lower bounds in \ref{3}), \ref{1}), and \ref{4}).

\subsection{Proof of  Lower Bound  \ref{2}), i.e.,  \eqref{eq:lowergen4}} \label{sec:lowergen4}

If $t_\ell=0$, then \eqref{eq:lowergen4} follows from lower bound \eqref{eq:lowergen2}. 
Moreover, if $t_\ell+r_\ell \leq 1$, then there is nothing to prove, as the multiplexing gain cannot be negative.

Thus, in the following we assume that $t_\ell+r_\ell \geq 2$ and $t_\ell\geq 1$, and present a scheme that achieves the lower bound in \eqref{eq:lowergen4} under this assumption.
Our scheme is similar to the scheme for the asymmetric network described in Section~\ref{sec:ach} when this latter is specialized to $t_r=r_r=0$. (In particular our scheme here disregards 
the right side-information available to the transmitters and the receivers.) 

The idea is again to silence some of the transmitters, which decomposes our asymmetric network  into several subnets, and  to apply a scheme based
on Costa's dirty-paper coding and successive interference cancellation to communicate over the subnets. However, here,  due to the two-sided interference, \emph{pairs} of consecutive transmitters are silenced and  
the dirty-paper coding and the successive interference cancellation strategies  are used to "cancel" 
\emph{two} interference signals.

Define\begin{IEEEeqnarray}{rCl}
\beta_2 & \triangleq&  (t_\ell+r_\ell+1)\\
\gamma_2 & \triangleq & \left \lfloor \frac{K}{\beta_2} \right \rfloor
\end{IEEEeqnarray} 
and recall that in Proposition~\ref{prop:genlb} we defined $\kappa_2  \triangleq K \mod \beta_2$ and 
\begin{equation}
\theta_2 \triangleq \begin{cases} 2, & \textnormal{if } \kappa_2 \geq 2 \\ 1, &\textnormal{if } 
\kappa_2=1 \\ 0, & \textnormal{if } \kappa_2=0.\end{cases}
\end{equation}
\subsubsection{Splitting the Network into Subnets}


We  silence transmitters   $\{ m \beta_2 +1\}_{m=0}^{\gamma_2-1}$ and 
transmitters $\{m\beta_2\}_{m=1}^{\gamma_2}$. Moreover, if
$\theta_2=1$  we also silence transmitter $(\gamma_2 \beta_2+1)$ and
if $\theta_2=2$ then also transmitters $(\gamma_2\beta_2+1)$ and
$K$. Notice that in total we silence $2 \gamma_2+\theta_2$ transmitters.
Silencing the chosen subset of transmitters splits the network into $\gamma_2$ 
non-interfering subnets if $\theta_2=0$ and into $\gamma_2+1$
non-interfering subnets if $\theta_2\geq 1$. In both cases, the first
$\gamma_2$ subnets all have the same topology and consist of $\beta_2-2$ active transmit antennas and
of $\beta_2$ receive antennas. In fact, 
 the $m$-th subnet, for $m \in \left\{ 1, \ldots, \gamma_2\right
\}$, consists of transmit antennas $( (m-1) \beta_2 +2), \ldots,
(m\beta_2 -1)$ and receive antennas $((m-1) \beta_2 +1), \ldots,
m\beta_2 $. We call these subnets generic. If $\theta_2 \geq 1$, then
there is an additional last smaller subnet which consists of
$\max\{\kappa_2-2,0\}$ active transmit antennas and $\kappa_2$ receive
antennas. More precisely, it consists   of  transmit
antennas $ (K- \kappa_2+ 2), \ldots, (K-1)$ (i.e., of no transmit
antennas if $\kappa \leq 2$) and of receive antennas $( K- \kappa_2+ 1), \ldots, K$.


The scheme employed over a subnet depends on whether the scheme is generic or reduced and on the parameter $r_\ell\geq 0$. We describe the schemes in the following subsections.

%

\subsubsection{Scheme over a Generic Subnet when $r_\ell\geq 1$}\label{sec:k1}

We assume that the first subnet is generic and describe the scheme for
this first subnet. 

We transmit Messages $M_{2}, \ldots,
M_{r_\ell+t_\ell}$ over the subnet.
Define the sets
  \begin{IEEEeqnarray}{rCl}
 \mathcal{F}_1 &\triangleq&  \{2, \ldots, r_\ell+1 \}   \label{eq:F1}\\
 \mathcal{F}_2 &  \triangleq& \{ r_\ell+2, \ldots, r_\ell+ t_\ell\}.\label{eq:F2}
   \end{IEEEeqnarray}
Messages $\{M_k\}_{k\in\mathcal{F}_1}$ are transmitted as follows. 
 \begin{itemize}
 \item For each $k\in\mathcal{F}_1$ we construct a single-user Gaussian codebook $\mathcal{C}_k$ of power $P$, blocklength $n$, and rate $R_k=\frac{1}{2} \log(1+P)$. The code $\mathcal{C}_k$ is revealed to Transmitter~$k$ and to Receivers~$2,\ldots, k$.
\item Each Transmitter~$k\in\mathcal{F}_1$ ignores the side-information about other transmitters' messages
  and codes for a Gaussian single-user channel.  That is, it picks the codeword from codebook~$\mathcal{C}_k$  that corresponds to its message $M_k$ and  sends this codeword over the channel.

\item Receiver~$k\in\mathcal{F}_1$, uses successive
  interference cancellation to decode  its desired Message $M_{k}$. 
 Let $\hat{X}_0^n$ and $\hat{X}_1^n$ be two all-zero sequences of length $n$. Receiver~$k$ initializes $j$ to $2$, and while $j \leq k$:
  \begin{itemize}
  \item It decodes Message~$M_{j}$ based on the difference 
  \begin{equation}\label{eq:Yj1diff}
  Y_{j-1}^n - \alpha \hat X_{j-2}^n - \hat{X}_{j-1}^n 
  \end{equation}using an optimal ML-decoder. 
 Let $\hat{M}_{j}$ denote the resulting guess.
  \item It picks the codeword $x_j^n(\hat{M}_j)$ from codebook $\mathcal{C}_j$ that corresponds to the guess $\hat{M}_j$ and produces this codeword as its reconstruction of the input $\hat{X}_{j}^n$:
  \begin{equation}
  \hat{X}_{j}^n = x_j^n(\hat{M}_j).
  \end{equation}
  \item It increases the index $j$ by $1$.
  \end{itemize}
  \item    Notice that Receiver~$k\in\mathcal{F}_1$ has access to the output signals $Y_1^n,\ldots Y_k^n$ because $k\leq r_\ell+1$. 
    \item For each $k \in\mathcal{F}_1$, if the previous two messages were decoded correctly,  $\hat{M}_{k-2}=M_{k-2}$ and $\hat{M}_{k-1}=M_{k-1}$, 
    \begin{equation}
    Y_{k-1}^n- \alpha \hat{X}_{k-2}^n- \hat{X}_{k-1}^n = \alpha X_k^n+N_{k-1}^n.
    \end{equation}
    Thus, in this case, Message $M_k$ is decoded based on the interference-free outputs $\alpha X_k^n+N_{k-1}^n$, and, by construction of the code $\mathcal{C}_k$, 
    the average probability of error  
    \begin{equation}\label{eq:errf1}
    \Pr[\hat{M}_k=M_k] \to 0 \quad \textnormal{ as } n\to \infty.
    \end{equation}
 \end{itemize}
 
If $t_\ell\geq 2$,  Messages $\{M_{k}\}_{k\in\mathcal{F}_2}$ are transmitted as follows: 
 \begin{itemize}
\item For each $k\in\mathcal{F}_2$, construct a dirty-paper code $\mathcal{C}_k$ of power~$P$ and rate $R_k=\frac{1}{2}\log(1+\alpha^2 P)$  for noise variance~$1$ and interference variance~$(\alpha^2 P+ P)$ (which is the variance of $\alpha X_{k-2}+ X_{k-1}$). The code $\mathcal{C}_k$ is revealed to Transmitters~$k,\ldots, r_\ell+t_\ell$ and  to Receiver~$k$.
\item Each Transmitter~$k\in\mathcal{F}_2$
  computes the "interference term"
  $\alpha X_{k-2}^n + {X}^n_{k-1}$
and  applies the dirty-paper code $\mathcal{C}_k$  to encode its message $M_k$ and mitigate  the ``interference''
  $\alpha X_{k-2}^n + {X}^n_{k-1}$. 
Denoting the resulting sequence by $\tilde{X}_{k}^n$, the transmitter sends the scaled
  version 
  \begin{equation} 
  X_k^n = \frac{1}{\alpha} \tilde{X}_k^n.
  \end{equation}

\item Each Receiver~$k\in\mathcal{G}_2$ considers only the outputs at the antenna of its left neighbor, $Y_{k-1}^n$. It  uses code $\mathcal{C}_k$ to  apply dirty-paper decoding based on the outputs
   \begin{IEEEeqnarray}{rCl}
   \bvect{Y}_{k-1}& = & \alpha X_{k-2}^n + X_{k-1}^n + \alpha X_{k}^n+ \bvect{N}_{k}\\
   & = & \tilde X_{k}^n +  \underbrace{\alpha X_{k-2}^n + X_{k-1}^n}_{\textnormal{``interference"}}+ \bvect{N}_{k}.
   \end{IEEEeqnarray}
   
\item Notice that Transmitter~$k\in\mathcal{G}_2$ can compute the sequences
  ${X}^n_{k-2}$ and $X_{k-1}^n$, because in our scheme they only depend on Messages~$M_{r_\ell},\ldots, M_{k-2}$ and $M_{r_\ell},\ldots, M_{k-1}$, respectively.

   \item By construction, the sequence $\tilde X_k^n$, which encodes Message $M_k$,  can completely mitigate the ``interference"   $\alpha X_{k-2}^n + {X}^n_{k-1}$, and 
   the average probability of error 
       \begin{equation}\label{eq:errf2}
    \Pr[\hat{M}_k=M_k] \to 0 \quad \textnormal{ as } n\to \infty.
    \end{equation}
\end{itemize}
To summarize, with the described scheme, we sent Messages~$M_2,\ldots, M_{r_\ell+t_\ell}$ with vanishingly small probability of error, see~\eqref{eq:errf1} and \eqref{eq:errf2}, and at rates 
\begin{IEEEeqnarray}{rCl}
R_2 = \ldots = R_{r_\ell+t_\ell} &=&\frac{1}{2}\log(1+\alpha^2 P). 
\end{IEEEeqnarray}  

\subsubsection{Scheme over a Generic Subnet when $r_\ell=0$}\label{sec:k2}
In this case the set $\mathcal{F}_1$ is empty  
whereas by the assumption
$t_\ell+r_\ell\geq 2$, $t_\ell\geq 2$ and the set $\mathcal{F}_2$ is non-empty. 
We transmit Messages $\{M_{k-1}\}_{k\in\mathcal{F}_2}$  over the
subnet.

Specifically, each Transmitter $k\in \mathcal{F}_2$ employs the dirty-paper scheme as described in th previous subsection~\ref{sec:k1}, except 
that now, instead of
sending its own message $M_k$, it sends its left-neighbor's
message $M_{k-1}$ (to which it has access because $t_\ell \geq 1$). Accordingly, the outputs $Y_{k-1}^n$, for $k\in\mathcal{F}_2$, are now used by Receiver~$k-1$ to decode its desired message $M_{k-1}$. 

Here, for each $k\in\mathcal{F}_2$, the probability of error of Message $M_{k-1}$  equals  the probability of error of Message~$M_k$ in the previous subsection~\ref{sec:k1}. Thus, by~\eqref{eq:errf2}, 
     for all $k\in\mathcal{F}_2$:  \begin{equation}\label{eq:errf2zero}
    \Pr[\hat{M}_{k-1}=M_{k-1}] \to 0 \quad \textnormal{ as } n\to \infty.
    \end{equation}

We conclude that with the described scheme, the messages $M_1,\ldots, M_{t_\ell-1}$ are communicated with vanishingly small probability of error and at rates 
\begin{IEEEeqnarray}{rCl}
R_1 = \ldots = R_{r_\ell+t_\ell-1} &=&\frac{1}{2}\log(1+\alpha^2 P).   
\end{IEEEeqnarray}  
 \begin{conclusion}\label{rem:127}
Our schemes for generic subnets described here and in the previous subsection~\ref{sec:k1}  achieve a multiplexing gain of $r_\ell+t_\ell-1$ over a generic subnet when $r_\ell\geq 1$ and when $r_\ell=0$, respectively. Both schemes  use all the $(t_\ell+r_\ell-1)$ active
   transmit antennas of the subnet; but they use only the first $(t_\ell+r_\ell-1)$
   receive antennas and ignore the last two receive antennas of the subnet. 
      \end{conclusion}
\subsubsection{Scheme over a Reduced Subnet}

Over the reduced subnet  we use one of the two schemes for generic subnets of Subsections~\ref{sec:k1} and \ref{sec:k2}, but with reduced side-information parameters 
\begin{subequations}
 \begin{align}
 r_\ell'&\triangleq \min\left[\left(\kappa_2-1\right),r_\ell\right]\\
 t_\ell'&\triangleq \min\left[\left(\kappa_2-r_\ell-1\right)_+,t_\ell\right].
  \end{align}
 \end{subequations}
By Conclusion~\ref{rem:127}, this achieves a multiplexing gain of $\max\{\kappa_2-2,0\}$ over a reduced subnet.

\subsubsection{Analysis of the Performance over the Entire Network}
Over the first $\lfloor K/\beta_2\rfloor$ generic subnets we achieve a multiplexing gain of $\beta_2-2$ and, if it exists, then over the last reduced subnet we achieve a multiplexing gain of $\max\{\kappa_2-2,0\}$. Thus, over the entire network we achieve a multiplexing gain of 
\begin{equation}
K- 2\gamma_2-\theta_2 = \begin{cases} K- 2 \lfloor K/\beta_2\rfloor - 2, & \textnormal{ if } \kappa_2 \geq 2\\K- 2 \lfloor K/\beta_2\rfloor-\kappa_2 & \textnormal{ if } \kappa_2 < 2.
\end{cases}
\end{equation}
This establishes the desired lower bound.


\subsection{Proof of  Lower Bound~\ref{3})} \label{sec:lowergen3}

By symmetry, this lower bound follows directly from \eqref{eq:lowergen4}.
In particular, if $t_r\geq 1$ and $t_r+r_r \geq 2$, a scheme that is symmetric to the scheme described in the previous subsection~\ref{sec:lowergen4} achieves the desired multiplexing gain in \ref{3}). 
We briefly sketch this  scheme 
because we will use it to prove the lower bound in \ref{1}), \eqref{eq:lowergen1}, in
Subsection~\ref{sec:lowergen1} ahead. 

Define 
\begin{IEEEeqnarray}{rCl}
\beta_2' &\triangleq &(t_r+r_r+1),\\
\gamma_2' & \triangleq & \left \lfloor \frac{K}{\beta_2'} \right \rfloor, \\
\kappa_2' & \triangleq & K \mod \beta_2',
\end{IEEEeqnarray}
and\begin{equation}
\theta_2' \triangleq \begin{cases} 2, & \textnormal{if } \kappa_2' \geq 2 \\ 1, &\textnormal{if } 
\kappa_2'=1 \\ 0, & \textnormal{if } \kappa_2'=0.\end{cases}
\end{equation}

\subsubsection{Splitting the Network into Subnets}
We silence transmitters
$\{m\beta_2'+1\}_{m=0}^{\gamma_2'-1}$ and transmitters
$\{m\beta_2'\}_{m=1}^{\gamma_2'}$. Moreover, if $\theta_2'=1$ then we
also silence transmitter $(\gamma_2' \beta_2'+1)$ and if $\theta_2'=2$
then also transmitters $(\gamma_2'\beta_2'+1)$ and $K$. This splits
the network into $\gamma_2'$ generic subnets with $\beta_2'-2$ active
transmit antennas and $\beta_2'$ receive antennas, and if
$\theta_2'\in\{1,2\}$ then there is an additional last reduced subnet
with $\max\{\kappa_2'-2,0\}$ active transmit antennas and
$\kappa_2'$ receive antennas. 

The scheme that we employ in the subnets depends on whether the subnet is generic or reduced and on 
the parameter $r_r\geq 0$. 

\subsubsection{Scheme over a Generic Subnet when $r_r\geq 1$}\label{sec:h1}
Define the sets
$\mathcal{F}_3$ and $\mathcal{F}_4$ as:
\begin{IEEEeqnarray*}{rCl}
 \mathcal{F}_3 &\triangleq& \{ 2, \ldots, t_r\} \\
 \mathcal{F}_4 &  \triangleq& \{t_r+1, \ldots, t_r+r_r \}.
   \end{IEEEeqnarray*}
Assume that the first subnet is generic. Then, over this first subnet we transmit messages $M_{2}, \ldots, M_{t_r+r_r}$. 

Messages $\{M_{k}\}_{k\in\mathcal{F}_3}$ are transmitted in a similar way as  Messages $\{M_{k}\}_{k\in\mathcal{G}_3}$ in the scheme in Subsection~\ref{sec:ach},  and Message $\{M_{k}\}_{k\in\mathcal{F}_4}$ are transmitted in a similar way as  Messages~$\{M_{k}\}_{k\in\mathcal{G}_4}$ in that scheme. The only difference is that here, each  dirty-paper code~$\mathcal{C}_k$, for $k\in \mathcal{F}_3$, has to be designed for an interference variance $(\alpha^2P+P)$ so that it can mitigate the ``interference" $X_{k+1}^n+ \alpha X_{k+2}^n$;  likewise, during the successive interference cancellation steps, each Receiver~$k\in\mathcal{F}_4$ has to cancel  the two ``interference" terms $X_{k+1}^n$ and $\alpha X_{k+2}^n$. 

For brevity, we omit the details of the scheme and of the analysis. It can be shown that the scheme achieves a multiplexing gain of $t_r+r_r-1$ over the generic subnet.

\subsubsection{Scheme over a Generic Subnet when $r_r=0$}
In this case, the set $\mathcal{F}_4$ is empty whereas, by the assumption $t_r+r_r \geq 2$, the set $\mathcal{F}_3$ is nonempty. We transmit messages $M_{3}, \ldots, M_{t_r+r_r+1}$ over the subnet. 

Messages $\{M_{k+1}\}_{k\in\mathcal{F}_3}$ are transmitted in the same way as messages $\{M_{k+1}\}_{k\in\mathcal{G}_3}$ in Subsection~\ref{sec:ach}. For brevity, we omit details and analysis. It can be shown that such a scheme achieves a multiplexing gain of $t_r+r_r-1$ over the generic subnet.

   \begin{conclusion}\label{rem:148} 
 Our schemes in the previous subsection~\ref{sec:h1} and here achieve a multiplexing gain of $r_r+t_r-1$ over a generic subnet when $r_r\geq 1$ and when $r_r=0$, respectively. Both schemes use all  $(t_r+r_r-1)$ active
   transmit antennas of the subnet; but they use only the last $(t_r+r_r-1)$
   receive antennas and ignore the first two receive antennas of the subnet. 
      \end{conclusion}

\subsubsection{Scheme over a Reduced Subnet}
Over a reduced subnet we employ the schemes for a generic subnet described above, but with reduced side-information parameters
\begin{subequations}
 \begin{align}
 t_r'&\triangleq \min\left[\left(\kappa_2'-2\right)_+,t_r\right]\\
 r_r'&\triangleq \min\left[\left(\kappa_2'-t_r-2\right)_+,r_r\right]
 \end{align}
 \end{subequations} 
By Conclusion~\ref{rem:148}, such a scheme achieves a multiplexing gain of $\max\{\kappa_2'-2, 0\}$ over the reduced subnet.

\subsection{Proof of  Lower Bound \ref{1}), i.e.,  \eqref{eq:lowergen1}} \label{sec:lowergen1}
If $t_\ell+r_\ell=0$ or $t_r+r_r=0$, then the proof follows directly from the lower bounds in \ref{2})
or \ref{3}). 
If $t_\ell+t_r+r_\ell+r_r \leq  2$, there is nothing to prove.

Thus in the following we assume that $t_\ell+t_r+r_\ell+r_r \geq 3$
and $(t_\ell+r_\ell), (t_r+r_r)\geq 1$. 
Define
\begin{IEEEeqnarray}{rCl}
\beta_1 & \triangleq&  (t_\ell+t_r+r_\ell+r_r)\\
\gamma_1 & \triangleq & \left \lfloor\frac{K}{\beta_1} \right \rfloor,
\end{IEEEeqnarray} 
and recall that  in Proposition~\ref{prop:genlb} we defined $\kappa_1  \triangleq K \mod \beta_1$ and 
\begin{equation}
\theta_1 \triangleq \begin{cases} 2, & \textnormal{if } \kappa_1 \geq 2 \\ 1, &\textnormal{if } 
\kappa_1=1 \\ 0, & \textnormal{if } \kappa_1=0.\end{cases}
\end{equation}
\subsubsection{Splitting the Network into Subnets}

We silence transmitters   $\{ m \beta_1 +1\}_{m=0}^{\gamma_1-1}$ and
transmitters $\{m\beta_1\}_{m=1}^{\gamma_1}$. Moreover, if
$\theta_1=1$, then we also silence transmitter $(\gamma_1 \beta_1
+1)$, and if $\theta_1=2$, then also transmitters $(\gamma_1 \beta_1
+1)$ and $K$. Thus, in total we silence $2 \gamma_1 + \theta_1$
transmitters. 
 This splits the network into $\gamma_1$ or $\gamma_1+1$
non-interfering subnets: the first $\gamma_1$ generic subnets consist of 
$(\beta_1-2)$ transmit antennas and $\beta_1$
receive antennas, and if there is an
additional last subnet then it is smaller and consists of 
$\max\{\kappa_1-2,0\}$ transmit antennas and of $\kappa_1$
receive antennas. 

The scheme employed in each subnet depends on whether the subnet is generic or reduced. 

\subsubsection{Scheme over a Generic Subnet}

We assume that the first subnet is generic and describe the scheme for this first subnet. To this end, define the groups 
  \begin{IEEEeqnarray*}{rCl}
 \mathcal{F}_{1/2} &\triangleq&  \{2, \ldots, r_\ell+t_\ell\}\\
 \mathcal{F}_{3/4} &  \triangleq& \{( r_\ell+t_\ell+1), \ldots, (r_\ell+ t_\ell+t_r+r_r-1)\}. 
 \end{IEEEeqnarray*}
 
Our scheme is a combination of the two schemes for generic subnets described in Sections~\ref{sec:lowergen4} and \ref{sec:lowergen3}.  Over the left 
part of the subnet that consists of transmit antennas~$k \in \mathcal{F}_{1/2}$ and receive antennas $1,\ldots, (t_\ell+r_\ell-1)$ we use the scheme in Section~\ref{sec:lowergen4}. 
Over the right 
part of the subnet that consists of transmit antennas~$k \in \mathcal{F}_{3/4}$ and receive antennas
 $(r_\ell+t_\ell+2),\ldots, (t_\ell+r_\ell+t_r+r_r)$ we use the scheme in Section~\ref{sec:lowergen3} where the set $\mathcal{F}_3$ needs to be replaced by $\{( r_\ell+t_\ell+1), \ldots, (t_\ell+r_\ell+t_r-1)\}$ and the set $\mathcal{F}_4$  by $\{(r_\ell+t_\ell+t_r), \ldots, (r_\ell+t_\ell+t_r+r_r-1)\}$ . Thus, the combined scheme utilizes all the transmit antennas in the subnet but only receive antennas~$1,\ldots, r_\ell+t_\ell-1$ and $r_\ell+t_\ell+2, \ldots, r_\ell+t_\ell+t_r+r_r+2$, i.e., it ignores the two receive antennas
 $(t_\ell+r_\ell)$ and $(t_\ell+r_\ell+1)$, see also Conclusions~\ref{rem:127} and \ref{rem:148}.

 Since the transmit antennas $k\in\mathcal{F}_{1/2}$ in the ``left-hand" scheme do not influence the signals
 observed at receive antennas $(r_\ell+t_\ell+2),\ldots, (t_\ell+r_\ell+t_r+r_r)$ employed in the ``left-hand" scheme, and the
 signals sent  at transmit antennas~$k\in \mathcal{F}_{3/4}$ in the ``right-hand" scheme do not influence the signals observed at receive antennas
 $1, \ldots, (t_\ell+r_\ell-1)$ employed in the
 ``left-hand" scheme,  the performance of the two
schemes can be analyzed separately. 
 By Conclusions~\ref{rem:127} and \ref{rem:148} we  achieve a multiplexing gain of $r_\ell+t_\ell-1$ over the left part of the subnet and a multiplexing gain of $t_r+r_r-1$ over the right part of the subnet. 
Thus, 
 we achieve a multiplexing gain $r_\ell+t_\ell+t_r+r_r-2$ over the entire subnet. 
 
\subsubsection{Scheme over a Reduced Subnet} We employ the same scheme as over a generic subnet but with reduced side-information parameters. Details and analysis are omitted for brevity. Such a scheme can achieve a multiplexing gain of $\max\{\kappa_1-2, 0\}$ over a reduced subnet. 
 
\subsubsection{Analysis of Performance over the Entire Network}
Over the first $\lfloor K/\beta_1\rfloor$ generic subnets we achieve a multiplexing gain of $\beta_1-2$ and, if it exists, then over the last reduced subnet we achieve a multiplexing gain of $\max\{\kappa_1-2,0\}$. Thus, over the entire network we achieve a multiplexing gain of 
\begin{equation}
K- 2\gamma_1-\theta_1 = \begin{cases} K- 2 \lfloor K/\beta_1\rfloor - 2, & \textnormal{ if } \kappa_1 \geq 2\\K- 2 \lfloor K/\beta_1\rfloor-\kappa_1 & \textnormal{ if } \kappa_1 < 2.
\end{cases}
\end{equation}
This establishes the desired lower bound.

\subsection{Proof of  Lower Bound  \ref{4}), i.e.,  \eqref{eq:lowergen2}} \label{sec:lowergen2}

In our scheme the transmitters ignore their side-information.
Define
\begin{IEEEeqnarray}{rCl}
\beta_3 & \triangleq&  (r_\ell+r_r+3)\\
\gamma_3 & \triangleq & \left \lfloor \frac{K}{\beta_3} \right \rfloor, 
\end{IEEEeqnarray} 
and recall that in Proposition~\ref{prop:genlb} we defined $\kappa_3  \triangleq K \mod \beta_3$ and 
\begin{equation}
\theta_3 \triangleq \begin{cases} 2, & \textnormal{if } \kappa_3 \geq 2 \\ 1, &\textnormal{if } 
\kappa_3=1 \\ 0, & \textnormal{if } \kappa_3=0.\end{cases}
\end{equation}
\subsubsection{Splitting the Network into Subnets}
We silence transmitters   $\{ m \beta_3 +1\}_{m=0}^{\gamma_3-1}$ and
transmitters $\{m\beta_3\}_{m=1}^{\gamma_3}$. Moreover, if
$\theta_3=1$, we also silence transmitter $\beta_3 \gamma_3 +1$, and
if $\theta_3=2$, we also silence transmitters $\beta_3 \gamma_3 +1$ and
$K$. Notice that in total we have silenced $2 \gamma_3 + \theta_3$
transmitters.

This splits the network into $\gamma_3$ or $\gamma_3+1$
non-interfering subnets: the first $\gamma_3$ subnets consist of
$\beta_3-2$ active transmit antennas and $\beta_3$ receive
antennas (we call these subnets generic), and if an additional last subnet exists it is smaller and
consists of $\max\{\kappa_3-2,0\}$ transmit and $\kappa_3$
receive antennas. 

The scheme employed over the subnets depends on whether the subnet is generic or reduced. 

\subsubsection{Scheme over a Generic Subnet}

 We assume that the first subnet is generic and describe our scheme for this first subnet. 

Define
 \begin{IEEEeqnarray*}{rCl}
 \mathcal{H}_1 &\triangleq&  \{2, \ldots, r_\ell+1\}\\
\mathcal{H}_2 & \triangleq & \{r_\ell+2\}\\  
\mathcal{H}_3 &  \triangleq &\{ r_\ell+3, \ldots, r_\ell+ r_r+2\} .
  \end{IEEEeqnarray*}
 We only sketch the scheme. 
 \begin{itemize} 
 \item Messages $M_2,\ldots, M_{r_\ell+r_r+2}$ are transmitted over the subnet.
 \item 
For each $k\in(\mathcal{H}_1\cup \mathcal{H}_2 \cup \mathcal{H}_3)$, Transmitter~$k$
encodes its Message $M_k$ using a Gaussian
 point-to-point code. 
 \item For each $k\in\mathcal{H}_1$, Receiver~$k$ decodes its
 message using  
 successive interference cancellation from the \emph{left}, starting with the \emph{first} antenna
  in the subnet. 
  These messages can be decoded with arbitrary small probability of error (for sufficiently large blocklengths), 
whenever
\begin{equation}\label{eq:Rates1}
R_k \leq \frac{1}{2}\log \left( 1+ \alpha^2 P\right), \qquad \forall k\in\mathcal{H}_1.
\end{equation}
\item  Similarly, for each $k\in\mathcal{H}_3$, Receiver~$k$ decodes its
 message using  
 successive interference cancellation but now from the \emph{right} and starting with the \emph{last} antenna
  in the subnet. These messages can be decoded with arbitrary small probability of error (for sufficiently large blocklengths), 
whenever
\begin{equation}\label{eq:Rates2}
 R_{k}\leq \frac{1}{2}\log \left( 1+ \alpha^2 P\right), \qquad \forall k \in \mathcal{H}_3.
\end{equation}

\item  Receiver~$r_\ell+2$, which
  has access to antennas $2, \ldots, (r_\ell+ r_r+2)$, decodes its desired Message $M_{r_\ell+2}$ by decoding \emph{all} the
  transmitted messages $M_{2}, \ldots, M_{r_\ell+r_r+2}$ using an optimal MIMO decoder \cite{Telatar-ETT-99}.
In this step, we have arbitrary small probability of error, whenever
\begin{equation}\label{eq:sumofrates}
\sum_{i=2}^{r_\ell+r_r+2} R_i\leq \frac{1}{2} \log \left( \deta{ \mat{I} + P\trans{\mat{H}}_{r_\ell+r_r+1}\mat{H}_{r_\ell+r_r+1}}\right);
\end{equation} 
where here for ease of notation we wrote $\mat{H}_{r_\ell+r_r+1}$ instead of $\mat{H}_{r_\ell+r_r+1}(\alpha)$.
Notice that since the channel matrix $H_{r_\ell+r_r+1}(\alpha)$ is non-singular
  and does not depend on the power $P$, by \cite{Telatar-ETT-99}:
  \begin{eqnarray}\label{eq:MIMO}
\lefteqn{  \varlimsup_{P\rightarrow \infty} \frac{
  \frac{1}{2} \log \left( \deta{ \mat{I} + P\trans{\mat{H}}_{r_\ell+r_r+1}\mat{H}_{r_\ell+r_r+1} }\right)}{\frac{1}{2}\log (P)}} \qquad \nonumber  \\ & = &r_\ell+r_r+1.\hspace{4cm}
  \end{eqnarray}
\end{itemize}
Combining \eqref{eq:Rates1}--\eqref{eq:MIMO}, we conclude that the described scheme can achieve a multiplexing gain of $r_\ell+r_r+1$ over the entire subnet. 

\subsubsection{Scheme over a Reduced Subnet} We employ the same scheme as over a generic subnet but with reduced side-information parameters.  Such a scheme  can achieve a multiplexing gain of $\max\{\kappa_3-2, 0\}$ over a reduced subnet. Details and analysis omitted.
 
\subsubsection{Analysis of Performance over the Entire Network}
Over the first $\lfloor K/\beta_3\rfloor$ generic subnets we achieve a multiplexing gain of $\beta_3-2$ and, if it exists, then over the last reduced subnet we achieve a multiplexing gain of $\max\{\kappa_3-2,0\}$. Thus, over the entire network we achieve a multiplexing gain of 
\begin{equation}
K- 2\gamma_3-\theta_3 = \begin{cases} K- 2 \lfloor K/\beta_3\rfloor - 2, & \textnormal{ if } \kappa_3 \geq 2\\K- 2 \lfloor K/\beta_3\rfloor-\kappa_3 & \textnormal{ if } \kappa_3 < 2.
\end{cases}
\end{equation}
This establishes the desired lower bound.


\mw{
\section{Proof of Proposition~\ref{th:genub}}\label{sec:ub}

\subsection{Proof of Upper Bound~\ref{1a}), i.e., \eqref{eq:genub1}} \label{sec:ub1}

Define 
\begin{IEEEeqnarray*}{rCl}
\beta_4 & \triangleq & t_\ell+t_r+r_\ell+ r_r+4,\\
\gamma_4 & \triangleq & \left \lfloor \frac{K}{\beta_4} \right \rfloor,
 \end{IEEEeqnarray*}
 and recall that $\kappa_4  \triangleq  K- \gamma_4 \beta_4$ and that $\theta_4$ equals $1$ if $\kappa_4 \geq \min\{t_\ell+r_\ell+1, t_r+r_r+1\}$, and it equals 0 otherwise.
 
 The proof is based on the Dynamic-MAC Lemma~\ref{lem:dynMAC}. To describe the choice of parameters for which we wish to apply this lemma, we need the following definitions. Define for every positive integer $p\geq 2$ and every non-zero number $\alpha$ the matrix 
$\mat{M}_{p}(\alpha)$ as the $p\times p$ matrix with diagonal elements $\alpha$, first upper off-diagonal elements~$1$, second upper off-diagonal elements~$\alpha$, and all other elements~$0$. 
That means, the row-$j_r$ column-$j_c$ entry of the matrix $\mat{M}_p(\alpha)$ equals $\alpha$ 
if $j_r=j_c$ or $j_r=j_c-2$, it equals $1$ if $j_r=j_c-1$, and it equals 0 otherwise.  Let 
$\mat{M}_{p}^{\textnormal{inv}}(\alpha)$ denote the inverse matrix of $\mat{M}_p(\alpha)$. This 
inverse always exists because $\deta{\mat{M}}_p(\alpha)=\alpha^p$, which by our assumption 
$\alpha\neq 0$ is nonzero. 
As we will see shortly, our main interest lies in the inverses $\mat{M}_{t_\ell+r_\ell+1}^{\textnormal{inv}}(\alpha)$ and $\mat{M}_{t_r+r_r+1}^{\textnormal{inv}}(\alpha)$. To simplify notation, we therefore denote 
the row-$j_r$ column-$j_c$ entry of $\mat{M}_{t_\ell+r_\ell+1}^{\textnormal{inv}}(\alpha)$ by $a_{j_r,j_c}$ and the row-$j_r$ column-$j_c$ entry of $\mat{M}_{t_r+r_r+1}^{\textnormal{inv}}(\alpha)$ 
by $b_{j_r,j_c}$. 
 
 We treat the cases $\theta_4=0$ and $\theta_4=1$ separately. If $\theta_4=1$, then we apply Lemma~\ref{lem:dynMAC} to the following choices: 
 \begin{itemize}
 \item $q=1$;
 \item $g=2\gamma_4$; 
 \item $\mathcal{A} = \bigcup_{m=0}^{\gamma_4-1} \mathcal{A}'(m)$,  where for $m \in\{0,\ldots,  \gamma_4-2\} $:
  \begin{equation}\label{eq:Amp} 
  \mathcal{A}'(m)\triangleq \{ (m \beta_4 + r_\ell +2), \ldots, (m \beta_4 + r_\ell+t_\ell+t_r+3)\},
  \end{equation} 
  and     \begin{IEEEeqnarray}{rCl}
  \lefteqn{
    \mathcal{A}'(\gamma_4-1)}\nonumber \\  &\triangleq &\{((\gamma_4 -1) \beta_4 + r_\ell + 2), \ldots, (\gamma_4 \beta_4-r_r +3)\}\nonumber \\ &  & \hspace{1.9cm} \cup \{ (\gamma_4 \beta_4 + r_\ell +2) ,\ldots, K\};
  \end{IEEEeqnarray}
  \item  $\mathcal{B}_1= \mathcal{K} \backslash \mathcal{A}$;
  \item for $i$ even and $0\leq i \leq g$:
\begin{IEEEeqnarray}{rCl}\label{eq:genieodd}
\vect{V}_{i} & = & \sum_{j=1}^{t_r+r_r+1}  \alpha b_{1, j} \bfN_{\frac{i}{2} \beta_4 -j} \nonumber \\ & & + \sum_{j=1}^{t_\ell+r_\ell+1}  (a_{1, j} + \alpha a_{2,j})  \bfN_{\frac{i}{2} \beta_4 +1+j} \nonumber \\  & & -\bfN_{\frac{i}{2} \beta_4+1},
\end{IEEEeqnarray}
 and for $i$ odd and $1\leq i \leq g-1$:
 \begin{IEEEeqnarray}{rCl}\label{eq:genieeven}
\bfV_{i}& = & \sum_{j=1}^{t_r+r_r+1} ( b_{1, j} + \alpha b_{2,j})  \bfN_{\frac{i-1}{2} \beta_4 -j}  \nonumber \\ && + \sum_{j=1}^{t_\ell+r_\ell+1}  \alpha a_{1, j}   \bfN_{\frac{i-1}{2} \beta_4 +1+j} - \bfN_{\frac{i-1}{2} \beta_4}. \IEEEeqnarraynumspace
\end{IEEEeqnarray}
  \end{itemize}
Thus, if $\theta_4=1$, 
 \begin{equation}\label{inter2}
  \mathcal{K} \backslash \mathcal{R}_{\mathcal{A}} = 
  \big\{ m \beta_4+1, (m+1)\beta_4\big\} _{m=0}^{\gamma_4-1} \cup \{ \gamma_4\beta_4+1\}.
  \end{equation}

If $\theta_4=0$, we apply  Lemma~\ref{lem:dynMAC} to the choices
\begin{itemize}
\item $q=1$;
\item $g=2\gamma_4-1$;
\item $\mathcal{A} = \bigcup_{m=0}^{\gamma_4-1} \mathcal{A}'(m)$, 
where  $\{\mathcal{A}'(m)\}_{m=0}^{\gamma_4-2}$ are defined in \eqref{eq:Amp} and where
\begin{equation}
 \mathcal{A}'(\gamma_4-1)\triangleq \{ ((\gamma_4-1)\beta_4+r_\ell+2),\ldots, (K-r_r-1)\};
\end{equation} 
\item $\mathcal{B}_1 = \mathcal{K} \backslash \mathcal{A}$;
\item $\{\bfV_{m}\}_{m=0}^{2(\gamma_4-1)}$ are given by \eqref{eq:genieodd} and \eqref{eq:genieeven} and  
 \begin{IEEEeqnarray}{rCl}\label{eq:genie2gamma}
\bfV_{2\gamma_4-1}& = & \sum_{j=1}^{t_r+r_r+1} ( b_{1, j} + \alpha b_{2,j})  \bfN_{K-j}  -\bfN_{K}.\IEEEeqnarraynumspace
\end{IEEEeqnarray}
\end{itemize}
Thus, if $\theta_4=0$, 
 \begin{IEEEeqnarray}{rCl}\label{inter2b}
 \lefteqn{
  \mathcal{K} \backslash \mathcal{R}_{\mathcal{A}} } \nonumber \\
  & = &\big\{ m \beta_4+1, (m+1)\beta_4\big\} _{m=0}^{\gamma_4-2} \cup \big\{(\gamma_4-1)\beta_4+1, K \big\}.\nonumber \\
  \end{IEEEeqnarray}

One readily verifies that both for $\theta_4=0$ and $\theta_4=1$ the differential entropy $h\big( \{\vect{N}_k\}_{k\in\mathcal{R}_\mathcal{A}}| \vect{V}_0,\ldots, \vect{V}_q\big)$ is finite and does not depend on the power constraint $P$, since neither does the genie-information. 
In Appendix~\ref{app:rem2} we show that also Assumption~\eqref{eq:assum_fun} in the Dynamic MAC Lemma is satisfied, and hence the lemma applies. It gives the desired upper bound, because by~\eqref{inter2} and \eqref{inter2b},
\begin{equation}
| \mathcal{R}_{\mathcal{A}} |= 2\gamma_4+ \theta_4.
\end{equation}

\subsection{Proof of Upper Bound~\ref{3a}), i.e., \eqref{eq:genub2}}\label{sec:ub2}

The proof is again based on the Dynamic-MAC Lemma~\ref{lem:dynMAC}. 
We first give some definitions. 

Define 
\begin{IEEEeqnarray}{rCl}
\beta_5 & \triangleq & t_\ell+t_r+r_\ell+r_r+3 \\ 
\gamma_5 & \triangleq & \left \lfloor  \frac{K}{ \beta_5} \right \rfloor
\end{IEEEeqnarray}
and recall that $\kappa_5  \triangleq  K- \beta_5 \gamma_5  $ and $\theta_5$ equals 1 if $\kappa_5 \geq t_r +r_r+ 2$ and 0 otherwise.

 For $j_r, j_c \in\{1, \ldots, t_\ell+r+\ell+1\}$, denote the row-$j_r$ column-$j_c$ entry of the matrix $\mat{H}_{t_\ell+r_\ell+1}(\alpha)$ by $h_{j_r,j_c}$. 
Also, choose a set of real numbers $\{d_2,\ldots, d_{t_\ell+r_\ell+1}\}$ so that 
\begin{equation}
h_{1, j_c} = \sum_{j_r=2}^{t_\ell+r_\ell+1} d_{j_r} h_{j_r, j_c}, \quad j_c \in\{1,\ldots, t_\ell+r_\ell+1\}. 
\end{equation}
Such a choice always exists because of the assumption $\det(\mat{H}_{t_\ell+r_\ell+1})=0$.

We treat the cases $\theta_5=1$ and $\theta_5=0$ separately. If $\theta_5=1$, we apply the Dynamic-MAC Lemma to the choices: 
\begin{itemize}
\item $q=2\gamma_5+1$;
\item $g=2 \gamma_5$;
\item $\mathcal{A}= \bigcup_{m=0}^{\gamma_5} \mathcal{A}''(m)$, where 
\begin{equation}\label{eq:Ampp0}
\mathcal{A}''(0)\triangleq\{1,\dots, t_r +1 \},
\end{equation}
for $m\in \{1,\ldots, \gamma_5-1\}$:
\begin{equation}\label{eq:Ampp}
\mathcal{A}''(m)\triangleq\{ m\beta_5 -t_\ell +1 ,\dots, m \beta_5+ t_r +1 \},
\end{equation}
and 
\begin{equation}
\mathcal{A}''(\gamma_5) \triangleq \{ (\gamma_5 \beta_5 -t_\ell+1),\ldots, (K-r_r-1)\};
\end{equation}
\item for 
 $i$ odd and $1 \leq i \leq 2 \gamma_5-1$, 
 \begin{IEEEeqnarray}{rCl}\label{eq:Bodd}
\mathcal{B}_i=  \left\{ \frac{(i-1)}{2} \beta_5 +t_r +r_r +r_\ell+3 \right\};\IEEEeqnarraynumspace
\end{IEEEeqnarray}
 for  $i$ even and $2 \leq i \leq 2 \gamma_5$:
\begin{IEEEeqnarray}{rCl}\label{eq:Beven}
\mathcal{B}_i= & &\left \{ \left( \left( \frac{i}{2} -1\right)\beta_5 +t_r +2 \right), \ldots,  \right. \\ & &\qquad  \left.  \left(\left(\frac{i}{2} -1\right)\beta_5 +t_r +r_r+r_\ell + 2 \right) \right\},  \nonumber \\ 
\end{IEEEeqnarray}
and
\begin{IEEEeqnarray}{rCl}
\mathcal{B}_{2\gamma_5+1}=\{( K-r_r), \ldots, K\};
\end{IEEEeqnarray}
\item  for $i$ even and $0 \leq i \leq 2 (\gamma_5-1)$:\footnote{Recall that $a_{j_r,j_c}$ denotes the row-$j_r$ column-$j_c$ entry of the matrix $\mat{M}_{t_\ell+r_\ell+1}^{\textnormal{inv}}(\alpha)$ defined in the previous Subsection~\ref{sec:ub1}; and where similarly  $b_{j_r,j_c}$ denotes the row-$j_r$ column $j_c$ entry of the matrix $\mat{M}_{t_r+r_r++1}^{\textnormal{inv}}(\alpha)$ also defined in Subsection~\ref{sec:ub1}.}
\begin{IEEEeqnarray}{rCl} \label{eq:genie1}
\vect{V}_{i} &= &\sum_{j=2}^{ t_r+r_r+1} d_j \bfN_{\frac{i}{2} \beta_4 + t_r+r_r+2+j} \nonumber \\  & &  + \alpha   \sum_{ j_c= 1}^{t_r+r_r+1} b_{1,j_c} \bfN_{\frac{i}{2} \beta_4 + t_r+r_r+1-j_c} \nonumber \\  & & - \bfN_{\frac{i}{2} \beta_4 + t_r+r_r+3},
\end{IEEEeqnarray}
for $i$ odd and $1 \leq i \leq 2 \gamma_5-1$: 
\begin{IEEEeqnarray}{rCl}
\vect{V}_{i}& = & \sum_{ j_c= 1}^{t_r+r_r+1} (\alpha b_{2, j_c} + b_{1,j_c} ) \bfN_{\frac{i-3}{2} \beta_4 +t_r+r_r+1-j_c} \nonumber \\  & & +   \sum_{j_c=1}^{t_\ell+r_\ell+1} \alpha a_{1, j_c} \bfN_{\frac{i-3}{2} \beta_4 +t_r+r_r+3+j_c} \nonumber \\ & &- N_{\frac{i-3}{2} \beta_4 + t_r+r_r+2},\label{genie9}
\end{IEEEeqnarray}
and 
\begin{IEEEeqnarray}{rCl}
\vect{V}_{2\gamma_5} & \triangleq & \sum_{ j_c= 1}^{t_r+r_r+1} (\alpha b_{2, j_c} + b_{1,j_c} ) \bfN_{K-1-j_c} - N_{K}.\nonumber \\\IEEEeqnarraynumspace \label{eq:genie10}
\end{IEEEeqnarray}
\end{itemize}

Thus, if $\theta_5=1$, 
\begin{IEEEeqnarray}{rCl}\label{inter3}
\lefteqn{ \mathcal{K}\backslash\mathcal{R}_{\mathcal{A}} } \nonumber \\ & =& \big\{(m \beta_5+ t_r+r_r +2) , (m\beta_5+t_r+r_r+3)\big\}_{m=0}^{\gamma_5-1} \nonumber \\ & &\hspace{2mm}\cup \{K\}.
\end{IEEEeqnarray}

If  $\theta_5=0$, we apply the Dynamic-MAC Lemma to the following choices:
\begin{itemize}
\item $q=2 \gamma_5$;
\item $g=2\gamma_5-1$;
\item $\mathcal{A}= \bigcup_{m=0}^{\gamma_5} \mathcal{A}''(m)$, where $\big\{\mathcal{A}''(m)\big\}_{m=0}^{\gamma_5-1}$ are defined in~\eqref{eq:Ampp0} and \eqref{eq:Ampp}, and where
\begin{equation*}
\mathcal{A}''(\gamma_5) \triangleq \{ (\gamma_5 \beta_5 -t_\ell+1),\ldots, K\};
\end{equation*}
\item the sets $\big\{\mathcal{B}_i\big\}_{i=1}^{2\gamma_5}$ are defined in~\eqref{eq:Bodd} and \eqref{eq:Beven}; 
\item $\{\bfV_{m}\}_{m=0}^{\gamma_5-1}$ is defined in~\eqref{eq:genie1} and \eqref{genie9}.
\end{itemize}

Thus, if $\theta_5=0$, 
\begin{IEEEeqnarray}{rCl}\label{inter3b}
\lefteqn{ \mathcal{K} \backslash\mathcal{R}_{\mathcal{A}} } \nonumber \\ &=&  \big\{(m \beta_5+ t_r+r_r +2) , (m\beta_5+t_r+r_r+3) \big\}_{m=0}^{\gamma_5-1}. \nonumber \\
\end{IEEEeqnarray}

One readily verifies that both for $\theta_5=0$ and $\theta_5=1$ the differential entropy $h\big( \{\vect{N}_k\}_{k\in\mathcal{R}_\mathcal{A}}| \vect{V}_0,\ldots, \vect{V}_q\big)$ is finite and does not depend on the power constraint $P$, since neither does the genie-information. 
In Appendix~\ref{app:dymup} we show that  also Assumption~\eqref{eq:assum_fun} of the Dynamic-MAC Lemma is satisfied, and hence the lemma applies. It gives the desired upper bound, because by~\eqref{inter3} and \eqref{inter3b},
\begin{equation}
| \mathcal{R}_{\mathcal{A}} |= 2\gamma_5+ \theta_5.
\end{equation}

}


\section*{Acknowledgments} 
 This work was supported by the  Israel Science
Foundation (ISF), by the European Commission in the framework
of the FP7 Network of Excellence in Wireless COMmunications NEWCOM\#, and by the city of Paris under the program ``Emergences".
\appendices

\section{Proof of Lemma~\ref{lem:alpha}}\label{sec:alpha}

By definition, $\deta{\mat{H}_{1}(\alpha)}=1$. Therefore, the integer $p$ has to be at least 2 and Statement 1.) in the lemma follows.

Statement 2.) can be proved as follows. We define $\mat{H}_0(\alpha)\triangleq 1$ and note that also  $\mat{H}_1(\alpha)=1,$ irrespective of $\alpha$. We then have for each positive integer $q\geq 2$:
\begin{equation}\label{eq:detrecursion}
\deta{\mat{H}_{q}(\alpha)} = \deta{\mat{H}_{q-1}(\alpha)} - \alpha^2 \deta{\mat{H}_{q-2}(\alpha)}.
\end{equation}
Thus, $\deta{\mat{H}_{p}(\alpha)}=0$ implies  that the two determinants $\deta{\mat{H}_{p-1}(\alpha)}$ and $\deta{\mat{H}_{p-2}(\alpha)}$ are either both~0 or both non-zero, and similarly, that the two determinants $\deta{\mat{H}_{p+1}(\alpha)}$ and $\deta{\mat{H}_{p+2}(\alpha)}$ are either both 0 or both non-zero. Applying this argument iteratively, we see that the determinants
$\deta{\mat{H}_{p-2}(\alpha)}$ and $\deta{\mat{H}_{p-1}(\alpha)}$ can only be 0 if all "previous" determinants $\deta{\mat{H}_{0}(\alpha)}, \ldots, \deta{\mat{H}_{p-3}(\alpha)}$ are zero. Similarly, for the determinants $\deta{\mat{H}_{p+1}(\alpha)}$ and $\deta{\mat{H}_{p+2}(\alpha)}$. However, since $\deta{\mat{H}_0(\alpha)}=\deta{\mat{H}_1(\alpha)}=1$, we conclude that  $\deta{\mat{H}_{p-2}(\alpha)}, \deta{\mat{H}_{p-1}(\alpha)}, \deta{\mat{H}_{p+1}(\alpha)}$, and $\deta{\mat{H}_{p+2}(\alpha)}$ must be non-zero, which proves Statement 2.)

\mw{
\section{Proof that Assumption~\eqref{eq:assum_fun} holds in Section~\ref{sec:ub1}} \label{app:rem2} 

 By~\eqref{inter2} and \eqref{inter2b} it suffices to show that if $\theta_4=0$, then the output sequences $\{{\bfY}_{m\beta_4+1},{\bfY}_{(m+1)\beta_4}\}_{m=0}^{\gamma_4-2}$,
${\bfY}_{(\gamma_4-1)\beta_4+1},$ and ${\bfY}_{K}$ can be reconstructed, and if $\theta_4=1$, then the output sequences  
$\{\bf{Y}_{m\beta_4+1},{Y}^n_{(m+1)\beta_4}\}_{m=0}^{\gamma_4-1}$ and $\bfY_{\gamma_4 \beta_4 +1}$ can be reconstructed. 
   
Notice first that using the given encoding functions $f_1,\ldots, f_n$
the input sequences
$\{{\bfX}_{m\beta_4+t_\ell+r_\ell+2},{\bfX}_{m\beta_4+t_\ell+r_\ell+3}
\}_{m=0}^{ \gamma_4-1}$ can be computed from Messages
$\{M_{k}\}_{k\in\mathcal{A}}$. Moreover, if $\theta_4=0$ then
additionally also the input sequences $\bfX_{ (\gamma_4-1) +
  r_\ell+t_\ell+4}, \ldots, \bfX_{K- r_r - t_r -1}$ can be computed
from $\{M_{k}\}_{k\in\mathcal{A}}$, and if $\theta_4=1$ additionally also the input sequences
${\bfX}_{\gamma_4\beta_4+t_\ell+r_\ell+2},{\bfX}_{K-r_r-t_r-1}$ can be
computed from $\{M_{k}\}_{k\in\mathcal{A}}$. 
The  result is then proved by showing that each of the desired
output sequences can be expressed as a linear combination of the
genie-information, these reconstructed inputs, and the outputs observed
by the group-$\mathcal{A}$ receivers. 

We start with $\bfY_{\beta_4}$. 
Notice that by the channel law~\eqref{eq:symchannel}, the linear systems \eqref{eq:lin10} and \eqref{eq:lin11} on top of the next page hold for every time-instant $t\in\{1,\ldots, n\}$.
\begin{figure*}
\begin{IEEEeqnarray}{rCl}\label{eq:lin10}
\begin{pmatrix}
  Y_{\beta_4- 1,t} \\ Y_{\beta_4-2,t} \\ \vdots \\ Y_{t_\ell+r_\ell+4,t} \\ Y_{t_\ell+r_\ell+3,t}
\end{pmatrix}  = \mat{M}_{t_r+r_r+1}(\alpha) \begin{pmatrix} X_{\beta_4,t}\\  X_{\beta_4-1,t} \\ \vdots\\ X_{t_\ell+r_\ell+5,t}\\  X_{t_\ell+r_\ell+4,t}  \end{pmatrix} + \begin{pmatrix}  0 \\ 0 \\ \vdots \\ \alpha X_{t_\ell+r_\ell+3} \\   \alpha X_{t_\ell+r_\ell+2,t} + X_{t_\ell+r_\ell+3,t}  \end{pmatrix} + \begin{pmatrix}
N_{\beta_4-1,t}\\ N_{\beta_4-2,t}  \\ \vdots \\ N_{t_\ell+r_\ell+4,t} \\N_{t_\ell+r_\ell+3,t} 
\end{pmatrix}  
\end{IEEEeqnarray}
\begin{IEEEeqnarray}{rCl}\label{eq:lin11}
\begin{pmatrix}
Y_{\beta_4+2,t} \\ Y_{\beta_4+3,t} \\ \vdots \\  Y_{\beta_4+t_\ell+r_\ell+1,t} \\ Y_{\beta_4+t_\ell+r_\ell+2,t} 
\end{pmatrix}  = \mat{M}_{t_\ell+r_\ell+1}(\alpha) \begin{pmatrix} X_{\beta_4+1,t}\\  X_{\beta_4+2,t} \\ \vdots\\ X_{\beta_4+t_\ell+r_\ell,t}\\  X_{\beta_4+t_\ell+r_\ell+1,t}  \end{pmatrix} + \begin{pmatrix}  0 \\ 0 \\ \vdots \\ \alpha X_{\beta_4+t_\ell+r_\ell+2} \\   X_{\beta_4+t_\ell+r_\ell+2,t} + \alpha X_{\beta_4+t_\ell+r_\ell+3,t}  \end{pmatrix} + \begin{pmatrix}
N_{\beta_4+2,t}\\ N_{\beta_4+3,t}  \\ \vdots \\ N_{\beta_4+t_\ell+r_\ell+1,t} \\N_{\beta_4+t_\ell+r_\ell+2,t} 
\end{pmatrix}  \nonumber \\
\end{IEEEeqnarray}
\hrulefill
\end{figure*}
Recalling that $a_{j_r,j_c}$ denotes the row-$j_r$ column-$j_c$ entry
of the inverse matrix
$\mat{M}_{t_\ell+r_\ell+1}^{\textnormal{inv}}(\alpha)$ and that
$b_{j_r,j_c}$ denotes the row-$j_r$ column-$j_c$ entry of the
inverse matrix $\mat{M}_{t_r+r_r+1}^{\textnormal{inv}}(\alpha)$, it is
easily checked that \eqref{eq:lin10} implies:
\begin{IEEEeqnarray}{rCl} 
\lefteqn{
 \sum_{j=1}^{t_r+r_r+1}  b_{2,j}   \vect{Y}_{\beta_4 -j} - (b_{2, t_r+r_r} \alpha +  b_{2, t_r+r_r+1})\vect{X}_{t_\ell+r_\ell+3}} \nonumber \qquad \\ & & \hspace{1cm}  - b_{2, t_r+r_r+1} \alpha \vect{X}_{t_\ell+r_\ell+2 }\nonumber \\& = & \bfX_{\beta_4-1} +   \sum_{j=1}^{t_r+r_r+1}  b_{2,j}  \vect{N}_{\beta_4 -j}; \hspace{3cm}\label{eq:X1}
 \end{IEEEeqnarray}
and
 \begin{IEEEeqnarray}{rCl}
 \lefteqn{
 \sum_{j=1}^{t_r+r_r+1}  b_{1, j}  \vect{Y}_{\beta_4 -j}  -
( b_{1, t_r+r_r} \alpha +  b_{1, t_r+r_r+1} ) \vect{X}_{t_\ell+r_\ell+3}} \qquad \nonumber \\ & & \hspace{1cm}   -  b_{1, t_r+r_r+1}  \alpha\vect{X}_{t_\ell+r_\ell+2 }\nonumber \\& = & \bfX_{\beta_4  } +  \sum_{j=1}^{t_r+r_r+1}  b_{1,j} \vect{N}_{\beta_4 -j};    \hspace{3cm}\label{eq:X2}
 \end{IEEEeqnarray}
 and that \eqref{eq:lin11} implies:
 \begin{IEEEeqnarray}{rCl}
\lefteqn{\sum_{j=1}^{t_\ell+r_\ell+1}  a_{1, j}    \vect{Y}_{ \beta_4 +1+j} -   a_{1, t_r+r_r+1} \alpha \vect{X}_{\beta_4+t_\ell+r_\ell+3}  } \qquad\nonumber \\ & & -( a_{1,t_r+r_r} +   a_{1,t_r+r_r+1} \alpha) \vect{X}_{\beta_4+t_\ell+r_\ell+2}   \nonumber \\  & = & \bfX_{\beta_4+1}+\sum_{j=1}^{t_\ell+r_\ell+1}  a_{1, j}    \vect{N}_{ \beta_4 +1+j}. \label{eq:X3}
\end{IEEEeqnarray}
Since the genie-information has been chosen so that
\begin{IEEEeqnarray}{rCl}\label{eq:Ybeta4}
\vect{Y}_{\beta_4} & = &     \alpha \left( \bfX_{\beta_4-1} +   \sum_{j=1}^{t_r+r_r+1}  b_{2,j}  \vect{N}_{\beta_4 -j}\right) \nonumber \\ & & + \left(\bfX_{\beta_4  } +  \sum_{j=1}^{t_r+r_r+1}  b_{1,j} \vect{N}_{\beta_4 -j} \right) \nonumber \\ & &+ \alpha \left(  \bfX_{\beta_4+1}+\sum_{j=1}^{t_\ell+r_\ell+1}  a_{1, j}    \vect{N}_{ \beta_4 +1+j}\right) - \bfV_{1}\nonumber \\
 \end{IEEEeqnarray}
 the desired linear combination representing $\bfY_{\beta_4}$ is
  obtained by combining the linear combinations on the left-hand
 sides of Equations \eqref{eq:X1}--\eqref{eq:X3} with the
 genie-information $\bfV_1$.

We next consider $\bfY_{\beta_4+1}$. 
By~\eqref{eq:lin11}, 
 \begin{IEEEeqnarray}{rCl}
\lefteqn{\sum_{j=1}^{t_\ell+r_\ell+1}  a_{1, j}    \vect{Y}_{ \beta_4 +1+j} -   a_{2, t_r+r_r+1} \alpha \vect{X}_{\beta_4+t_\ell+r_\ell+3}  } \qquad\nonumber \\ & & -( a_{2,t_r+r_r} +   a_{2,t_r+r_r+1} \alpha) \vect{X}_{\beta_4+t_\ell+r_\ell+2}   \nonumber \\  & = & \bfX_{\beta_4+2}+\sum_{j=1}^{t_\ell+r_\ell+1}  a_{2, j}    \vect{N}_{ \beta_4 +1+j}.\label{eq:X4}
\end{IEEEeqnarray}
Since the genie-information $\bfV_2$ has been chosen so that
\begin{IEEEeqnarray}{rCl}\label{eq:Ybeta41}
\vect{Y}_{\beta_4+1} & = &     \alpha \left( \bfX_{\beta_4} +
  \sum_{j=1}^{t_r+r_r+1}  b_{1,j}  \vect{N}_{\beta_4 -j}\right)
\nonumber \\ & & + \left(\bfX_{\beta_4+1  } +
  \sum_{j=1}^{t_\ell+r_\ell+1}  a_{1,j} \vect{N}_{\beta_4 +1+j}
\right) \nonumber \\ & &+\alpha \left(
  \bfX_{\beta_4+2}+\sum_{j=1}^{t_\ell+r_\ell+1}  a_{2, j}
  \vect{N}_{ \beta_4 +1+j}\right) - \bfV_{2}\nonumber \\ 
\end{IEEEeqnarray}
the desired linear combination representing $\bfY_{\beta_4+1}$ is
 obtained by combining the left-hand sides of \eqref{eq:X2},
\eqref{eq:X3}, and \eqref{eq:X4} with the genie-information $\bfV_2$.

The desired linear combinations representing the outputs
$\{\bfY_{m\beta_4}\}_{m=2}^{\gamma_4-1+ \theta_4}$
can be obtained from the
equations that result when in \eqref{eq:X1}--\eqref{eq:Ybeta4}  each vector
$\bfX_k$, for $k\in\{1,\ldots, K\}$, is replaced by
$\bfX_{k+(m-1)\beta_4}$, each vector $\bfY_k$ by
$\bfY_{k+(m-1)\beta_4}$,  each vector $\bfN_k$ by
$\bfN_{k+(m-1)\beta_4}$, and the genie-information $\bfV_{1}$ is
replaced by $\bfV_{2m-1}$. 

The  linear combinations representing the outputs 
$\{\bfY_{m\beta_4+1}\}_{m=2}^{\gamma_4-1+\theta_4}$
are obtained from
the equations that result when in \eqref{eq:X2}, \eqref{eq:X3},
\eqref{eq:X4}, and \eqref{eq:Ybeta41} the vectors
$\bfX_k$, $\bfY_k$, and $\bfN_k$, for $k\in\{1,\ldots, K\}$, are
replaced by the vectors
$\bfX_{k+(m-1)\beta_4}$,$\bfY_{k+(m-1)\beta_4}$, and
$\bfN_{k+(m-1)\beta_4}$ and the genie-information $\bfV_{2}$ is
replaced by $\bfV_{2m}$. When $m=0$ all the out-of-range indices should be ignored, that means, $\bfX_{k}$, $\bfY_k$,
$\bfN_k$ are assumed to be deterministically 0 for all $k\leq 0$. 

Finally, if $\theta_4=0$, then the desired  linear combination representing $\bfY_K$ can
be obtained 
by combining the
equations that result when in Equations
\eqref{eq:X1}, \eqref{eq:X2}, and \eqref{eq:Ybeta4} the vectors
$\bfX_k$, $\bfY_k$, and $\bfN_k$ are
replaced by the vectors
$\bfX_{K-\beta_4}$,$\bfY_{K-\beta_4}$, and
$\bfN_{K-\beta_4}$ and the genie-information $\bfV_{1}$ is
replaced by $\bfV_{2\gamma_4-1}$. Again, all out-of-range indices
should be ingored, i.e., $\bfX_{k}$, $\bfY_k$,
$\bfN_k$ are assumed to be deterministically 0 for all $k> K$.


\section{Proof that Assumption~\eqref{eq:assum_fun} holds in Section~\ref{sec:ub2}}\label{app:dymup}

Notice that for $i\in\{1,\ldots, 2\gamma_5\}$ odd, 
\begin{equation}
\mathcal{R}_{\mathcal{B}_i}  \backslash (\mathcal{R}_{\mathcal{B}_i} \cap \mathcal{R}_{\mathcal{A}_i}) = \left\{  \frac{(i-1)}{2} \beta_5+t_r+r_r+3\right\}, 
\end{equation}
 and for $i$ even,
\begin{equation}
\mathcal{R}_{\mathcal{B}_i}  \backslash (\mathcal{R}_{\mathcal{B}_i} \cap \mathcal{R}_{\mathcal{A}_i}) = \left\{\frac{(i-1)}{2} \beta_5 +t_r+r_r+2+2\right\},
\end{equation}
and moreover, if $\theta_5=1$, 
\begin{equation}
\mathcal{R}_{\mathcal{B}_{2\gamma_5+1}}  \backslash (\mathcal{R}_{\mathcal{B}_{2\gamma_5+1}} \cap \mathcal{R}_{\mathcal{A}_{2\gamma_5+1}}) = \{K\}. 
\end{equation}
Thus, for $i\leq 2\gamma_5-1$ odd we need to show that the output sequence $\bfY_{\frac{(i-1)}{2} \beta_5+t_r+r_r+3}$ can be  reconstructed from the messages $\{M_{k} \}_{k\in\mathcal{A}_i}$, the outputs $\{\bfY_{k}\}_{k\in\mathcal{R}_{\mathcal{A}_i}}$, and the genie-information $\{\bfV_m\}_{m=0}^g$. Similarly, for $i$ even we need to show that  $\bfY_{ \frac{(i-1)}{2} \beta_5 +t_r+r_r+2}$ can be reconstructed, and for $i=2\gamma_5+1$, we need to show that $\bfY_K$ can be reconstructed. 

%


Using the encoding functions $f_1,\ldots, f_n$, for each
$i$ that is odd and satisfies $1\leq i \leq 2\gamma_5-1$  the inputs $\bfX_{\frac{i-1}{2} \beta_5}$, $\bfX_{\frac{i-1}{2}
    \beta_5+1}$, and $\bfX_{\frac{i+1}{2} \beta_5}$ can be computed from
      the messages $\{M_{k}\}_{k\in\mathcal{A}_{i}}$.  For each $i$
      that is even and that satisfies $2\leq i \leq 2\gamma_5$ the inputs $\bfX_{\frac{i-1}{2}
        \beta_5}$,$\bfX_{\frac{i-1}{2} \beta_5+1}$, $\bfX_{\frac{i+1}{2}
            \beta_5}$, and $\bfX_{\frac{i+1}{2} \beta_5+1}$ can be
              computed from messages $\{M_{k}\}_{k\in
                \mathcal{A}_{i}}$. Finally, if $\theta_5=1$, then
                inputs $X_{K-t_\ell-r_\ell-2}$ and
                $X_{K-t_\ell-r_\ell-1}$ can be computed from the messages   $\{M_{k}\}_{k\in
                \mathcal{A}_{2\gamma_5+1}}$. 

We start with $i=1$ and outputs $\bfY_{t_r+r_r+ 3}$.  
By the channel law \eqref{eq:symchannel}, the linear system \eqref{eq:linsym} on top of the next page holds for every time $t\in\{1,\ldots, n\}$. 
\begin{figure*}
\begin{equation}\label{eq:linsym}
\begin{pmatrix}{Y}_{t_r+r_r+3,t}\\ {Y}_{t_r+r_r+4,t}\\\vdots\\ {Y}_{\beta_5-1,t}\\ {Y}_{\beta_5,t}\end{pmatrix}=
H_{t_\ell+r_\ell+1} (\alpha)\begin{pmatrix} {X}_{t_r+r_r+3,t}\\  {X}_{t_r+r_r+4,t}\\\vdots\\ {X}_{\beta_5-1,t}\\  {X}_{\beta_5,t}\end{pmatrix}+
\begin{pmatrix}\alpha  {X}_{t_r+r_r+2,t}\\0\\\vdots\\0\\\alpha {X}_{\beta_5+1,t}\end{pmatrix}+
\begin{pmatrix}\vect{N}_{t_r+r_r+3,t}\\\vect{N}_{t_r+r_r+4,t}\\\vdots\\\vect{N}_{\beta_5-1,t}\\\vect{N}_{\beta_5}\end{pmatrix}
\end{equation}
\end{figure*}
Recalling the definition of the parameters $\{d_2,\ldots, d_{t_\ell+r_\ell+1}\}$ in Section~\ref{sec:ub2} and because $\deta{\mat{H}_{t_\ell+r_\ell+1}(\alpha)} =0$,  \eqref{eq:linsym} implies:
 \begin{IEEEeqnarray}{rCl} \label{eq:outputlin}
{\bfY}_{t_r+r_r+3} & = &\sum_{j=2}^{t_\ell+r_\ell+1} d_j \left(\bfY_{t_r+r_r+2+j}-\bfN_{t_r+r_r+2+j} \right) \nonumber  \\ & & - \alpha d_{t_\ell+r_\ell+1} \bfX_{\beta_5+1} + \alpha \bfX_{t_r+r_r+2} +\bfN_{t_r+r_r+3}. \nonumber \\
 \end{IEEEeqnarray}
We next notice that by the channel law~\eqref{eq:symchannel}, for every
time $t\in\{1,\ldots, n\}$, the linear system in \eqref{eq:lin1c}
holds, where the matrix $\mat{M}_{t_r+r_r+1}(\alpha)$ is defined in Section~\ref{sec:ub1}. 
\begin{figure*}
\begin{IEEEeqnarray}{rCl}\label{eq:lin1c}
\begin{pmatrix}
  Y_{t_r+r_r+1,t} \\ Y_{t_r+r_r,t} \\ \vdots \\ Y_{2,t} \\ Y_{1,t}
\end{pmatrix}  = \mat{M}_{t_r+r_r+1}(\alpha) \begin{pmatrix} X_{t_r+r_r+2,t}\\  X_{t_r+r_r+1,t} \\ \vdots\\ X_{3,t}\\  X_{2,t}  \end{pmatrix} + \begin{pmatrix}  0 \\ 0 \\ \vdots \\ \alpha X_{1,t} \\  X_{1,t}  \end{pmatrix} + \begin{pmatrix}
N_{t_r+r_r+1,t}\\ N_{t_r+r_r,t}  \\ \vdots \\ N_{2,t} \\N_{1,t} 
\end{pmatrix}  
\end{IEEEeqnarray}
\hrulefill
\end{figure*}
Recalling that $b_{j_r, j_c}$ denotes the row-$j_r$ column-$j_c$ entry
of the inverse $\mat{M}_{t_r+r_r+1}^{\textnormal{inv}}(\alpha)$,
Equation~\eqref{eq:lin1c} implies:
\begin{IEEEeqnarray}{rCl}\label{eq:lin2} 
\lefteqn{ \sum_{ j_c= 1}^{t_r+r_r+1} b_{1,j_c}  \bfY_{t_r+r_r+2-j_c}  -( b_{1, t_r+r_r+1} + \alpha b_{1,t_r+r_r})\bfX_1}\nonumber \hspace{1cm}\\ 
 & = & \bfX_{t_r+r_r+2 } +  \sum_{ j_c= 1}^{t_r+r_r+1} b_{1,j_c} \bfN_{t_r+r_r+2-j_c} .  \IEEEeqnarraynumspace
\end{IEEEeqnarray}
Finally, by the definition of the genie-information $\bfV_0$,  
 combining \eqref{eq:outputlin} with \eqref{eq:lin2} yields the desired
linear combination
 \begin{IEEEeqnarray}{rCl} \label{eq:outputlin2}
\lefteqn{{\bfY}_{t_r+r_r+3}}\quad  \nonumber \\  & = &\sum_{j=2}^{t_\ell+r_\ell+1} d_j \bfY_{t_r+r_r+2+j}  + \alpha  \sum_{ j_c= 1}^{t_r+r_r+1} b_{1,j_c}  \bfY_{t_r+r_r+2-j_c} \nonumber  \\ & &  -( b_{1, t_r+r_r+1} + \alpha b_{1,t_r+r_r})\bfX_1 \nonumber \\ & & - \alpha d_{t_\ell+r_\ell+1} \bfX_{\beta_5+1}  -\bfV_0.
 \end{IEEEeqnarray}

For each $i$ odd and $3\leq i\leq 2\gamma_5-1$ the desired linear
combination representing $\bfY_{\frac{i-1}{2} \beta_5+ t_r+r_r+3}$ can
be found in a similar way. 
Specifically, using Equations similar to
\eqref{eq:linsym}--\eqref{eq:outputlin2} one can show that 
 \begin{IEEEeqnarray}{rCl} \label{eq:outputlingen}
{\bfY}_{\frac{i-1}{2} \beta_5+t_r+r_r+3}
& = &\sum_{j=2}^{t_\ell+r_\ell+1} d_j \bfY_{\frac{i-1}{2}
  \beta_5+t_r+r_r+2+j} \nonumber \\  & &  + \alpha  \sum_{ j_c= 1}^{t_r+r_r+1} b_{1,j_c}
\bfY_{\frac{i-1}{2} \beta_5+t_r+r_r+1-j_c} \nonumber  \\ & &  -( b_{1,
  t_r+r_r+1} + \alpha b_{1,t_r+r_r})\bfX_{\frac{i-1}{2} \beta_5+1}
\nonumber \\ & &- 
\alpha b_{1,t_r+r_r+1} \bfX_{\frac{i-1}{2}\beta_5}    \nonumber \\ & & - \alpha \beta_{t_\ell+r_\ell+1} \bfX_{\frac{i+1}{2}
  \beta_5+1}  -\bfV_{i-1}.
 \end{IEEEeqnarray}

We next consider the case where $i$ is even and $2 \leq i \leq 2\gamma_5$, where we wish to reconstruct 
$\bfY_{(\frac{i}{2} -1)\beta_5+ t_r+r_r+2}$. 
The construction of the desired linear combination is similar to 
Appendix~\ref{app:rem2}, that means it is  based on equations
that are similar to equations
\eqref{eq:X1}--\eqref{eq:Ybeta4}. Obviously, \eqref{eq:lin10} remains
valid if for each $k\in\{1,\ldots, K\}$ the symbols $X_{k,t}$, $Y_{k,t}$, and $N_{k,t}$
are replaced by $X_{k+(\frac{i}{2} -1)\beta_5-(t_\ell+r_\ell+2),t}$,
$Y_{k+(\frac{i}{2} -1)\beta_5-(t_\ell+r_\ell+2),t}$, and $N_{k+(\frac{i}{2} -1)\beta_5-(t_\ell+r_\ell+2),t}$, and
therefore similar to \eqref{eq:X1} and \eqref{eq:X2} we obtain:
\begin{IEEEeqnarray}{rCl} 
\lefteqn{
 \sum_{j=1}^{t_r+r_r+1}  b_{2,j}   \vect{Y}_{(\frac{i}{2} -1)\beta_5+t_\ell+r_\ell+2-j} - b_{2, t_r+r_r+1} \alpha
\vect{X}_{(\frac{i}{2} -1)\beta_5+1  } } \qquad \nonumber \\  & &  -
 (b_{2, t_r+r_r} \alpha +  b_{2, t_r+r_r+1})\vect{X}_{(\frac{i}{2} -1)\beta_5+1} \nonumber
\\& = & \bfX_{t_r+r_r+1} +   \sum_{j=1}^{t_r+r_r+1}  b_{2,j}  \vect{N}_{(\frac{i}{2} -1)\beta_5+1t_r+r_r+2 -j}   \hspace{0.7cm}\label{eq:X1f}
 \end{IEEEeqnarray}
and
 \begin{IEEEeqnarray}{rCl}
 \lefteqn{
 \sum_{j=1}^{t_r+r_r+1}  b_{1, j}  \vect{Y}_{(\frac{i}{2} -1)\beta_5+t_\ell+r_\ell+2 -j}   -  b_{1, t_r+r_r+1} \alpha\vect{X}_{(\frac{i}{2} -1)\beta_5+1   }} \quad \nonumber \\ 
& & - ( b_{1, t_r+r_r} \alpha +  b_{1, t_r+r_r+1} ) \vect{X}_{(\frac{i}{2} -1)\beta_5+1} 
 \nonumber
 \\& = & \bfX_{(\frac{i}{2} -1)\beta_5+t_r+r_r+2 } +  \sum_{j=1}^{t_r+r_r+1}  b_{1,j} \vect{N}_{(\frac{i}{2} -1)\beta_5+t_r+r_r+2 -j}.  \nonumber \\ \label{eq:X2f}
 \end{IEEEeqnarray}
Since also \eqref{eq:lin11} remains valid if for each $k\in\mathcal{K}$ the symbols $X_{k,t}$, $Y_{k,t}$, and $N_{k,t}$
are replaced by $X_{k+(\frac{i}{2} -1)\beta_5-(t_\ell+r_\ell+3),t}$,
$Y_{k+(\frac{i}{2} -1)\beta_5-(t_\ell+r_\ell+3),t}$, and $N_{k+(\frac{i}{2} -1)\beta_5-(t_\ell+r_\ell+3),t}$, we
obtain similarly to \eqref{eq:X3}:
 \begin{IEEEeqnarray}{rCl}
\lefteqn{\sum_{j=1}^{t_\ell+r_\ell+1}  a_{1, j}    \vect{Y}_{ (\frac{i}{2} -1)\beta_5+t_r+r_r +2+j} -   a_{1, t_r+r_r+1} \alpha \vect{X}_{(\frac{i}{2} +1)\beta_5+1}  } \quad\nonumber \\ & & -( a_{1,t_r+r_r} +   a_{1,t_r+r_r+1} \alpha) \vect{X}_{(\frac{i}{2} -1)\beta_5}   \nonumber \\  & = & \bfX_{(\frac{i}{2} -1)\beta_5+t_r+r_r+3}+\sum_{j=1}^{t_\ell+r_\ell+1}  a_{1, j}    \vect{N}_{(\frac{i}{2} -1)\beta_5+ t_r+r_r+2+j}. \label{eq:X3ff}\nonumber \\ 
\end{IEEEeqnarray}
Now, since the genie-information $\bfV_{i-1}$ has been chosen so that Equality~\eqref{eq:Ybeta4df} on top of the next page holds, 
\begin{figure*}
\begin{IEEEeqnarray}{rCl}\label{eq:Ybeta4df}
\vect{Y}_{(\frac{i}{2} -1)\beta_5+  t_r+r_r+2}  & = &     \alpha \left( \bfX_{(\frac{i}{2} -1)\beta_5+t_r+r_r+1} +   \sum_{j=1}^{t_r+r_r+1}  b_{2,j}  \vect{N}_{(\frac{i}{2} -1)\beta_5+t_r+r_r+2-j}\right) \nonumber \\ & & + \left(\bfX_{(\frac{i}{2} -1)\beta_5+t_r+r_r+2  } +  \sum_{j=1}^{t_r+r_r+1}  b_{1,j} \vect{N}_{(\frac{i}{2} -1)\beta_5+t_r+r_r+2 -j} \right) \nonumber \\ && + \alpha \left(  \bfX_{(\frac{i}{2} -1)\beta_5+t_r+r_r+3}+\sum_{j=1}^{t_\ell+r_\ell+1}  a_{1, j}    \vect{N}_{ (\frac{i}{2} -1)\beta_5+t_r+r_r+1+j}\right) - \bfV_{i-1}\IEEEeqnarraynumspace
 \end{IEEEeqnarray}
 \hrulefill
 \end{figure*}
 the desired linear combination representing $\bfY_{(\frac{i}{2} -1)\beta_5+t_\ell+r_\ell+2}$ can be obtained by combining \eqref{eq:X1f}--\eqref{eq:Ybeta4df}. 
 
 If $\theta_5=1$, then the desired linear combination representing $\bfY_K$ can be found in a similar manner as in the previous Appendix~\ref{app:rem2}. The details are omitted.

}



\section{Proof of Proposition~\ref{prop:poweroffset}}\label{sec:poweroffset}

\begin{lemma}
\label{6thouless}
For an integer $p$ and a real number $\alpha$, denote $u_p(\alpha) = \deta{\mat{H}_p(\alpha)}$. Then the following holds.
\begin{enumerate}
	\item $u_p(\alpha)$ is a polynomial in $\alpha$, $u_p(0)=1$, and it satisfies the following second order recursion:
	\begin{align}
	\label{6thouless1}
	u_{p+2}(\alpha)=u_{p+1}(\alpha)-\alpha^2 u_p(\alpha),
	\end{align}
	with the initial conditions $u_0(\alpha)=u_1(\alpha)=1$. We denote by $E_{p}$ the set of roots of $u_p(\alpha)$.
	\item For $\alpha\neq0$, define
	\[v_p(\alpha)\triangleq\frac{u_p(\alpha)}{\left(-\alpha\right)^p}.\]
	Then $v_p(\alpha)$ satisfies the second order recursion:
	\begin{align}
	\label{6thouless2}
	v_{p+2}(\alpha)=-\frac{1}{\alpha}v_{p+1}(\alpha)- v_p(\alpha),
	\end{align}
	with the initial conditions $v_{-1}(\alpha)=0$ and $v_0(\alpha)=1$.
	Moreover, for all $p\geq1$ and $l\geq0$,
\begin{align}\label{6thouless3}
	&\begin{pmatrix}v_i&\cdots&v_{l+p-1}\end{pmatrix}\mat{H}_p=\nonumber \\
	&\qquad\qquad\begin{pmatrix}-\alpha v_{l-1}&0&\cdots&0&-\alpha v_{l+p} \end{pmatrix}
	\end{align}
	where for simplicity we wrote $v_l$ for $v_l(\alpha)$.
	
\end{enumerate}
\end{lemma}
\mw{
\begin{proof}
Omitted. 
%
\end{proof}

We give a proof of Proposition~\ref{prop:poweroffset} for the case where $q$ is odd. The  case $q$ even goes along the same lines. We define $\gamma''' \triangleq (q-1)/2$ and
\begin{align*}
L&\triangleq r_\ell+t_\ell\\
\beta''' &\triangleq 2L+4.
\end{align*}

The first part of the proof follows the first part of the proof of the Dynamic-MAC Lemma, see Section \ref{sec:converse}. We construct a Cognitive MAC as in Section~\ref{sec:converse} using  parameters
\begin{itemize}
\item $q=2$;
\item $g=2\gamma'''$;
\item $\mathcal{A} = \bigcup_{m=1}^{\gamma'''} \mathcal{A}'''(m)$ where 
\[\mathcal{A}'''(0)\triangleq\{r_\ell+2,\dots,L+t_r+2\},\]
for $1\leq m\leq\gamma'''-1$,
\[\mathcal{A}'''(m)\triangleq\{m\beta''' +r_\ell+1,\dots,m\beta''' +L+t_r+2\},\]
and 
\[\mathcal{A}'''(\gamma''')\triangleq\{\gamma'''\beta'''+r_\ell+1,\dots,K\};\]
\item $\mathcal{B}_1=\{r_\ell+1\}$ and $\mathcal{B}_2=\mathcal{K} \backslash (\mathcal{A}\cup\mathcal{B}_1)$; 
\item the genie-information 
\[\vect{V}_0\triangleq-\alpha v_{L+1} \vect{X}_{L+1}+\sum_{j=0}^{L}v_j\vect{N}_{j+1},\]
and where the rest of the genie-informations $\{\vect{V}_i\}_{i=1}^{2\gamma'''}$ is similar to the genie-information described in~\eqref{eq:genieodd} and \eqref{eq:genieeven}.
\end{itemize}
By the  choice above,  
\begin{equation} \mathcal{K} \backslash \mathcal{R}_{\mathcal{A}} = \{ 1\} \cup \big\{ m\beta'''-1, m\beta'''\big\}_{ m=1}^{\gamma'''}.
\end{equation}


 Notice that unlike in the proof in Section~\ref{sec:ub1}, here, part of the genie-information depends on the transmitted signal $\vect{X}_{L+1}$. (But notice that the signal to noise ratio of $\vect{X}_{L+1}$ with respect to $\sum_{j=0}^{L}v_j\vect{N}_{j+1}$ goes to 0 like $\left(\alpha-\alpha^*\right)^\nu$ as  $\alpha$ goes to $\alpha^*$.)
 
Our choice of parameters satisfies Assumption~\eqref{eq:assum_fun} in the Dynamic-MAC Lemma, and thus we can follow the steps in the proof of~\eqref{eq:supseteq} to deduce that the capacity region of the original network is included in the capacity region of the Cognitive MAC. 
That Assumption~\eqref{eq:assum_fun} is satisfied for $i=1$ follows because from the messages $\{M_{k}\}_{k\in\mathcal{A}}$ one can reconstruct $\vect{X}_{L+2}$, and because by 
\begin{IEEEeqnarray*}{rCl}\begin{pmatrix}\vect{Y}_1\\\vdots\\\vect{Y}_{L}\\{\vect{Y}}_{L+1}\end{pmatrix}
&=& H_{L+1}\begin{pmatrix}\vect{X}_1\\\vdots\\ \vdots \\\vect{X}_{L+1}\end{pmatrix}+\begin{pmatrix}\vect{N}_1\\ \vdots \\\vdots\\\vect{N}_{L+1}\end{pmatrix} + \alpha \begin{pmatrix} 0 \\ \vdots \\ 0 \\ \vect{X}_{L+2}\end{pmatrix} 
  \end{IEEEeqnarray*}
and by  (\ref{6thouless3}), applied to $p=L+1$ and $l=0$, 
\begin{IEEEeqnarray*}{rCl}
\sum_{j=0}^{L}v_j\vect{Y}_{j+1}&=&\alpha v_{L}\vect{X}_{L+2}-\alpha v_{L+1} \vect{X}_{L+1}+\sum_{j=0}^{L}v_j\vect{N}_{j+1}\nonumber \\
&= & \alpha v_L \vect{X}_{L+2} + \vect{V}_0,
\end{IEEEeqnarray*}
and thus it is possible to reconstruct $\bfY_1$. 

For $i=2$, Assumption~\eqref{eq:assum_fun} follows by similar considerations as in Appendix~\ref{app:rem2}. Appendix~\ref{app:rem2} also shows how to choose the genie-signals $\{\bfV_m\}_{m=1}^{2\gamma'''}$. 

}

Let us now bound the sum-capacity of the Cognitive MAC:
\begin{align*}
n\Ca_{\textnormal{MAC},\Sigma}&\leq I\left(\left\{\vect{Y}_i\right\}_{i\in\mathcal{A}'''},\left\{\vect{V}_{i}\right\}_{0\leq i\leq2\gamma'''};M_1\dots,M_K\right)\\
&=I\left(\left\{\vect{Y}_i\right\}_{i\in\mathcal{A}'''};M_1,\dots,M_K|\left\{\vect{V}_{i}\right\}_{0\leq i\leq2\gamma'''}\right)\\
&\qquad +I\left(\left\{\vect{V}_{i}\right\}_{0\leq i\leq2\gamma'''};M_1,\dots,M_K\right).
\end{align*}

We deal with each term separately.
\begin{align*}
&I(\left\{\vect{Y}_i\right\}_{i\in\mathcal{A}};M_1,\dots,M_K|\left\{\vect{V}_{i}\right\}_{0\leq i\leq2\gamma'''})\\
&\leq \sum_{i\in\mathcal{A}}h(\vect{Y}_{i})-h\bigg(\{\vect{N}_i\}_{i\in\mathcal{A}}\bigg|\left\{\vect{V}_{i}\right\}_{1\leq i\leq2\gamma'''},\sum_{j=0}^{L}v_j\vect{N}_{j+1}\bigg)\\
&\leq n(K-2\gamma'''-1)\frac{1}{2}\log(P)+n f_1(P,\alpha),
\end{align*}
where $f_1$ is such that $\lim_{\alpha\rightarrow\alpha_0}\lim_{P\rightarrow\infty}f_1(P,\alpha)$ exists and is finite.

Moreover, as can be verified, the genie-information $\left\{\vect{V}_i\right\}_{1\leq i\leq2\gamma'''}$ is independent of $(\vect{V}_0 ,M_1,\dots,M_K)$, and 
\begin{align}
&I(\left\{\vect{V}_{i}\right\}_{0\leq i\leq2\gamma'''};M_1,\dots,M_K) \nonumber \\
&\qquad=I(\vect{V}_0;M_1,\dots,M_K)\nonumber\\
&\qquad\leq n\frac{1}{2}\log\left(1+\frac{P\alpha^2 v_{L+1}^2(\alpha)}{\left\|\begin{pmatrix}v_0&\cdots&v_{L}\end{pmatrix}\right\|_2^2}\right)\nonumber\\
&\qquad=n\frac{1}{2}\log\left(P\abs{\alpha-\alpha^*}^{2\nu}\right)+n f_2(P,\alpha),
\end{align}
where $f_2$ is such that $\lim_{\alpha\rightarrow\alpha_0}\lim_{P\rightarrow\infty}f_2(P,\alpha)$ exists and is finite.
The last equality follows because for every non-zero $\alpha_0$, the limit $\lim_{\alpha \to \alpha_0}\left\|\begin{pmatrix}v_0&\cdots&v_{L}\end{pmatrix}\right\|_2^2$ exists, is finite, and larger than 0, and because by definition $\alpha^*$ is a root of the polynomial $v_{L+1}^2(\alpha)$ with multiplicity $2\nu$. 

Taking $c_0(\alpha)=\lim_{P\rightarrow\infty}\left(f_1(P,\alpha)+f_2(P,\alpha)\right)$ concludes the proof.

\bibliographystyle{ieeetr}

\end{document}